\documentclass[%
reprint,
superscriptaddress,
nofootinbib,
amsmath,amssymb,
aps,longbibliography
]{revtex4-1}

\usepackage{amsthm,bm,xparse,xcolor,tikz}
\usepgfmodule{decorations}
\usetikzlibrary{decorations,decorations.pathmorphing,patterns}
\usepackage[T1]{fontenc}
\usepackage{makecell}
\usepackage{soul}

\AtBeginDocument{%
    \newwrite\bibnotes
    \def\bibnotesext{Notes.bib}
    \immediate\openout\bibnotes=\jobname\bibnotesext
    \immediate\write\bibnotes{@CONTROL{REVTEX41Control}}
    \immediate\write\bibnotes{@CONTROL{%
    apsrev41Control,author="08",editor="1",pages="1",title="0",year="0"}}
     \if@filesw
     \immediate\write\@auxout{\string\citation{apsrev41Control}}%
    \fi
  }%

\usepackage[nodayofweek]{datetime}  
\usepackage{hyperref}
\usepackage{xcolor}
\usepackage{soul}
\usepackage[capitalise]{cleveref}
\usepackage{enumerate}   
\definecolor{mylinkcolor}{rgb}{0,0,0.8} 
\hypersetup{unicode=true, %
  bookmarksnumbered=false,bookmarksopen=false,bookmarksopenlevel=1, %
  breaklinks=true,pdfborder={0 0 0},colorlinks=true}%
\hypersetup{%
  anchorcolor=mylinkcolor,citecolor=mylinkcolor, %
  filecolor=mylinkcolor,linkcolor=mylinkcolor, %
  menucolor=mylinkcolor,runcolor=mylinkcolor, %
  urlcolor=mylinkcolor}%

\newtheorem{theorem}{Theorem}
\newtheorem{lemma}{Lemma}
\newtheorem{corollary}{Corollary}
\newtheorem{claim}{Claim}

\theoremstyle{definition}
 
\newtheorem{definition}{Definition}
\newtheorem{remark}{Remark}
\newtheorem{example}{Example}

\newcommand{\ket}[1]{| #1 \rangle}
\newcommand{\bra}[1]{\langle #1 |}
\newcommand{\braket}[2]{\langle #1|#2\rangle}
\newcommand{\ketbra}[2]{|#1\rangle\!\langle#2|}
\newcommand{\id}{\mathbb{I}}
\newcommand{\tr}{{\mathrm {tr}}}
\newcommand{\ot}{\otimes}
\newcommand{\proj}[1]{\ket{#1}\!\bra{#1}}

\newcommand{\score}{\omega}

\newcommand{\mcB}{\mathcal{B}}
\newcommand{\mcC}{\mathcal{C}}
\newcommand{\mcD}{\mathcal{D}}

\newcommand{\mcF}{\mathcal{F}}
\newcommand{\mcG}{\mathcal{G}}
\newcommand{\mcH}{\mathcal{H}}
\newcommand{\mcI}{\mathcal{I}}

\newcommand{\mcM}{\mathcal{N}}
\newcommand{\mcN}{\mathcal{N}}

\newcommand{\mcP}{\mathcal{P}}

\newcommand{\mcS}{\mathcal{S}}

\newcommand{\mcW}{\mathcal{W}}

\newcommand{\mbC}{\mathbb{C}}

\newcommand{\mbN}{\mathbb{N}}

\newcommand{\mbR}{\mathbb{R}}

\usepackage[most]{tcolorbox}

\newcommand{\chg}[1]{{\color{black}#1}}

\usepackage{circuitikz}
\usepackage{float}

\usepackage[normalem]{ulem}
\usepackage{xcolor}

\newcommand\colouredsout[2]{%
  \renewcommand\sout{\bgroup\markoverwith
  {\textcolor{#1}{\rule[0.5ex]{2pt}{0.8pt}}}\ULon}%
  \sout{#2}%
}

\newtheorem{protocol}{Protocol}
\usepackage{mdframed}
\def\M{\ensuremath\mathcal}
\def\B{\ensuremath\mathbf}

\usepackage[normalem]{ulem}

\usepackage{subcaption}

\newif\ifarxiv 
\arxivtrue

\begin{document}

\title{Composable framework for device-independent state certification }

\author{Rutvij Bhavsar\textsuperscript{(a)\,(b) }}
\email{rutvij.bhavsar@kcl.ac.uk}
\affiliation{School of Electrical Engineering, Korea Advanced Institute of Science and Technology (KAIST), 291 Daehak-ro, Yuseong-gu, Daejeon 34141, Republic of Korea}

\author{Lewis Wooltorton\textsuperscript{(a)}}
\email{lewis.wooltorton@ens-lyon.fr}
\affiliation{Department of Mathematics, University of York, Heslington, York, YO10 5DD, United Kingdom}
\affiliation{Quantum Engineering Centre for Doctoral Training, H. H. Wills Physics Laboratory and Department of Electrical \& Electronic Engineering, University of Bristol, Bristol BS8 1FD, United Kingdom}
\affiliation{Inria, ENS de Lyon, LIP, 46 Allee d’Italie, 69364 Lyon Cedex 07, France}

\author{Joonwoo Bae}
\email{ joonwoo.bae@kaist.ac.kr }
\affiliation{School of Electrical Engineering, Korea Advanced Institute of Science and Technology (KAIST), 291 Daehak-ro, Yuseong-gu, Daejeon 34141, Republic of Korea}

\date{12$^{\text{th}}$ December, 2025}

\begin{abstract}
Certifying a quantum state in a device-independent (DI) manner, in which no trust is placed on the internal workings of any physical components, is a fundamental task bearing various applications in quantum information theory. The composability of a state certification protocol is key to its integration as a subroutine within information-theoretic protocols. In this work, we present a composable certification of quantum states in a DI manner under the assumption that a source prepares a finite sequence of independent quantum states that are not necessarily identical. We show that the security relies on the DI analog of the fidelity, called the {\it extractability}. We develop methods to compute this quantity under local operations and classical communication in certain Bell scenarios that self-test the singlet state, which may also be of independent interest. Finally, we demonstrate our framework by certifying the singlet state in a composable and DI manner using the Clauser–Horne–Shimony–Holt inequality.  
\end{abstract}

\maketitle

\renewcommand{\thefootnote}{ }
\footnotetext{(a) These authors contributed equally to this work. \\
(b) Current affiliation: Department of Mathematics, King’s College London, Strand, London, WC2R 2LS, United Kingdom}

\def\thefootnote{\arabic{footnote}}
\setcounter{footnote}{0}

\section{Introduction} Certifying the non-classical properties of a source is essential for its use in quantum information processing. This is often achieved by modeling the physical mechanisms which govern either the source's behavior directly, or the measurements used to characterize it. However, deviations between a realistic implementation and this model are difficult to rule out, and any assumptions which are not met in practice could render the certification invalid. The device-independent (DI) approach~\cite{Ekert,ABGMPS,ColbeckThesis,PAMBMMOHLMM,BarrettNonlocalResource} addresses this issue. Here, certification is obtained through the observation of nonlocal input-output correlations, removing the need to understand the source's inner workings or perform trusted measurements.

The aim of DI state certification (DISC) is to certify the presence of a certain multi-partite ``target'' state in a state that is stored in the memory of the user(s). Recently, there have been significant advances in DISC~\cite{bancal2021self,govcanin2022sample}, in which the states and measurements are treated as (separate) black-boxes, and tools from robust self-testing~\cite{MayersYao,mayers2004self,kaniewski2016analytic} enable certification when the source is not assumed to be independent and identically distributed (i.i.d.), and the collected statistics are finite. A related protocol was also outlined in~\cite{AFB19,philip25}, where DI lower bounds on the rate of distillable entanglement were derived for non-i.i.d. sources. Such protocols have also been successfully demonstrated experimentally~\cite{martins24,Storz24,Schmid24}.

However, when using any protocol as a subroutine in a larger system, a notion of composability is essential. For example, in cryptography~\cite{Ekert,BHK,ABGMPS,PABGMS,bcktwo,VV2,ADFRV,ColbeckThesis,PAMBMMOHLMM,CK2,MS1,MS2,CR_free}, composable security statements ensure that the concatenation of two secure protocols results in another secure protocol~\cite{Renner,Portmann14}. A composable approach to DISC is therefore crucial if the stored state is intended for many practical purposes.

We here address the need for the framework of composable certification and present composable guarantees for DISC protocols, enabling their integration as subroutines in broader information theoretic tasks, such as cryptography. Furthermore, our protocols only assume that finitely many independent states are emitted by the source, from which a single state is randomly selected and preserved in a memory for future applications. The remaining states then interact with measurement devices that can have a memory in general, and the statistics are used to certify the preserved state. 

We prove the security of our protocol using the \emph{extractability}, a DI variant of the fidelity between each measured state and the target~\cite{kaniewski2016analytic,SupicSelfTest}. A bound on the extractability follows from observing a Bell inequality violation, and we provide a generic way to compute this quantity in Bell scenarios with binary inputs and binary outputs\footnote{We refer to this scenario with two parties as the minimal Bell scenario.}, which may be of independent interest in the context of nonlocality and entanglement theory. Specifically, the standard definition of extractability quantifies how robust a self-test is, by measuring the minimal distance between a ``physical'' state which produces a given Bell violation and the target state. This distance is typically computed after applying local operations to the physical state, accounting for degrees of freedom undetectable from the statistics alone~\cite{Coopmans19}. However, this notion can be restrictive in the context of state certification protocols, where the relevant class of \textit{free operations} may depend on both the intended application and the experimental setup. For example, in cryptographic settings, it is natural to consider local operations and classical communication (LOCC) as the allowed free operations. Our security framework can be applied to any chosen class of operations, and we propose a method to compute bounds on the extractability under LOCC for any Bell inequality in the minimal Bell scenario certifying the singlet state. As LOCC extractability is typically greater than that obtained under local operations alone, our protocols benefit from tighter security bounds than those derived using the standard approach. 

The paper is structured as follows. In \cref{sec:prior_work,sec:self-test} we provide the necessary background. In \cref{sec: measSetup} we introduce the measurement scenarios considered. \cref{sec:introPro} then outlines the corresponding DISC protocols. A composable security framework is then presented in \cref{sec:compSec}, and in \cref{sec:proofMain} we bound the security of the introduced protocols according to this definition. \cref{sec: extractability} discusses the main tool needed to provide security, namely, how to bound the extractability. We then provide an explicit example using the Clauser–Horne–Shimony–Holt (CHSH) inequality in \cref{sec:example}, and conclude with a discussion in \cref{sec:disc}. \ifarxiv All protocols and assumptions are detailed in \cref{app:assumptions,app:allProtocols}, and the proofs of all security claims can be found in \cref{app: security_proofs_par,app: security_proofs_sec,app:seqLem}. A detailed derivation of the extractability bounds can be found in \cref{app: LOCC_extractibiltity}.\else Technical details of the protocols and assumptions are provided in the Supplementary Material~\cite[Sections 1 and 2]{supp}, along with the full proofs of all security claims~\cite[Sections 3 to 5]{supp}. A detailed derivation of the extractability bounds appears in the Supplementary Material~\cite[Section 6]{supp}.\fi

\section{Prior work}
\label{sec:prior_work}

\chg{Certifying properties of entangled states in a device-independent fashion has been extensively studied. Often considered is the amount of private randomness contained in the measurement outcomes, as a function of the observed Bell violation. This can be estimated in the asymptotic regime~\cite{ABGMPS,TanDI,Bhavsar2023,Brown2024deviceindependent}, and lifted to the finite size setting using tools such as the entropy accumulation theorem~\cite{DFR,EAT2,MetgerGEAT} and quantum probability estimation framework~\cite{Zhang20}, enabling composably secure DI randomness expansion~\cite{ColbeckThesis,PAMBMMOHLMM,CK2,MS1,MS2}, amplification~\cite{CR_free,Foreman2023practicalrandomness} and key distribution~\cite{Ekert,BHK,ABGMPS,PABGMS,VV2,ADFRV}. 

A stronger characterization of the source is possible by estimating properties of the emitted states directly, rather than the measurements performed on them. For example, various entanglement measures can be quantified as a function of the observed Bell violation in the asymptotic regime~\cite{AFB19,zhu2024interplay}. In the finite size regime, Refs.~\cite{AFB19,philip25} certify the largest number of maximally entangled states that can be extracted from an untrusted source under LOCC operations, a quantity known as the one-shot distillable entanglement. This protocol outputs a subset of states (intended for a future task) which the entanglement certification applies to, and does not require an i.i.d. assumption on the source.   

DISC achieves the strongest type of characterization, namely, that the states emitted by the source are close to (many copies of) a target state. Bancal \textit{et al.}~\cite{bancal2021self} provided a DI state verification protocol (in which all states are measured; see~\cite{govcanin2022sample} for a detailed comparison between verification and certification) using the CHSH inequality. Notably, this protocol does not require an i.i.d. assumption on the source and was used to verify the successful distribution of entangled states via a quantum network link. Gočanin \textit{et al.}~\cite{govcanin2022sample} later developed a general framework for DI state verification based on any robust self-testing result. This approach achieves an optimal sample efficiency in the fully non-i.i.d. setting. The authors of Ref.~\cite{govcanin2022sample} further propose a framework for DI state certification, which outputs multiple certified copies with an optimal sample efficiency under the assumption of an independently distributed source. 

In the framework of Refs.~\cite{bancal2021self,govcanin2022sample}, the  average extractability of the state ensemble is certified up to a fixed confidence level. This contrasts the definition of composable security adopted in cryptography, such as quantum key distribution (QKD), in which the distinguishability between an idealized version of the protocol and its real implementation is shown to be small~\cite{Renner,Portmann14}. Applying such a framework to DISC is one of the new contributions of our work. Inspired by~\cite{govcanin2022sample}, we present a general framework for DISC based on extractability under the assumption of independent states. The key difference is that we introduce and apply a composable security definition inspired by quantum cryptography. We do however loose the sample efficiency achieved in~\cite{govcanin2022sample}, a trade off which is necessary in the case of exact certification (see \cref{sec:introPro} for a discussion of different types of certification) owing to the inherently stronger demand of composability~\cite{wiesner2024}. As a proof of principle demonstration, we also consider certifying a single copy of the target state. Extending our framework to the certification of multiple copies in a sample-efficient and fully non-i.i.d. manner is a promising future direction, which we elaborate on in \cref{sec:disc}. Additionally, while composability is not directly addressed in Refs.~\cite{AFB19,philip25}, we expect an analysis similar to the one presented here would directly apply to one-shot distillable entanglement certification. However as aforementioned, certifying the state directly as considered in this work and~\cite{bancal2021self,govcanin2022sample} is a stronger demand than certifying its distillable entanglement.} 

\section{Nonlocality and self-testing}
\label{sec:self-test}
We now review self-testing, which can be viewed as a form of DI state verification in the asymptotic limit. Consider a bipartite Bell scenario, in which two non-communicating devices perform local measurements on some unknown shared quantum state. The inputs to each device are uniformly distributed binary random variables, denoted $X$ and $Y$, respectively. The outputs of each device are binary random variables, denoted $A$ and $B$, which for a given pair of inputs $X=x,\, Y=y$ are distributed according to the device's behavior, $\mathbf{p} = \{p(a,b|x,y)\}$. A behavior is quantum if it can be realized by local quantum measurements on a bipartite state $\rho \in \mathcal{S}( \mathcal{H}_{Q_{A}} \otimes \mathcal{H}_{Q_{B}})$, where $\mathcal{S}(\mathcal{H})$ denotes the set of normalized states on a Hilbert space $\mathcal{H}$. Specifically, the probabilities are given by the Born rule, $p(a,b|x,y) = \tr\big[ (M_{a|x} \otimes N_{b|y})\rho]$, where $\{ \{M_{a|x}\}_{a \in \{0,1\}} \}_{x \in \{0,1\}}$ and $\{ \{N_{b|y}\}_{b \in \{0,1\}} \}_{y \in \{0,1\}}$ are POVMs on $\mathcal{H}_{Q_{A}}$ and $\mathcal{H}_{Q_{B}}$, respectively. 

In certain cases, a given quantum behavior $\mathbf{p}^{*} = \{p^{*}(a,b|x,y)\}$ guarantees the presence of a particular state $\ket{\psi^*}$ up to operations which cannot be detected from the behavior, called local isometries. If this is the case, we say the behavior $\mathbf{p}^{*}$ \emph{self-tests} $\ket{\psi^*}$~\cite{MayersYao,SupicSelfTest}. Such statements are the strongest form of verification possible. However, in noisy systems it is often the case that the behavior $\mathbf{p}^{*}$ cannot be achieved exactly. We say $\mathbf{p}^{*}$ \emph{robustly} self-tests $\ket{\psi^*}$ if witnessing a behavior $\mathbf{p}'$ $\epsilon$-close to $\mathbf{p}$ guarantees the presence of a state $\rho'$ $\epsilon$-close to $\psi^* = \ketbra{\psi^*}{\psi^*}$, according to a given metric (see, e.g., Ref.~\cite{SupicSelfTest} for a definition). Robust self-testing can thus be viewed as a form of DI state verification when given access to infinitely many copies of an identical (potentially noisy) state. 

Self-testing statements can also be formulated without relying on the full input-output distributions. One approach is to express them in terms of functionals of the behavior, which may be non-linear in general \cite{rai2022self}. However, the most common and often more practical method is to use linear functionals, which correspond to Bell inequalities \cite{Bell,CHSH}. 

In the minimal Bell scenario, it will be convenient to work with the expected values $\langle A_{x} \rangle = \sum_{a}(-1)^{a}p(a|x)$, $\langle B_{y} \rangle = \sum_{b}(-1)^{b}p(b|y)$ and $\langle A_{x}B_{y} \rangle = \sum_{a,b}(-1)^{a+b}p(a,b|x,y)$. When the behavior is quantum, we define the associated observables $A_{x} = \sum_{a} (-1)^{a}M_{a|x}$ and $B_{y} = \sum_{b} (-1)^{b}N_{b|y}$, which satisfy $\langle A_{x}B_{y} \rangle = \tr[(A_{x} \otimes B_{y})\rho]$, $\langle A_{x} \rangle =   \tr[(A_{x} \otimes \id_{Q_{B}})\rho]$ and $\langle B_{y} \rangle = \tr[(\id_{Q_{A}} \otimes B_{y})\rho]$. A Bell expression is then a linear combination of the probabilities $\{p(a,b|x,y)\}$, or equivalently the expectations $\{\langle A_{x}\rangle, \langle B_{y} \rangle, \langle A_{x}B_{y} \rangle\}$, and for quantum behaviors we consider Bell operators of the form \chg{\cite{Brunner_review}}
\begin{equation}
    B = \sum_{x,y \in \{-1,0,1\}} c_{xy} \, A_{x} \otimes B_{y}, \label{eq:bop}
\end{equation}
for some real coefficients $c_{xy}$, where we defined $A_{-1} := \id_{Q_{A}}$ and $B_{-1} := \id_{Q_{B}}$. The Bell value is given by $\langle B \rangle = \tr[B\rho]$ for a state $\rho$, and the maximum quantum and classical values of $\langle B \rangle$ are denoted by $\eta^{\mathrm{Q}}$ and $\eta^{\mathrm{L}}$, respectively. A Bell inequality $\langle B \rangle \leq \eta^{\text{L}}$ self-tests a state $\ket{\psi^*}$ if every quantum state which achieves $\langle B \rangle = \eta^{\text{Q}}$ is equivalent to $\ket{\psi^*}$ up to local isometries. 

Of particular interest to this work are Bell inequalities which self-test the maximally entangled pair of qubits, $\phi^{+} = \ketbra{\phi^+}{\phi^+}$ where $\ket{\phi^+} = (\ket{00} + \ket{11})/\sqrt{2}$. Such inequalities have been fully characterized in the minimal Bell scenario~\cite{wang2016all,Le2023quantumcorrelations,Barizien2024custombell,WBC3}. An important example is the CHSH inequality~\cite{CHSH}, which is the only non-trivial facet Bell inequality in this scenario up to relabelings. Its Bell operator takes the form $B_{\text{CHSH}} := A_{0}\otimes (B_{0} + B_{1}) + A_{1} \otimes (B_{0} - B_{1})$, and has local and quantum bounds of $2$ and $2\sqrt{2}$, respectively. 

\section{Measurement scenarios}
\label{sec: measSetup}

Real sources emit a finite sequence of non-i.i.d. states. Furthermore, a real measurement device may also be non-i.i.d., e.g., through the use of a memory. It is then clear that, on their own, self-testing statements are not enough for state verification or certification in practice. To address this, we consider certifying a source which emits $n$ independent but not necessarily identical states, $\rho = \bigotimes_{i=1}^{n}\rho_{i}$, where $\rho_{i} \in \mcS(\mcH_{Q_{i}^{A}} \otimes \mcH_{Q_{i}^{B}})$. We also assume the user has access to a quantum memory in which states can be stored, and treat all measurement devices as black-boxes, making no assumptions on their internal workings. 

\begin{figure}[H]
\centering
\resizebox{0.25\textwidth}{!}{%
\begin{circuitikz}
\tikzstyle{every node}=[font = \normalsize] 
\draw  (5.25,13.5) rectangle (7.25,11.5);
\draw [->, >=Stealth] (3.75,13) -- (5.25,13);
\draw [->, >=Stealth] (7.25,13) -- (8.75,13);
\draw [->, >=Stealth] (3.75,12) -- (5.25,12);
\draw [->, >=Stealth, dashed] (5.75,11.5) -- (5.75,10);
\draw [->, >=Stealth, dashed] (6.75,11.5) -- (6.75,10);
\node   at (3.5,13.25) {$I_{i}^{A}$};
\node   at (3.5,12.25) {$Q^{A}_{i}$};
\node   at (5.25,10.25) {$X_{i}$};
\node   at (6.25,10.25) {$A_{i}$};
\node   at (9,13.25) {$O_{i}^{A}$};
\node   at (6.25,12.45) {$\mathcal{N}_{i}^{A}$};
\end{circuitikz}
}%

\caption{A figure illustrating the measurement channels on Alice's side. The channels on Bob's side have an identical structure. Solid lines indicate quantum systems, while dashed lines represent classical systems. $I_{i}^{A}$ and $O_{i}^{A}$ represent auxiliary information available to, and generated by, the measurement of $Q_{i}^{A}$, respectively. The generation of Alice's input $X_{i}$ is absorbed as part of the channel output, whilst $A_{i}$ represents the output of a quantum measurement. }
\label{fig:measBox}
\end{figure}

Associated to each index $i$ is a channel, $\mcN_{i}^{A}:I_{i}^{A}Q_{i}^{A} \to A_{i}X_{i}O_{i}^{A}$, shown in \cref{fig:measBox}. The systems $A_{i}X_{i}$ are classical, whilst $I_{i}^{A}$ contains any auxiliary information available to the device prior to measurement (e.g., information from previous rounds), and $O_{i}^{A}$ contains auxiliary output information (e.g., information to pass on to future rounds). \chg{The channel $\mathcal{N}_{i}^{A}$ first samples the input $X_{i}$ and then performs the corresponding measurement on the (quantum) system held by the measurement device. Thus, while $X_{i}$ serves as an input to the measurement device, from this perspective it is naturally regarded as an output of the channel $\mcN_{i}^{A}$.} We similarly define channels $\mcN_{i}^{B}:I_{i}^{B}Q_{i}^{B} \to B_{i}Y_{i}O_{i}^{B}$. A joint channel $\mcN_{i}$ is described by performing $\mcN_{i}^{A}$ and $\mcN_{i}^{B}$ in a space-like separated manner, i.e., $\mcN_{i} = \mcN_{i}^{A} \otimes \mcN_{i}^{B}$. Then, the systems $X_{i}Y_{i}$ and $A_{i}B_{i}$ denote the inputs and outputs to the Bell test, respectively.  

We then consider two possible arrangements of the channels, corresponding to parallel and sequential setups. In the parallel setup, $n-1$ isolated measurement devices $M_{i}$ each receive $\rho_{i}$ and implement the channel $\mcN_{i}$. Following the isolation of the devices, the auxiliary systems $I_{i}^{A}O_{i}^{A}$ and $I_{i}^{B}O_{i}^{B}$ can be omitted, and all measurements occur simultaneously. This can be equivalently viewed as the action of a single memoryless measurement device acting on each system sequentially. We illustrate this scenario in \cref{fig:meas_par}. 

\begin{figure}[!ht]
\resizebox{0.3\textwidth}{!}{%
\begin{circuitikz}
\tikzstyle{every node}=[font=\LARGE]
\draw  (3.75,18.25) rectangle (6.25,15.75);
\draw [->, >=Stealth] (2.5,17) -- (3.75,17);
\node [font=\LARGE] at (5,17) {$\mcN_{1}$};
\node [font=\LARGE] at (9,17) {$A_{1}X_{1}B_{1}Y_{1}$};
\node [font=\LARGE] at (1.5,17) {$Q_{1}$};
\draw [dashed] (1.25,15.25) -- (8.75,15.25);
\draw [->, >=Stealth, dashed] (6.25,17) -- (7.5,17);
\draw  (3.75,14.75) rectangle (6.25,12.25);
\draw [->, >=Stealth] (2.5,13.5) -- (3.75,13.5);
\node [font=\LARGE] at (5,13.5) {$\mcN_{2}$};
\node [font=\LARGE] at (9,13.5) {$A_{2}X_{2}B_{2}Y_{2}$};
\node [font=\LARGE] at (1.5,13.5) {$Q_{2}$};
\draw [dashed] (1.25,11.75) -- (8.75,11.75);
\draw [->, >=Stealth, dashed] (6.25,13.5) -- (7.5,13.5);
\draw  (3.75,9.5) rectangle (6.25,7);
\draw [->, >=Stealth] (2.5,8.25) -- (3.75,8.25);
\node [font=\LARGE] at (5,8.25) {$\mcN_{n}$};
\node [font=\LARGE] at (9,8.25) {$A_{n}X_{n}B_{n}Y_{n}$};
\node [font=\LARGE] at (1.5,8.25) {$Q_{n}$};
\draw [->, >=Stealth, dashed] (6.25,8.25) -- (7.5,8.25);
\draw [dashed] (1.25,10) -- (8.75,10);
\node [font=\LARGE] at (5,11) {$\vdots$};
\end{circuitikz}
}%
\caption{\chg{The parallel setup, in which $n-1$ devices each perform an isolated measurement on a composite quantum system $Q_{i} = Q_{i}^{A}Q_{i}^{B}$. Each measurement is performed in an isolated manner with respect to the $AB$ partition.} }
\label{fig:meas_par}
\end{figure}

In the sequential scenario, we consider a single measurement device $M$, which performs all channels in sequence. Specifically, the measurement of $\rho_{i}$ precedes that of $\rho_{i+1}$, and the auxiliary output $O_{i}^{A}$ is sent to the input $I_{i+1}^{A}$ (and similarly for $B$). In other words, we associate $I_{i+1}^{A} = O_{i}^{A}$ and $I_{i+1}^{B} = O_{i}^{B}$. Information in $O_{i}^{A}$ could, for example, consist of the inputs and outputs of rounds $1$ to $i$. Note that systems pertaining to $A$ and $B$ are still assumed to be separated, e.g., $O_{i}^{A}$ cannot influence $I_{i+1}^{B}$. See \cref{fig:meas_seq} for an illustration.   

\begin{figure}[!ht]
\centering
\resizebox{0.3\textwidth}{!}{%
\begin{circuitikz}
\tikzstyle{every node}=[font=\LARGE]
\draw  (3.75,18.25) rectangle (6.25,15.75);
\draw [->, >=Stealth] (2.5,17) -- (3.75,17);
\node [font=\LARGE] at (5,17) {$\mcN_{1}$};
\node [font=\LARGE] at (9,17) {$A_{1}X_{1}B_{1}Y_{1}$};
\node [font=\LARGE] at (1.5,17) {$Q_{1}$};
\draw [->, >=Stealth, dashed] (6.25,17) -- (7.5,17);
\draw  (3.75,14.75) rectangle (6.25,12.25);
\draw [->, >=Stealth] (2.5,13.5) -- (3.75,13.5);
\node [font=\LARGE] at (5,13.5) {$\mcN_{2}$};
\node [font=\LARGE] at (9,13.5) {$A_{2}X_{2}B_{2}Y_{2}$};
\node [font=\LARGE] at (1.5,13.5) {$Q_{2}$};
\draw [->, >=Stealth, dashed] (6.25,13.5) -- (7.5,13.5);
\draw  (3.75,9.5) rectangle (6.25,7);
\draw [->, >=Stealth] (2.5,8.25) -- (3.75,8.25);
\node [font=\LARGE] at (5,8.25) {$\mcN_{n}$};
\node [font=\LARGE] at (9,8.25) {$A_{n}X_{n}B_{n}Y_{n}$};
\node [font=\LARGE] at (1.5,8.25) {$Q_{n}$};
\draw [->, >=Stealth, dashed] (6.25,8.25) -- (7.5,8.25);
\node [font=\LARGE] at (5,11) {$\vdots$};
\draw [->, >=Stealth] (5,19) -- (5,18.25);
\draw [->, >=Stealth] (5,15.75) -- (5,14.75);
\draw [->, >=Stealth] (5,12.25) -- (5,11.5);
\draw [->, >=Stealth] (5,10.25) -- (5,9.5);
\draw [->, >=Stealth] (5,7) -- (5,6.25);
\node [font=\LARGE] at (5,19.5) {$I_{1}$};
\node [font=\LARGE] at (7.25,15.25) {$I_{2} = O_{1}$};
\node [font=\LARGE] at (7.25,11.75) {$I_{3} = O_{2}$};
\node [font=\LARGE] at (7.75,10) {$I_{n} = O_{n-1}$};
\node [font=\LARGE] at (5,5.75) {$O_{n}$};
\end{circuitikz}
}%
\caption{\chg{The sequential setup, in which a single device performs all channels $\mcN_{i}$ in a sequence on the systems $Q_{i} = Q_{i}^{A}Q_{i}^{B}$. The auxiliary information $O_{i} = O_{i}^{A}O_{i}^{B}$ generated by $\mcN_{i}$ is fed forward as the input register $I_{i+1} = I_{i+1}^{A}I_{i+1}^{B}$ for the subsequent measurement $\mcN_{i+1}$.}}
\label{fig:meas_seq}
\end{figure}

We also introduce a random variable $T$, which takes values in $t \in \{1,...,n\}$ according to a \chg{known} distribution $\text{Pr}[T=t] = p_{T}(t)$, and is sampled prior to all measurements. Specifically, the value of $T$ decides which state is not measured and, instead, held in the quantum memory. \chg{Note that the variables $T$, $X$, and $Y$ are all sampled from distributions known to the users. While these distributions may also be known to the adversary, the adversary cannot access the actual realized values of $T$, $X$, or $Y$. For a detailed list of the assumptions made in this work, see \ifarxiv \cref{app:assumptions}. \else the Supplementary Material~\cite[Section 1]{supp}. \fi}

\section{Protocols}
\label{sec:introPro}

We consider DISC protocols of the following form. First, the state $\bigotimes_{i=1}^{n} \rho_{i}$ is  prepared by an untrusted source, and accepted into the secure laboratory. The \chg{random variable} $T$ is generated, and depending on its value, $T = t$, the state $\rho_{t}$ is stored in a quantum memory. Next, the remaining states are sent to the measurement device, and either the parallel or sequential measurement setup is performed. The statistics are collected, and if they deviate from an expected value, the protocol aborts. \ifarxiv Specifically, Protocols \ref{prot: DISV} to \ref{prot: DISV_general_var} consist of a parallel setup, while Protocols \ref{prot: DISV_general_sequencial} and \ref{prot: DISV_general_sequencial_var} consist of a sequential setup.\else Specifically, Protocols 1 to 3 consist of a parallel setup, while Protocols 4 and 5 consist of a sequential setup (see the Supplementary Material~\cite[Section 2]{supp} for details). \fi 

In the event of not aborting, there are two possible variants of the protocol. In the first, the user applies a pre-decided channel (from the set of free operations) to the state stored in memory, with the aim of ``extracting'' a target state \ifarxiv (Protocols \ref{prot: DISV}, \ref{prot: DISV_general} and \ref{prot: DISV_general_sequencial}). \else (Protocols 1, 2 and 4 in the Supplementary Material~\cite[Section 2]{supp}). \fi However, performing this channel in practice may be unrealistic. Thus, in another variant, the user certifies the state held in memory to be equivalent to (i.e., has the potential to be converted to) the target up to the set of free operations \ifarxiv (Protocols \ref{prot: DISV_general_var} and \ref{prot: DISV_general_sequencial_var}).\else (Protocols 3 and 5 in the Supplementary Material~\cite[Section 2]{supp}). \fi

We also allow for freedom in the choice of target state. In \ifarxiv Protocol \ref{prot: DISV}\else Protocol 1\fi, we consider extracting the target state $\psi^*$ (e.g., the maximally entangled state) exactly. Due to imperfections such as noise however, this is out of reach for many practical applications, resulting in an overly stringent security criteria. Instead, the user may wish to relax this, and extract a state $\varepsilon$-close to $\psi^*$\footnote{The notion of closeness referred to here, and throughout the paper, is with respect to the trace distance, unless stated otherwise. Specifically, a state $\rho$ is said to be $\varepsilon$-close to $\sigma$ if $\frac{1}{2}|| \rho - \sigma ||_1 \leq \varepsilon$, where $||A ||_1 := \tr[\sqrt{A^{\dagger}A}]$ denotes the trace norm.}, or guarantee that the final state can be converted to any state $\varepsilon$-close to $\psi^*$ under free operations. We allow for this modification in \ifarxiv Protocols \ref{prot: DISV_general} to \ref{prot: DISV_general_sequencial_var}, \else Protocols 2 to 5 in the Supplementary Material~\cite[Section 2]{supp}, \fi where the distance of interest is the trace distance.

Additionally, we note that while the protocols outlined thus far are concerned with storing and certifying a single state, the proof techniques can be generalized to multiple states. However, the current bounds scale poorly with the number of states measured, and we leave improvements in sample efficiency, such as that presented in Ref.~\cite{govcanin2022sample}, to future work (see \cref{sec:disc} for further discussion).

We provide a template for all protocols considered in this work below, followed by a summary of each variant in \cref{tab:protocol_summary}. All variants are outlined in detail in \ifarxiv \cref{app:allProtocols}.\else the Supplementary Material~\cite[Section 2]{supp}. \fi

\chg{\begin{example}[CHSH-based DISC for the sequencial setup]\label{example: CHSH} 
Consider a source that produces $n$ bipartite states $\{ \rho_{i}\}_{i = 1}^{n}$. 
The goal is to certify that one of these states is $\varepsilon$-close to the maximally entangled pair of qubits $\proj{\phi^+}$, 
in the sense that it can be converted by LOCC operations to some state $\tau$ that is $\varepsilon$-close to $\proj{\phi^+}$. 
The protocol begins by choosing an index $t \in \{1, \dots, n\}$ at random and storing the corresponding state $\rho_t$ in a memory. Each of the remaining $n-1$ states are sent one by one to the same pair of non-communicating measurement devices\footnote{%
The assumption that the devices are non-communicating can be justified in two standard ways: 
(i) by employing shielding mechanisms that prevent any exchange of information, or 
(ii) by enforcing space-like separation during the measurement rounds, which guarantees 
that the devices cannot signal to one another.}. 

On every round $i \neq t$, each device receives a binary input $x_i$ (resp. $y_i$) chosen at random, and performs an unknown measurement on their share of the state $\rho_i$, producing a binary output $a_i$ (resp. $b_i$). A CHSH game is won in round $i$ if the outputs satisfy $a_i \oplus b_i = x_i \cdot y_i$, and is lost otherwise. From the $n-1$ measurement rounds, the empirical CHSH value is then computed, i.e., a value proportional to the number of rounds in which the CHSH game was won divided by $n-1$.  

Depending on the desired security requirement---namely the soundness and completeness parameters introduced in Section~\ref{sec:compSec}---an abort threshold for the CHSH value is chosen. If the empirical value does not exceed this value, the protocol aborts. Otherwise, the protocol accepts, and the stored state $\rho_t$ is certified to be $\varepsilon$-close to $\ketbra{\phi^+}{\phi^+}$ under LOCC operations and up to the chosen security requirement. Intuitively, a smaller $\varepsilon$ corresponds to a stricter acceptance condition: for example, exact certification of the target state requires the observed CHSH value to be very close to the maximal quantum value $2\sqrt{2}$. \end{example} }

\begin{tcolorbox}[colback=blue!5!white, colframe=blue!75!black, title= Template protocol]
Parameters: \\
$n \in \mathbb{N}^+$ -- number of rounds. \\ 
$p_{T}$ -- probability distribution of $T$. \\ 
$\omega_{\sharp}$ -- expected value of the Bell functional.  \\ 
$\kappa > 0$ -- parameter to set completeness error (see \cref{def: completeness}). \\ 
$\varepsilon \geq 0$ -- closeness parameter (See Table \ref{tab:protocol_summary}).  \\ 

\begin{enumerate}
    \item Generate a random \chg{number} $T$ according to $p_{T}$. If $T= t$, then store the state $\rho_{t}$ in the memory. 
    \item Generate the input strings $\B{X} = (X_{1},...,X_{n})$ and $\B{Y} = (Y_{1},...,Y_{n})$ uniformly. 
    \item Depending upon the setup (see Table \ref{tab:protocol_summary} and Section \ref{sec: measSetup}), perform the Bell tests using the generated input strings.
    \item Estimate the Bell parameter $\score_{\mathrm{exp}}$ using the input-output statistics. If $\score_{\mathrm{exp}} \leq \score_{\sharp} - \kappa$, then abort the protocol. 
    \item (Optional). From $\rho_{t}$, extract a state $\tilde{\rho}$ that is $\varepsilon$-close to the target state $\psi^*$.  
\end{enumerate}

\end{tcolorbox}

\begin{table}[h!]
\centering
\begin{tabular}{|l|l|l|l|}
\hline
Protocol & Setup      & Target & Extraction channel                                                                 \\ \hline
1              & Parallel   & $\psi^*$ & Yes                                                                   \\ \hline
2               & Parallel   & $\varepsilon$-close to $\psi^*$ & Yes\\ \hline
3             & Parallel   & $\varepsilon$-close to $\psi^*$ & No \\ \hline
4              & Sequential & $\varepsilon$-close to $\psi^*$ & Yes                    \\ \hline
5              & Sequential & $\varepsilon$-close to $\psi^*$  & No \\ \hline
\end{tabular}
\caption{Summary of all the protocols provided in this paper. ``Setup'' denotes the choice of measurement apparatus described in \cref{sec: measSetup}. ``Target'' indicates whether the final state is assessed to be close to a particular pure entangled state $\psi^*$, or any state $\varepsilon$-close to $\psi^*$ for some $\varepsilon>0$. ``Extraction channel'' indicates whether an extraction channel (within the chosen class of free operations) is applied to the final state or not. In general, such a channel cannot be implemented in a device-independent fashion. Note that the assumptions of the parallel setup can also be satisfied using a single memoryless measurement device.}
\label{tab:protocol_summary}
\end{table}

\section{Composable framework for Device-Independent State Certification}
\label{sec:compSec}
Having outlined the protocol structure, we here introduce a composable security definition for DISC. The security of any protocol must be rigorously defined, and the desired definition depends on its intended application. In cryptographic scenarios, the \textit{composable security definition}~\cite{Ben-OrMayers, BHLMO, Unruh, renner2005universally, Renner, bcktwo, Portmann14} has been adopted as the gold standard. This definition is particularly valuable since it enables a protocol to be securely integrated into a larger, composite system as a subroutine, and typically consists of two error parameters: \(\epsilon_s\), known as the \textit{soundness error} (see \cref{def: soundness}), and \(\epsilon_c\), known as the \textit{completeness error} (see \cref{def: completeness}).

To illustrate the importance of composability in a state certification context, consider one of the DISC protocols described in \cref{sec:introPro}. A possible application is to use the certified state to generate a bit of secret key. This key might subsequently serve as input to another cryptographic protocol. Thus, proving the security of the DISC protocol in isolation is insufficient; we require a security definition that is robust enough to guarantee security when the protocol is used as a building block within a larger system.

In non-cryptographic contexts, security proofs often involve statements such as: \textit{If the protocol does not abort, it outputs the desired state with a small failure probability.} However, such statements are generally not composable. For instance, consider an extreme adversarial strategy where the adversary sends copies of separable states in each round. By sheer luck, this strategy may achieve a large observed Bell value, despite being exceedingly unlikely. In this case, while the protocol mostly aborts, there remains a nonzero probability that it succeeds with an insecure output. If this output is then used in subsequent protocols (such as QKD) security guarantees break down, since conditioned on not aborting, the output state is always separable. \chg{By defining security conditioned on not aborting, the negligibly small probability that this attack succeeds has not been accounted for. In fact, the overall protocol is trivially secure under this attack since it almost always aborts.} 

Another example of insecurity arises from abort-based attacks. An adversary could manipulate the protocol to abort selectively in ways that are advantageous to them. For instance, consider a scenario where the adversary supplies the state 
\[
\rho = \sigma_1 \otimes \psi^* \otimes \psi^* \otimes \cdots \otimes \psi^*,
\]
where \(\sigma_1\) is a separable state and the remaining states are target states (\(\psi^*\)). The adversary could further instruct the devices to perform optimal measurements on \(\psi^*\) whenever a specific setting (e.g., \(t=1\)) is chosen. For all other settings, the adversary forces sub-optimal measurements, ensuring the protocol aborts. In such a scenario, whenever the protocol does not abort, the user is left with a separable state. Consequently, any future protocol relying on the non-aborting behavior of the DISC protocol is rendered insecure.

Adaptive adversarial strategies such as exploiting rare successful outcomes with separable states or leveraging abort-based attacks illustrate the inadequacy of non-composable security definitions. \chg{Specifically, while such attacks can never be ruled out, they can be tolerated using a composable definition as we illustrate below. This is essential to ensure security guarantees remain even when the protocol is integrated into a larger protocol.} 

To prove that a protocol is composably secure, it is necessary to define an \textit{ideal protocol} (see, e.g., Refs.~\cite{Portmann14,pirandola2020advances,primaatmaja2023security}), which aborts with the same probability as the real protocol, and produces a perfect output whenever it does not abort. Importantly, the ideal case is a hypothetical construct that cannot be implemented in practice-- it represents a theoretically perfect protocol. A real protocol is then deemed secure if it is virtually indistinguishable from this ideal. Specifically, security is quantified via a hypothetical game, in which the user performs either the real or ideal protocol. A third-party, referred to as the \textit{distinguisher}, is then tasked with identifying which protocol the user is running, with their success probability measuring security. Since the success probability of distinguishing two quantum states can be quantified using the trace distance, we define the following soundness criteria. 

\begin{definition}[Soundness]\label{def: soundness}
    A DISC protocol is $\epsilon_{s}$-sound if
    \begin{equation}
        \frac{1}{2}\| \rho_{\mathrm{real}} - \rho_{\mathrm{ideal}} \|_{1} \leq \epsilon_{s},
    \end{equation}
    where $\rho_{\mathrm{real}}$ and $\rho_{\mathrm{ideal}}$ are the real and ideal protocol outputs, respectively, and $\| M \|_{1} = \tr \sqrt{M^{\dagger} M}$ is the trace norm of an operator $M$. 
\end{definition}

Note that when computing the distinguishing probability, the distinguisher is assumed to have access to all available side information, along with the output of the real protocol. However, they cannot access any private data generated during the protocol's execution. 

Appropriately defining an ideal protocol can be challenging in general, since one must demonstrate indistinguishability under all possible circumstances. For DISC however, the choice is relatively straightforward, and is inspired by existing definitions in QKD. We consider an ideal DISC protocol which outputs the target state (or an equivalent state up to free operations if the final extraction step is omitted) whenever the protocol does not abort. Furthermore, the ideal protocol aborts with the same probability as the real protocol. Note however this abort probability is not revealed to the user. Similarly, if the goal is to certify a state that is $\varepsilon$-close (in trace distance) to the target state, the ideal protocol outputs any state \(\rho\) which is $\varepsilon$-close to $\psi^*$. \chg{A flow diagram comparing the real and ideal DISC protocol can be found in \cref{fig: flow}, and we provide technical details for all protocol variants in \ifarxiv \cref{app: security_proofs_par,app: security_proofs_sec}. \else the Supplementary Material~\cite[Sections 3 and 4]{supp}. \fi}

\begin{figure}[h!]
    \includegraphics[width=0.5\textwidth]{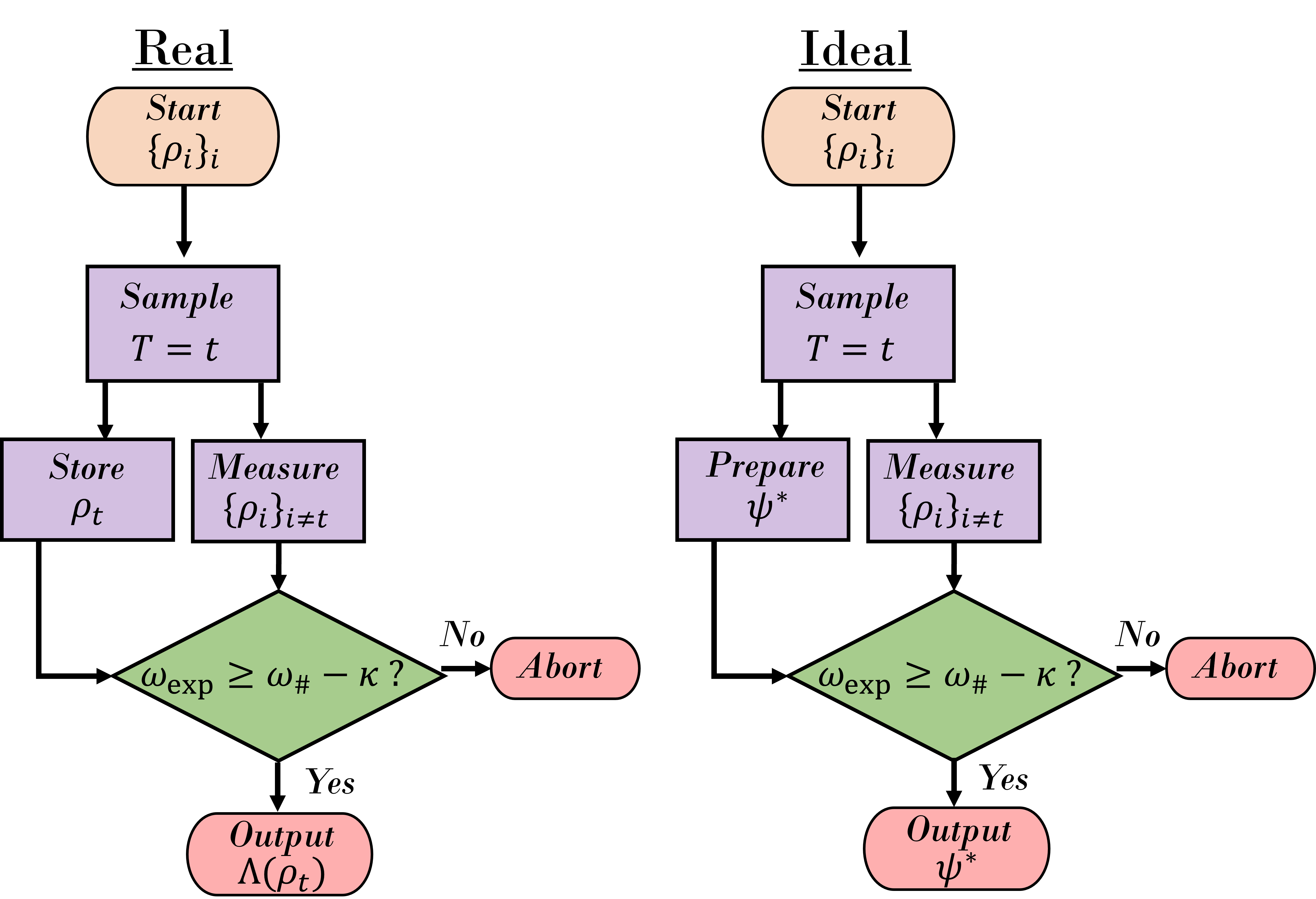}
        \caption{\chg{Flow diagram describing the structure of the real and ideal DISC protocol. Here, $\{\rho_{i}\}_{i=1}^{n}$ denotes a sequence of independent states emitted by the source, $T$ is a uniform random variable which takes values $t \in \{1,...,n\}$, and $\{\rho_{i}\}_{i\neq t}$ is shorthand for $\{\rho_{i} \, : \, i \in \{1,...,n\} \setminus t\}$. The empirical Bell value obtained by measuring $\{\rho_{i}\}_{i\neq t}$ is denoted $\omega_{\mathrm{exp}}$, while $\omega_{\sharp} - \kappa$ is the threshold Bell value below which the protocol aborts. For simplicity, we illustrate the variant in which the user outputs the extracted state $\Lambda(\rho_{t})$.}}
        \label{fig: flow}
\end{figure}

The soundness definition given above ensures that the real protocol is  \chg{virtually} indistinguishable from an ideal one, even under arbitrary adversarial strategies. \chg{This includes the aforementioned extreme cases where the protocol frequently aborts. In such situations, both the real and ideal protocols abort with high probability and are therefore nearly indistinguishable, and hence secure. Specifically, an $\epsilon_s$-sound DISC protocol ensures that an abort based attack can only succeed with probability no larger than $\epsilon_s$.}

However, if soundness is used as the sole criterion for security, a critical issue arises: a protocol that always aborts would trivially be considered secure, despite being of no practical use. To rule out such degenerate cases, we additionally require that there exists an honest implementation of the protocol that succeeds with high probability.  This is captured by the notion of completeness error.

\begin{definition}[Completeness]\label{def: completeness}
A DISC protocol is said to be $\epsilon_c$-complete if there exists an honest implementation that aborts with probability at most $\epsilon_c$.
\end{definition}

\chg{An honest implementation of the protocol is a source and set of measurement devices that model the expected behavior of an experimental implementation. 

\begin{example}[Honest implementation of a CHSH-based DISC protocol]
As a concrete example, recall the CHSH-based protocol discussed in the previous section (Example \ref{example: CHSH}). An example of an honest implementation in this context is a source $\mathsf{S}$ and measurement devices $\mathsf{M}$ with the following description.
The source prepares $n$ identical states $\{\rho_i\}_{i=1}^n$ of the form
\[
  \rho_i \;=\; (1-\mu)\,\proj{\phi^+}\;+\;\mu\,\frac{\id_{4}}{4}\,,
\]
where $\mu \in [0,1]$ (for example, one may take $\mu =   (4/3) \,\varepsilon$ to match a desired robustness parameter $\frac{1}{2}\|\rho_{i} - \phi^+\|_{1} = \varepsilon$) and $\id_{4}$ is the identity operator in dimension $4$.
The measurement devices perform the CHSH projective measurements with observables\footnote{Recall that, for binary projective measurements derived from observables with eigenvalues $\pm 1$, the POVM elements are $M_{a|x}=\tfrac{1}{2}\big(\id+(-1)^a A_x\big)$ and $N_{b|y}=\tfrac{1}{2}\big(\id+(-1)^b B_y\big)$, so that $M_{0|x}-M_{1|x}=A_x$ and $N_{0|y}-N_{1|y}=B_y$.}
\[
\begin{aligned}
A_0 &= \sigma_x,  & A_1 &= \sigma_z, \\
B_0 &= \tfrac{1}{\sqrt{2}}(\sigma_z + \sigma_x), &
B_1 &= \tfrac{1}{\sqrt{2}}(\sigma_z - \sigma_x),
\end{aligned}
\]
where $\sigma_{x},\, \sigma_{y}$ and $\sigma_{z}$ are the Pauli operators. Under these settings, the expected value of the CHSH functional for the honest implementation $(\mathsf{S},\mathsf{M})$ is
\[
  \omega(\mathsf{S},\mathsf{M})
  \;=\; \tr\!\big[\rho_i\,B_{\mathrm{CHSH}}\big]
  \;=\;(1-\mu)2\sqrt{2}\,.
\]
To achieve an $\epsilon_{c}$-complete protocol, one should compute a small deviation $\kappa > 0$ in the observed empirical value $\omega_{\mathrm{exp}}$ such that
\begin{equation*}
    \mathrm{Pr}[\omega_{\mathrm{exp}} \geq \omega(\mathsf{S},\mathsf{M}) - \kappa] \geq 1 - \epsilon_{c}.
\end{equation*}
We refer the reader to \ifarxiv \cref{app: security_proofs_par} \else the Supplemental Material~\cite[Section 3]{supp} \fi for an example of this calculation.
\end{example}
\begin{remark}
An honest implementation is one concrete intended realization of the protocol (honest source and devices) that could be implemented in the lab; it serves as a guarantee that acceptance occurs with high probability for well-behaved devices.
This is not to be confused with the ideal protocol in the definition of soundness (\cref{def: soundness}): when proving soundness, we do not assume that the devices behave honestly.\end{remark}}

Combining both completeness and soundness yields a complete notion of security within the composable framework.

\begin{definition}[Security]\label{def:total_security}
A DISC protocol is $(\epsilon_s, \epsilon_c)$-secure if it is both $\epsilon_s$-sound and $\epsilon_c$-complete.
\end{definition}

Specifically, this definition ensures that security is preserved when the protocol is composed with other information theoretic tasks. In particular, if a subsequent protocol is $\epsilon_s'$-sound, then the combined protocol that includes both the DISC protocol and the subsequent protocol will be $(\epsilon_s + \epsilon_s')$-sound. 

\chg{More concretely, consider a sequence of states $\rho_{\text{init}} = \bigotimes_{i=1}^{n}\rho_{i}$ which are certified by a DISC protocol $\mcP_{1}$, whose output is given by $\mcP_{1}(\rho_{\text{init}}) = \rho_{\text{real}}$ where $\mcP_{1}(\cdot )$ is a quantum channel describing the protocol's action. Let $\rho_{\mathrm{ideal}}$ denote the ideal output state of $\mcP_{1}$, and suppose $\mcP_{1}$ is $\epsilon_{1}$ sound according to \cref{def: soundness}, i.e.,
\[
\frac{1}{2} \|\rho_{\text{real}} - \rho_{\mathrm{ideal}}\|_1 \le \epsilon_{1}.
\]
Let $\mathcal{P}_{2}$ be a quantum channel describing the action of a subsequent protocol, which is $\epsilon_{2}$-sound. In particular, $\mathcal{P}_{2}$ satisfies 
\[
\frac{1}{2} \|\mathcal{P}_{2}(\rho_{\mathrm{ideal}}) - \sigma_{\mathrm{ideal}}\|_1 \le \epsilon_{2},
\]
where $\sigma_{\mathrm{ideal}}$ is the ideal output state of $\mathcal{P}_{2}$. What can we say about the composition $\mcP_{2} \circ \mcP_{1}$ when acting on $\rho_{\text{init}}$?  Using the triangle inequality and the soundness of $\mcP_{2}$, 
\begin{multline*}
\frac{1}{2} \|\mathcal{P}_{2} \circ \mcP_{1}(\rho_{\mathrm{init}}) - \sigma_{\mathrm{ideal}}\|_1 
\\ \le \frac{1}{2} \|\mathcal{P}_{2} \circ \mcP_{1}(\rho_{\mathrm{init}}) - \mcP_{2}(\rho_{\mathrm{ideal}})\|_1 + \epsilon_2.
\end{multline*}
By the contractivity of the trace distance under quantum channels (see, e.g.,~\cite[Exercise 9.1.9]{Wilde2013}),
\[
\frac{1}{2} \|\mathcal{P}_{2} \circ \mcP_{1}(\rho_{\mathrm{init}}) - \mcP_{2}(\rho_{\mathrm{ideal}})\|_1 \le \frac{1}{2} \| \mcP_{1}(\rho_{\mathrm{init}}) - \rho_{\mathrm{ideal}}\|_1 \le \epsilon_1,
\]
where the second equality follows from the soundness of $\mcP_{1}$. Combining these facts,
\[
\frac{1}{2} \|\mathcal{P}_{2} \circ \mcP_{1}(\rho_{\mathrm{init}}) - \sigma_{\mathrm{ideal}}\|_1 \le \epsilon_1 + \epsilon_2.
\]
Thus, the protocol formed by composing $\mcP_{1}$ and $\mcP_{2}$ remains (at least) additively secure. A high level of security for the composite protocol can therefore be achieved whenever the individual protocols are themselves highly secure. }
\section{Security proof}
\label{sec:proofMain}
We are now ready to present our main result: proving the DISC protocol described in \cref{sec:introPro} is secure, according to the composable definition given in \cref{sec:compSec}. For simplicity, we choose the target state to be the maximally entangled state $\ket{\psi^*} = \ket{\phi^+}$. However, the proof technique is general enough to replace $\ket{\phi^+}$ with any pure state which can be robustly self-tested, in the sense discussed in \cref{sec: extractability}. Furthermore, we restrict our analysis to the case where the abort condition is defined via a single linear function of the statistics. Specifically, we consider Bell functionals of the form  
\begin{equation}\label{eqn: Bell_score}
   \score := \sum_{x,y \in \{0,1\}} \gamma_{xy} \langle A_x  B_y \rangle, 
\end{equation}
for $\gamma_{xy} \in \mathbb{R}$, and the protocol aborts if the observed value $\score_{\text{exp}}$ is below a threshold value $\score_{\sharp} - \kappa$, where $\kappa > 0$ is chosen to achieve the desired completeness error. \chg{The associated Bell operator is given by $B$ in \eqref{eq:bop},} and we define the maximum Bell coefficient $\gamma^* := \max_{x,y}|\gamma_{xy}|$. When dealing with the sequential setup, we present our results for the CHSH functional, i.e., the case $\gamma_{00} = \gamma_{01} = \gamma_{10} = -\gamma_{11} = 1$, and leave generalizations to future work. For the parallel setup however, the protocol allows for the use of any Bell functional of the form \eqref{eqn: Bell_score}, such as the family~\cite[Proposition 1]{WBC3}. \ifarxiv The security proofs of all parallel protocols (Protocols \ref{prot: DISV} to \ref{prot: DISV_general_var} described in \cref{sec:parPro}) are given in \cref{app: security_proofs_par}. For the sequential protocols (Protocols \ref{prot: DISV_general_sequencial} and \ref{prot: DISV_general_sequencial_var} described in \cref{app:secSetup}), security is proven in \cref{app: security_proofs_sec}. \else The security proofs of all parallel protocols (Protocols 1 to 3) are given in the Supplementary Material~\cite[Section 3]{supp}. For the sequential protocols (Protocols 4 and 5), security is proven in the Supplementary Material~\cite[Section 4]{supp}. \fi 

In the remainder of this section, we present the security proof of one of the aforementioned protocols in detail. All other cases follow a similar structure, and details can be found in the \ifarxiv Appendix. \else Supplementary Material~\cite{supp}. \fi Specifically, below we consider the parallel protocol when the objective is to certify a state $\varepsilon$-close to the target state, \ifarxiv corresponding to \cref{prot: DISV_general} in \cref{sec:parPro}. \else corresponding to Protocol 2 in the Supplementary Material~\cite[Section 2A]{supp}. \fi

\begin{theorem}\label{thm:sec_main_text}
The DISC protocol 2 is $\epsilon_{s}$-sound with $\epsilon_{s}$ \chg{equal to the following:}
\begin{equation*}
\inf_{\delta  > 0} \max\Big\{ \exp \Big( - \frac{n - 1}{\gamma^*} \delta^{2} \Big), \, G_{\varepsilon}\Big( \frac{n-1}{n}[ \score_{\sharp} - \kappa - \delta] + \frac{\eta^{\mathrm{Q}}_{\mathrm{min}}}{n}\Big) \Big\},
\end{equation*}  
where $G_{\varepsilon}(\score)$ is any non-increasing concave function that upper-bounds the function  
$$
\Theta(\sqrt{1 - \Xi_{B}(\omega)} - \varepsilon) \cdot (\sqrt{1 - \Xi_{B}(\omega)} - \varepsilon).
$$  
Here, $\Theta$ is the Heaviside step function, $n$ is the total number of independent states generated by the source, the Bell value $\omega$ is given by \eqref{eqn: Bell_score}, which has a minimum quantum value $\eta^{\mathrm{Q}}_{\mathrm{min}}$, $\score_{\sharp} - \kappa$ is the value that defines the abort condition where $\kappa > 0$ chosen to achieve the desired completeness error, and $\Xi_{B}(\score)$ is the extractability function \chg{given in \cref{def:extract} (see also Figure \ref{fig: extract})}. 
\end{theorem}

Proof can be found \ifarxiv in \cref{sec:DISV_general} (cf. \cref{lem:p2sound}). \else in Supplementary Material~\cite[Section 3B]{supp}. \fi \chg{\noindent For this work, we choose \(G_{\varepsilon}(\cdot)\) to be the function  
\begin{equation}\label{eqn: G_epsilon}
  G_{\varepsilon}(\score) := 
  -\,\mathrm{convenv}\!\Big( - \xi_{B}(\omega, \varepsilon) \cdot \Theta\big(\xi_{B}(\omega, \varepsilon)\big)\Big),     
\end{equation}
where 
\(\xi_{B}(\omega, \varepsilon) := \sqrt{1 - \Xi_{B}(\omega)} - \varepsilon\), and \(\mathrm{convenv}(\cdot)\) denotes the convex envelope, i.e., the tightest convex lower bound of a given function (see the \ifarxiv \cref{rem:g_fun} \else the Supplementary Material~\cite[Remark 5]{supp} \fi for details on its computation). This choice ensures that \(G_{\varepsilon}(\score)\) is the optimal concave function required in Theorem~\ref{thm:sec_main_text}. Moreover, while \cref{thm:sec_main_text} is concerned with soundness, we also provide a proof of completeness in \ifarxiv \cref{lem:p1complete}. \else the Supplementary Material~\cite[Lemma 2]{supp}. \fi} 

The soundness parameter $\epsilon_{s}$ depends on two terms: $\exp \left( - \frac{n - 1}{\gamma^*} \delta^{2} \right)$ and $G_{\varepsilon}\left( \frac{n-1}{n}[ \score_{\sharp} - \kappa -  \delta] \right)$. Both need to be sufficiently small to guarantee security, and each can be understood as a distinct type of penalty. The former penalty increases when the number of copies used for testing is small. Finite statistics effects are prominent in this case, resulting in low statistical confidence of the observed Bell value and weaker security. The second penalty is defined in terms of the so-called \textit{extractability} \cite{bardyn2009device,BNS15,kaniewski2016analytic}, which is a function that indicates how close a state achieving a given Bell violation is to any state in the equivalence class \chg{(i.e., any state that can be transformed into the target state using the allowed class of operations)} of the target state. Bounding this quantity then becomes the central task in proving security. 

\section{Extractability}
\label{sec: extractability}
Let us here elaborate on the extractability, $\Xi_{B}(\score)$, a robustness measure in the context of the security proof. The extractability represents the worst case fidelity between the target state and any state achieving a given Bell violation, following the application of an extraction channel from a set of free operations. Typically, these are taken to be local operations (LO), which originates from the notion of local isometries discussed in \cref{sec:self-test} (see Ref.~\cite{Coopmans19} for the precise connection). This quantity has played a central role in previous works on DISC~\cite{govcanin2022sample,martins24}.
\begin{figure}
    \centering
    \includegraphics[width=\linewidth]{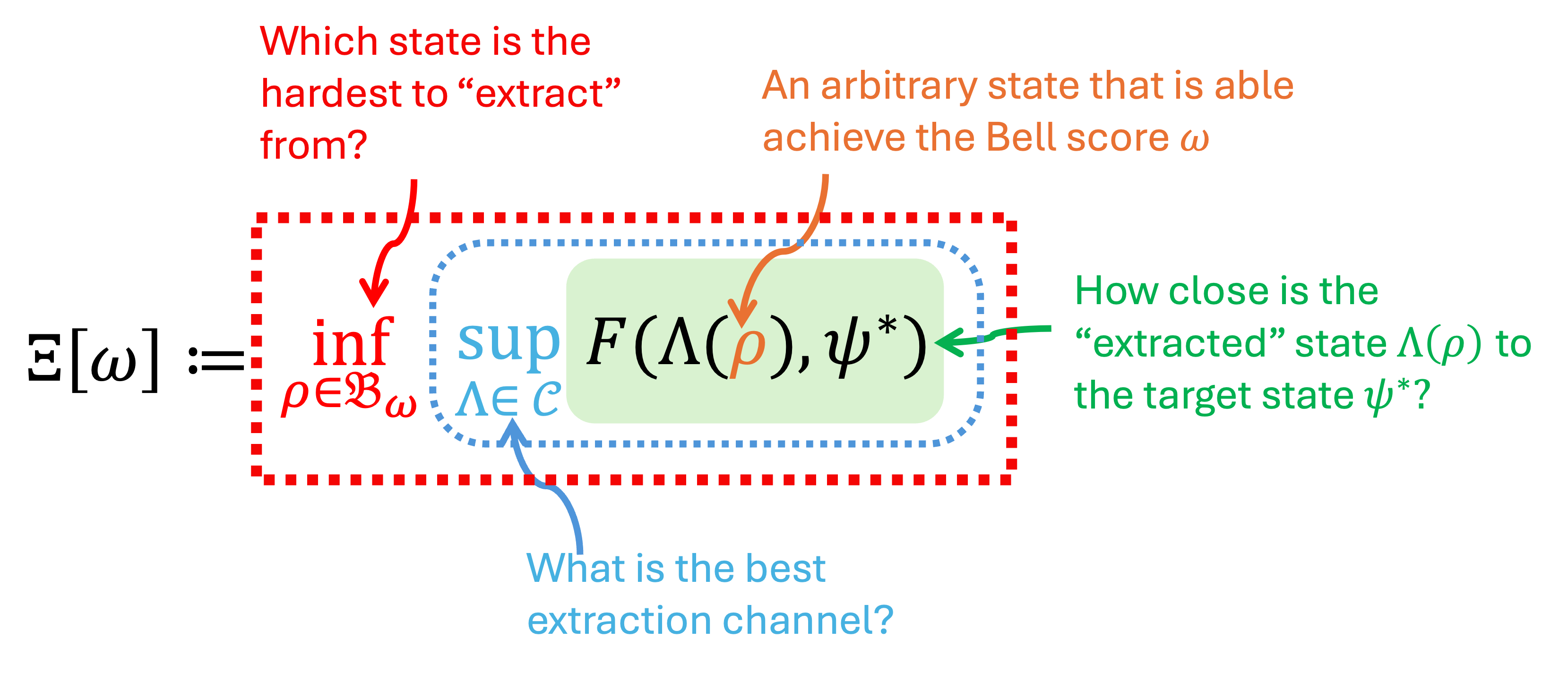}
    \caption{\chg{Definition of the extractibility function. Extractibility quantifies how close a state can be to a target state (or to one that can be converted to the target under a chosen set of channels). It is defined via a min–max optimization, which is generally difficult to compute or bound.}}
    \label{fig: extract}
\end{figure}

Throughout, we consider the extractability under a general class of operations, denoted by $\mathcal{C}$. The choice of $\mathcal{C}$ can be tailored to specific setups, particularly with future protocols in mind. In tasks such as QKD and entanglement distillation, LOCC can be chosen as the class of free operations. The LOCC class includes the LO class, and its use is beneficial since it provides tighter security. Specifically, computing tight lower bounds on the LO extractability for even the simplest case of the CHSH inequality is currently an open question. It is further known that for any CHSH violation up to approximately $2.05$~\cite{Coopmans19,Valcarce_2020}, a non-trivial LO extractability is not possible\footnote{\chg{The trivial extractability in this context is the maximum fidelity between any separable state and the target state (see~\cite[Section 3.6]{SupicSelfTest} for details). This value can always be achieved regardless of the underlying state. See also Example \ref{example: extract}}.}. In contrast, tight bounds on the LOCC extractability for the CHSH inequality are known to be non-trivial for any non-zero violation~\cite{bardyn2009device}.  

We emphasize here that permitting LOCC extraction channels does not enable classical communication between the devices during the Bell test. Indeed, if this were the case, DI certification would become impossible. In our protocol, all classical communication happens strictly after the devices have performed their measurements, during which the user can enforce space-like separation. We also remark that when the target state is the singlet $\ket{\phi^+}$, extractability under LOCC is closely related to the one-shot distillable entanglement defined in~\cite{AFB19}.   

\begin{definition}[Extractability]
Let $\mathcal{C}$ be a class of free operations between a physical system $Q_{A}Q_{B}$ and a target system $\hat{Q}_{A}\hat{Q}_{B}$, $B$ be a Bell operator, and  $\mathcal{B}_{\score} \subset \mcS(\mcH_{Q_{A}}\otimes \mcH_{Q_{B}})$ be the set of states in $Q_{A}Q_{B}$ that can achieve the Bell value $\langle B \rangle \geq \score$ using some measurements. Given a target state $\psi^* \in \mcS(\mcH_{\hat{Q}_{A}}\otimes \mcH_{\hat{Q}_{B}})$, the extractability $\Xi_{B}(\score)$ is defined via the following min-max optimization problem: 
\begin{equation}\label{eqn: extractability}
    \Xi_{B}(\score) := \inf_{\rho \in \mathcal{B}_{\score}} \sup_{\Lambda \in \mathcal{C}} F(\Lambda(\rho), \psi^*),
\end{equation}
where $F(\rho, \sigma) = \| \sqrt{\rho}\sqrt{\sigma}\|_{1}^{2}$ is the fidelity between two states $\rho,\sigma \in \mcS(\mcH)$. \label{def:extract}
\end{definition}

\chg{\begin{example}[Extractability of CHSH under LO channels]\label{example: extract}
    To gain some intuition on how the extractability behaves, consider the CHSH extractability. For simplicity, let us consider the case where $Q_{A},\,Q_{B}\,,\hat{Q}_{A}$ and $\hat{Q}_{B}$ are all qubit systems, $\mcC$ is the set of LO operations and $\psi^* = \phi^{+}$. It is known that all two-qubit states $\rho$ which achieve maximum violation, $\langle B_{\text{CHSH}} \rangle = 2\sqrt{2}$, are of the form $(U_{A} \otimes U_{B})\phi^+ (U_{A} \otimes U_{B})^{\dagger}$ for some local unitary $U_{A} \otimes U_{B}$ (see, e.g.,~\cite[Lemma 10]{WBC3}). In other words, every $\rho \in \mcB_{2\sqrt{2}}$ must be of this form. We then immediately see that for any $\rho \in \mcB_{2\sqrt{2}}$, there exists a $\Lambda \in \mcC$ such that
    \begin{equation}
        F(\Lambda(\rho), \phi^+) = 1,
    \end{equation}
    namely, $\Lambda(\rho) = (U_{A} \otimes U_{B})^{\dagger}\rho (U_{A} \otimes U_{B})$, implying $\Xi_{B_{\text{CHSH}}}^{\text{LO}}(2\sqrt{2}) = 1$.

    For the other extreme, consider the set of possible states which achieve a CHSH value of $\omega = 2$. The extractability can always be lower bounded by choosing a fixed channel of the form $\Lambda(\rho) = \ketbra{00}{00} \ \forall \rho$, which satisfies $F(\Lambda(\rho),\phi^+) = 1/2$. We therefore see $\Xi_{B_{\text{CHSH}}}^{\text{LO}}(2) \geq 1/2$. Moreover, the state $\ketbra{00}{00}$ belongs to $\mcB_{2}$, and $F(\Lambda(\ketbra{00}{00}),\phi^{+}) \leq 1/2$ for all local channels $\Lambda$\footnote{This follows from the fact that the maximum fidelity between $\phi^+$ and $\Lambda(\ketbra{00}{00})$ is achieved when $\Lambda(\ketbra{00}{00})$ is any pure separable state corresponding to the largest Schmidt coefficient of $\ket{\phi^+}$.}. This implies $\Xi_{B_{\text{CHSH}}}^{\text{LO}}(2) \leq 1/2$, and hence $\Xi_{B_{\text{CHSH}}}^{\text{LO}}(2) = 1/2$.     

    The extractability of CHSH under LO channels thus takes values in the interval $[1/2,1]$ for $\omega \in [2,2\sqrt{2}]$. It was shown by Kaniewski~\cite{kaniewski2016analytic} that a lower bound for all $\omega \in [2,2\sqrt{2}]$ is given by
    \begin{equation}
        \Xi_{B_{\text{CHSH}}}^{\text{LO}}(\omega) \geq \max\Big\{\frac{1}{2}\Big( 1 + \frac{\omega - \omega^*}{2\sqrt{2} - \omega^*}\Big),\frac{1}{2}\Big\},
    \end{equation}
    where $\omega^* = (16 + 14\sqrt{2})/17 \approx 2.11$ is the threshold CHSH value below which the extractibility is trivial (below $1/2$). It was shown by Refs.~\cite{Coopmans19,Valcarce_2020} that this threshold cannot be lowered below $\approx 2.05$.  This contrasts the tight bound on $\Xi_{B_{\mathrm{CHSH}}}$ when $\mcC$ is taken to be the class of LOCC operations, derived by Bardyn \textit{et al.}~\cite{bardyn2009device}:
    \begin{equation}
        \Xi_{B_{\text{CHSH}}}^{\text{LOCC}}(\omega) = \frac{1}{2}\Big( 1 + \frac{\omega - 2}{2\sqrt{2} - 2}\Big). \label{eq:LOCC_CHSH}
    \end{equation}
    
\end{example}}

The extractability function involves two levels of optimization. The inner optimization considers a state $\rho$ that can achieve the Bell value $\score$ and aims to transform it, via operations in the class $\mathcal{C}$, to a state as close as possible (in terms of fidelity, rather than trace distance) to the target state $\psi^*$. This provides a meaningful measure of how close a state is to the equivalence class of the target state. The outer optimization then finds the state $\rho$ for which this distance is smallest, provided $\rho$ can achieve the given Bell value $\score$ via some local measurement strategy.  

Computing the extractability function is challenging in general, since it involves a min-max optimization. Furthermore, the optimization runs over all channels in a given class, and all states compatible with a Bell value $\score$, without assuming their dimension. Additionally, the constraint $\rho \in \mathcal{B}_{\score}$ is nonlinear in both the state and the measurements. To address these challenges, existing works (focused on LO extractability) have bypassed the inner optimization over channels by selecting a fixed channel for all states $\rho$ and values $\omega$, resulting in a valid lower bound~\cite{kaniewski2016analytic}. The issue of not assuming the system dimension can, in general, be handled numerically via moment matrix approaches~\cite{BNS15}, or in the special case of Bell scenarios with binary inputs and binary outputs via a reduction to qubits known as Jordan's lemma~\cite{Jordan}.   

We here derive a sequence of lower bounds on the LOCC extractability in the minimal Bell scenario for Bell functionals of the form \eqref{eqn: Bell_score}. Moreover, our bounds can be \chg{improved further} at the expense of increasing computational cost. \chg{Our results are summarized in \cref{thm: main_text_extract}, and the following text explains how this can be used to obtain a sequence of lower bounds.}

\begin{theorem}\label{thm: main_text_extract}
Let $B$ be any Bell functional of the form \eqref{eqn: Bell_score} in the minimal Bell scenario. Then the LOCC extractability $\Xi_{B}(\score)$ satisfies:
\begin{equation}
    \Xi_{B}(\score) \chg{\geq} \mathrm{convenv} \Big(\min_{(a,b) \in \mathcal{F}_{\score}} f_{a,b}(\score)\Big),
\end{equation}
where $\mathrm{convenv}(\cdot)$ is the convex envelope (\chg{tightest} convex lower bound)  and 
\begin{equation}\label{eqn: f_{a,b}}
    \begin{aligned}
        f_{a,b}(\score) := \max \ & \lambda \, \score + \mu \\
        \mathrm{s.t.} \ & \ \sigma - \lambda B(a,b) - \mu \id_{4} \geq 0, \\
        & \ \tr_{\hat{Q}_{A}}[\sigma] = \ \tr_{\hat{Q}_{B}}[\sigma] = \frac{\id_{2}}{2}, \\
        & \ \sigma \in \mathcal{S}_{2}, \ \lambda \geq 0, \ \mu \in \mathbb{R}.
    \end{aligned} 
\end{equation}
Here $\mathcal{S}_{2} = \mcS(\mcH_{\hat{Q}_{A}} \otimes \mcH_{\hat{Q}_{B}})$ is the set of two-qubit density operators and $B(a,b)$ is the Bell operator $B$ constructed from the qubit observables:
\begin{eqnarray*}
    A_{x} &=& \cos(a) \, \sigma_{Z} + (-1)^{x} \sin(a) \, \sigma_{X}, \\
    B_{y} &=& \cos(b) \, \sigma_{Z} + (-1)^{y} \sin(b) \, \sigma_{X}.
\end{eqnarray*}
The set $\mathcal{F}_{\score}$ is defined as:
\begin{equation*}
    \mathcal{F}_{\score} = \left\{ (a,b) \in [0, \pi/2]^{\times 2} : \exists \rho \in \mathcal{S}_{2} \ \mathrm{s.t.} \ \tr[B(a,b)\rho] \geq \score \right\}.
\end{equation*}
\label{thm:SDP_main}
\end{theorem}

The proof of \cref{thm:SDP_main} is presented in \ifarxiv \cref{app: bounds on the singlet fidelity}, \else the Supplementary Material~\cite[Section 6C]{supp}, \fi (see also \ifarxiv \cref{fig: proof} \else the Supplementary Material~\cite[Figure 4]{supp} \fi for an informal overview) and let us here discuss its implications. The first significant aspect of this result is the dimensional reduction of the optimization problem required to compute the extractability. By employing Jordan's lemma, we reduce the problem to effectively computing the extractability function within the state space of a qubit pair. Subsequently, we perform a series of reductions inspired by techniques from device-independent randomness generation and QKD protocols in the minimal Bell scenario~\cite{PABGMS, Bhavsar_thesis, Bhavsar2023, zhu2024interplay}. These reductions, combined with other technical results, allow us to reformulate the optimization over unital LOCC channels into a standard optimization problem over a bounded domain. 

Assuming the projective measurements performed by the two devices on the qubit pair are known, the extractability can be computed numerically. This follows from the fact that, for a fixed $(a,b) \in \mbR^{2}$, the optimization \eqref{eqn: f_{a,b}} is a semidefinite program (SDP), which can be efficiently solved using numerical techniques \cite{Boyd2004}. However, the outer maximization over all possible two-qubit Bell operators \( B(a,b) \) still remains, complicating the original problem as it is no longer an SDP. 

To address this, we develop a technique to discretize the parameter space of the angles \( (a,b) \in [0, \pi/2] \times [0, \pi/2] \) into smaller rectangular domains \ifarxiv (see Appendix \ref{app: gridding}). \else (see the Supplementary Material~\cite[Section 6D]{supp}). \fi This discretization transforms the problem into solving multiple SDPs of the form \eqref{eqn: f_{a,b}}, each corresponding to a specific grid point within the rectangular domains. Specifically, for each domain, we relax the optimization problem and introduce a penalty term that scales with the dimensions of the domain, ensuring we reliably lower bound the global minimization. By reducing the size of each rectangular domain, we achieve tighter bounds on the extractability function at the cost of solving more optimization problems, and thus an increased computation time. Furthermore, as the size of each rectangle tends to zero, the method converges to a \chg{tighter} lower bound on the LOCC extractability. 

Note that an analytic method for computing extractability in the LOCC case was first introduced in~\cite{bardyn2009device}, where only the CHSH functional was considered. The approach presented here is significantly more general, encompassing more self-tests of the singlet-- that is, it applies to all Bell inequalities of the form \eqref{eqn: Bell_score}. Moreover, our method allows for the simultaneous use of multiple Bell inequalities, or even the full distribution, when bounding the extractability. As a result, it provides a framework for obtaining \chg{lower} bounds in the minimal Bell scenario when self-testing the singlet. Additionally, while \cref{thm:SDP_main} addresses the LOCC extractability, our gridding techniques can also be applied to other classes of free operations, such as the LO extractability for arbitrary Bell functionals in the minimal scenario, including those tailored to partially entangled states~\cite{AcinRandomnessNonlocality,BampsPironio,Coopmans19}.

\section{Results for example \ref{example: CHSH}: CHSH based protocol for certifying $\phi^+$}
\label{sec:example}
To illustrate our results, we consider a DISC protocol based on violating the CHSH inequality in the parallel measurement setup, where the objective is to certify a state $\varepsilon$-close to $\phi^+$ \ifarxiv (see \cref{prot: DISV_general} in the appendix). \else (see Protocol 2 in the the Supplementary Material~\cite{supp}). \fi A bound on the CHSH LOCC extractability was provided in~\cite{bardyn2009device}, given by the linear function $\Xi_{\text{CHSH}}(\omega) \geq = 1/2 + (\omega - 2)/(4\sqrt{2} - 4) =: g(\omega)$ for $\omega \in [2,2\sqrt{2}]$. Using this bound, we plot the penalty function $G_{\varepsilon}(\score)$ \chg{(defined via substituting $\Xi_{B}(\score)$ in \eqref{eqn: G_epsilon} with the function $g(\omega)$)} in \cref{fig: G_epsilon}, which provides an estimate of the security bounds as a function of the chosen abort condition, characterized by the Bell value $\score$.  

\begin{figure}[h!]
    \includegraphics[width=0.42\textwidth]{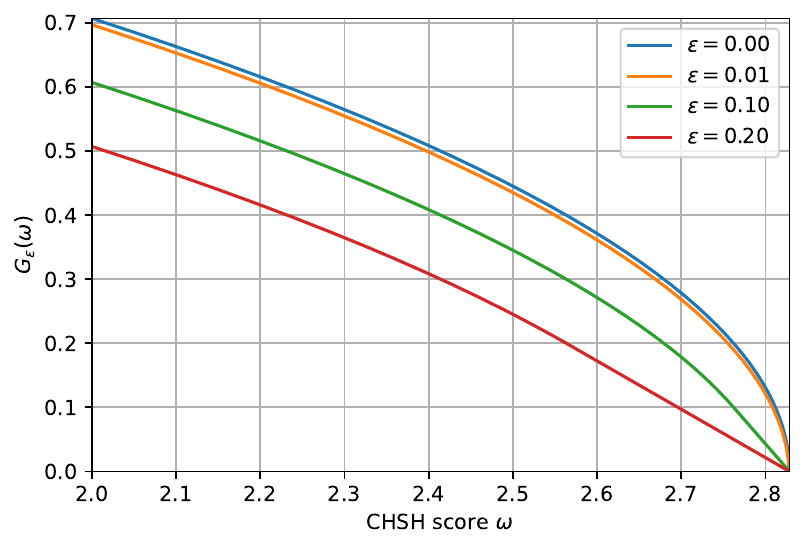}
        \caption{Graph of $G_{\varepsilon}(\omega)$ for different values of $\varepsilon$, using LOCC extractability.}
        \label{fig: G_epsilon}
\end{figure}

\begin{figure}[h!]
    \centering
    \begin{subfigure}[t]{0.42\textwidth}
        \centering
        \includegraphics[width=\textwidth]{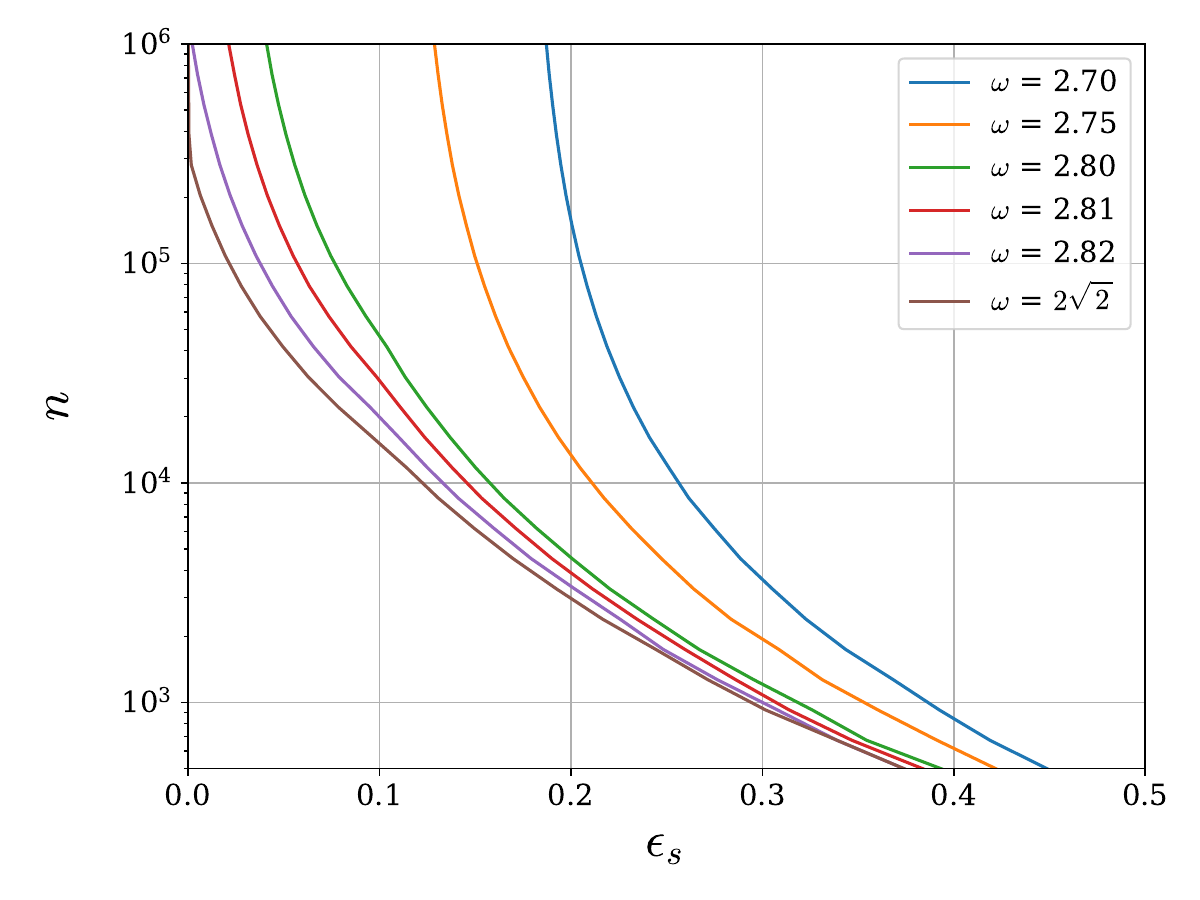}
        \caption{Plot of the security parameter $\epsilon_{s}$ versus the number of rounds $n$  for different CHSH values $\score$. Here, we set $\varepsilon = 0.1$ and $\kappa$ is chosen to achieve a completeness error of \chg{$\epsilon_{c} = 10^{-2}$. The choice of $\epsilon_{c}$ follows standard values used in related device-independent protocols (see, e.g., \cite{ADFRV}).} }
        \label{fig:n_vs_espilon_at_fixed_var_epsilon_comp}
    \end{subfigure}
    \hfill
    \begin{subfigure}[t]{0.42\textwidth}
        \centering
        \includegraphics[width=\textwidth]{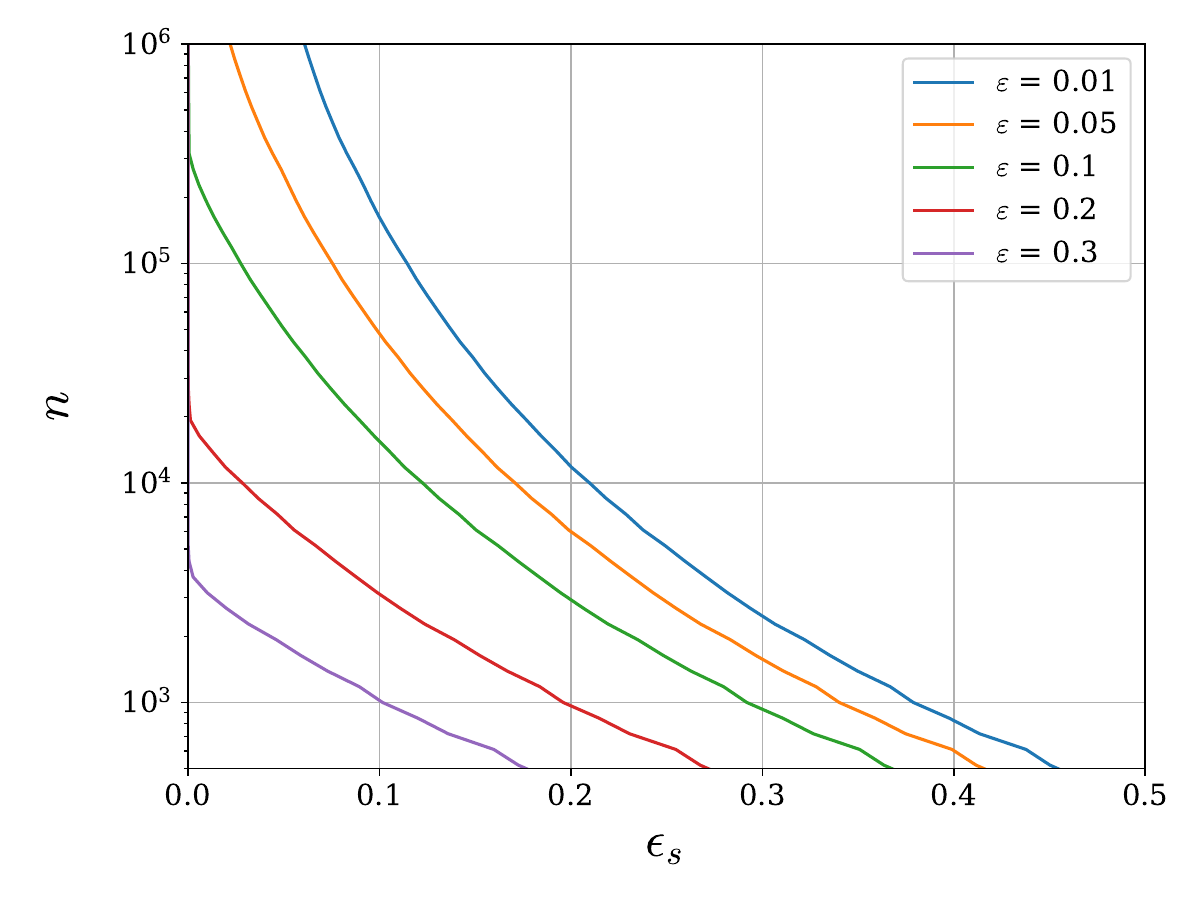}
        \caption{Plot of the security parameter \( \epsilon_{s} \) versus the number of rounds \( n \) for different values of \( \varepsilon \), assuming a CHSH value of \( \omega = 2\sqrt{2} \). The parameter $\kappa$ is chosen to achieve a completeness error of \chg{$\epsilon_{c} = 10^{-2}$}.}
        \label{fig:n_vs_espilon_at_fixed_score_completeness}
    \end{subfigure}
    \caption{Comparison of security parameters for different conditions using the CHSH inequality.}
    \label{fig:combined_figure_completeness}
\end{figure}

Knowing $G_{\varepsilon}(\score)$ enables us to compute the security bounds for the protocol via \cref{thm:sec_main_text}, which we present in \cref{fig:combined_figure_completeness}. From this figure, two key trends emerge: first, we obtain tighter security from higher CHSH values $\score$, which arises from a smaller value of $G_{\varepsilon}(\omega)$. Secondly, security improves as the number of rounds increases, owing to a smaller penalty due to finite statistics (captured by the exponential term in \cref{thm:sec_main_text}). Additionally, we observe that for larger values of the closeness parameter $\varepsilon$, security is higher for the same Bell value and number of rounds. This is expected, as certifying a state that is $\varepsilon$-close to the target state requires a lower Bell value than certifying the target state exactly. Thus, increasing $\varepsilon$ results in a smaller values of $G_{\varepsilon}(\omega)$, corresponding to a smaller penalty.  

\section{Discussion}
\label{sec:disc}

We have presented a composable approach for device-independent state certification under the assumption of an independent but not identically distributed source. We introduced a definition for composable DISC security, and provided a general framework for proving security in the two-input two-output Bell scenario under LOCC operations. 

For future directions, it would be interesting to apply our protocols in practice. Clearly, the advantage lies in the composable integration of DISC with any other composable protocol. For example, consider a protocol $\mathcal{P}$ which, when given as input the state $\phi^{+}$, outputs a state $\mathcal{P}(\phi^{+})$ satisfying $\| \mcP(\phi^{+}) - \sigma_{\text{ideal}}\|_{1} \leq \epsilon'$, where $\sigma_{\text{ideal}}$ is some target output state of $\mcP$. Now, suppose the DISC protocol outputs a state $\rho_{\text{real}}$ with the property\footnote{Here we have omitted the fact that both the DISC protocol and $\mcP$ may abort for ease of discussion.} $\| \rho_{\text{real}} - \phi_{+}\|_{1} \leq \epsilon$. Then the composable security definition ensures that, when the DISC protocol output is used as an input to $\mcP$, the result is secure: $\| \mcP(\rho_{\text{real}}) - \sigma_{\text{ideal}}\|_{1} \leq \epsilon + \epsilon'$. Such applications might include the certification of other quantum resources, along the lines of Ref.~\cite{SekatskiBuilidngBlocks}, where state certification is an essential building block. 

We also note that DISC is not vulnerable to the same device-reuse attacks as in, e.g., DIQKD~\cite{bckone}. This follows from the fact that no classical information is kept private from an adversary during the protocol. It is then an interesting question if DISC always remains secure when the measurement devices are reused.   

It would also be interesting to improve the DISC security statement. Both a large Bell violation and a large number of copies are currently required to obtain a composable security proof, which is limiting in practice. This could be due to two reasons. First, the proof technique relies on inequalities and bounds that may not be tight, suggesting room for improvement in obtaining sharper security bounds. For example, our security proof employs Hoeffding’s inequality, which could potentially be replaced by tighter alternatives such as those used in Ref.~\cite{govcanin2022sample}. The second reason is that, from a fundamental point of view, it is is an inherently strong requirement to ensure general security under any future protocol usage. Consequently, achieving tight security bounds for small Bell violations and low numbers of rounds may be infeasible~\cite{wiesner2024}.   

Nevertheless, it is notable that composable security can be achieved. Moreover, our security bounds are tighter than those obtained using LO extractability bounds. In realistic experimental implementations, improved certification could be achieved by relaxing the stringent fully DI assumptions and incorporating justified partial assumptions. For instance, the fair sampling assumption could be employed to account for poor detector efficiencies. Additionally, similar results may be obtained in a semi-DI setting, where assumptions on system dimensions are introduced.  

\chg{Another promising direction would be to certify more than one copy of the target state. In fact, our current approach can be straightforwardly modified to accommodate this. However, the soundness parameter $\epsilon_{s}$ scales as $ \sqrt{1 - c^{m}}$, where $c$ is the single copy extractability for the threshold Bell value which does not cause the protocol to abort, and $m$ is the number of certified copies. We therefore see that the resulting protocol will not be efficient, in the sense that a large number of copies can only be certified when $c \approx 1$, which demands both a near-maximum Bell violation and a large number of samples $n$. Furthermore, there are recent no-go results~\cite{wiesner2024} which rule out sample efficient and composable state certification when the desired certification is ``exact'', i.e., the target state is certified rather than tolerating small deviations from it. Thus, understanding what is possible for composable multiple copy DI state certification is an appealing direction.}

In addition, it would be useful to relax the assumption of independent state preparation in each round. Removing this assumption and allowing for general memory effects would lead to a more general security proof. \chg{As discussed in~\cite{govcanin2022sample}, the task of certifying one copy can be achieved under an arbitrarily correlated source (see also~\cite{HM19a,HM19b} for a device dependent approach), and such techniques may provide a path to establishing similar results in our composable framework. However, extending the DISC framework of Ref.~\cite{govcanin2022sample} to multiple copies in the fully non-i.i.d. setting remains an open question. Due to the more demanding requirement of composability, we expect this will be at least as challenging to establish in our case. Moreover,} there are potential limitations when considering a fully general measurement process. As discussed in Ref.~\cite[Section 2.2.1]{AFB19}, the ability to ``not measure'' a quantum system and hold it in memory necessitates some separation between the state and measurement devices. 

Our result on the LOCC extractability for the singlet state may also be of independent interest. This quantity serves as the natural DI counterpart to the well-known singlet fraction~\cite{horodecki1999general}, extending its relevance to the DI setting. This raises open questions regarding its potential applications in other areas of entanglement theory, as well as in the development of new DI protocols. Additionally, our results provide another avenue for exploring the relationship between nonlocality and entanglement~\cite{zhu2024interplay}, which remains a fundamental topic of investigation.

Finally, it is also interesting to consider a weaker certification criterion-- namely, certifying the presence of \textit{any} pure entangled state rather than a specific one. Such a certification could have significant cryptographic applications, as it has been demonstrated that randomness and cryptographic keys can be extracted in a fully device-independent manner from non-maximally entangled, yet still entangled, states~\cite{AcinRandomnessNonlocality,Woodhead_21}.  

\acknowledgements
The authors are grateful to Peter Brown, Roger Colbeck and Ivan {\v S}upi\'c for insightful discussions. We also thank Cameron Foreman, Ashutosh Rai, Olgierd Żurek, Mirjam Weilenmann and anonymous referees for their valuable feedback on earlier versions of this work. RB and JB are supported by the National Research Foundation of Korea (Grant No. NRF-2021R1A2C2006309, NRF-2022M1A3C2069728) and the Institute for Information \& Communication Technology Promotion (IITP) (RS-2023-00229524, RS-2025-02304540). LW acknowledges funding support from the Engineering and Physical Sciences Research Council (EPSRC Grant No. EP/SO23607/1) and the European Union’s Horizon Europe research and innovation programme under the project ``Quantum Secure Networks Partnership'' (QSNP, grant agreement No. 101114043). Preliminary investigations for this work were supported by the
EPSRC via the Quantum Communications Hub (Grant No. EP/T001011/1) during RB’s time at the University of
York.


%

\onecolumngrid
\appendix

\section{Overview of assumptions for DISC protocols}\label{app:assumptions}
In this section, we outline all assumptions made in our work.    
\begin{enumerate}
    \item Quantum theory is correct and complete.
    \item No information can leak in or out of the laboratory once the protocol has begun. 
    \item The untrusted source generates a sequence of independent states.
    \item The user has access to a secure quantum memory, and a trusted means to process classical information. 
    \item The user has access to a trusted source of perfect, private randomness. In particular, this implies the random variables $X_i$, $Y_i$, and $T$ are uniformly distributed to the user, and to any potential adversary present in the current protocol, or in any future protocol for which the current protocol serves as an input. 
    \item\label{ass: initial_seed_randomness} All initial states from the source are received in the laboratory before the \chg{random number} $T$ is generated.
\end{enumerate}

It is important to emphasize that, unlike standard device-independent protocols for quantum key distribution and randomness generation, this protocol requires a clear separation between states \chg{(generated solely by the source)} and measurements \chg{(performed by the measurement device)}, rather than treating them as a single uncharacterized ``black-box''. \chg{In particular, all entanglement produced during or before the protocol is attributed to the source only.} This distinction is critical, since treating states and measurements as a single box would render the protocol trivially insecure. For example, an eavesdropper could prepare a maximally entangled state $\ket{\phi^+}$ in each round, and instruct the devices to measure all states projectively, according to the optimal measurement strategy for some Bell inequality. This would destroy the entanglement in all the states, regardless of whether a particular round was intended to serve in the Bell test or not. Under this attack, the protocol will not abort, however, the output state stored for the user is separable. This violates the security requirement, namely, that the output state resembles $\ket{\phi^+}$ when the protocol does not abort. 

In contrast, our protocol eliminates this vulnerability by randomly choosing the output state before any interaction with the measurement device. This state is then held in a trusted quantum memory while the remaining states are measured, ensuring it is shielded from any external influence.

\chg{We now discuss Assumption \ref{ass: initial_seed_randomness}. It is essential to have access to a private source of randomness during the protocol in order to choose the stored state and perform the Bell test. In particular, it suffices to assume that this randomness is not available to the adversary before the protocol commences. Otherwise, the adversary could prepare the sequence of states $\bigotimes_{i=1}^{n} \rho_{i}$ with $\rho_{i} = \phi^+$ whenever $i \neq t$ and $\rho_{t} = \sigma$, where $\sigma$ is some separable state. If the measurement devices are instructed to always perform the optimal measurements for the desired Bell inequality, this would essentially amount to the abort-based attack discussed in the main text, except that it would now succeed with probability one. Assumption \ref{ass: initial_seed_randomness} excludes this attack, and is indispensable for the protocol to remain secure. 

We stress, however, that it is permissible for the adversary to learn the value of the random variable $T$ once the protocol has already commenced. At that stage, the adversary has no ability to pre-program the source and measurement devices in a coordinated manner to break security. Finally, we note that this assumption could be entirely dropped if the random numbers were generated by a randomness-generation protocol that itself is composable.\footnote{We thank the authors of \cite{wiesner2024} for pointing this out to us.} }

\section{Protocols for DISC}\label{app: protocol} \label{app:allProtocols}

We now present the DISC protocols discussed in the main text. Specifically, in \ifarxiv \cref{sec: measSetup} \else Section 4 of the main text \fi we considered two variants of the measurement setup. The first consists of a parallel scenario, in which each state $\rho_{i}$ is measured in isolation using a separate measuring device. The second is sequential, where the measurement of $\rho_{i}$ precedes that of $\rho_{i+1}$, and auxiliary information about the measurement in round $i$ can be used in round $i+1$. This setup consists of a single measurement device. \chg{Throughout, we use the notation $\mathcal{M}_{i}$ to denote the measurement channel associated to the index $i$, which include the settings $X_{i}Y_{i}$ as an input (see \cref{fig: DISV}). This is not to be confused with the channels $\mathcal{N}_{i}$ described in the main text, in which $X_{i}Y_{i}$ are included as outputs.} We consider both the task of certifying the target state exactly, and a state $\varepsilon$-close to the target state. 
\subsection{Parallel setup} \label{sec:parPro}

For the parallel setup, we remark that instead of requiring $n-1$ different non-communicating measurement devices, the protocol can be reinterpreted as involving a single measurement device without memory. This reformulation aligns the protocol with sequential protocols, where a single memoryless device processes the measurements one at a time. Nonetheless, no assumptions are made regarding the inner workings of the measurement devices. Furthermore, no assumptions are made about the generated states: the channels $\mathcal{M}_{i}$ can be pre-programmed in accordance with the state $\rho_{i}$, which itself can be pre-set by a potential adversary for the protocol.

With this in mind we present \Cref{prot: DISV}, which consists of certifying the maximally entangled state $\ket{\phi^+}$, following the action of an optimal extraction channel, using generalized Bell functionals of the form 
\begin{equation}
    \omega = \sum_{x,y \in \{0,1\}} \gamma_{x,y} \langle A_{x}B_{y} \rangle \label{eq:appBop}.
\end{equation}
The maximum and minimum values of $\omega$ for quantum behaviors are denoted $\eta^{\text{Q}}_{\text{min}}$ and $\eta^{\text{Q}}_{\text{max}}$, respectively. The maximum and minimum values for local behaviors are denoted $\eta^{\text{L}}_{\text{min}}$ and $\eta^{\text{L}}_{\text{max}}$, respectively. We label the sequence of states produced by the untrusted source $\big\{\rho_{i} \in \mcS(\mcH_{Q^{A}_{i}} \otimes \mcH_{Q_{i}^{B}}) \big\}_{i=1}^{n} $,  and for each $\rho_{i}$ we associate a channel $\Lambda_{i} \in \mcC$ which satisfies $F(\Lambda_{i}(\rho_{i}),\phi^+) = \sup_{\Lambda \in \mcC} F(\Lambda(\rho_{i}),\phi^+)$. Each measurement device is labeled $M_{i}$, consisting of isolated sub-devices $M_{i}^{A}$ and $M_{i}^{B}$

\begin{remark}
    Note that we have implicitly assumed the supremum over channels in $\mcC$ is achievable. If this is not the case, we define $\Lambda_{i}$ as any channel which achieves a fidelity arbitrarily close to $\sup_{\Lambda \in \mcC} F(\Lambda(\rho_{i}),\phi^+)$. 
\end{remark}

\vspace{0.5cm}

\begin{center}
\begin{mdframed}[linecolor=black, roundcorner=5pt, skipabove=10pt, skipbelow=10pt, backgroundcolor=white, splittopskip=10pt, splitbottomskip=10pt]
\begin{protocol}[Certification of the $\phi^+$ state]\label{prot: DISV}  
\noindent\textbf{Parameters:}\\
$n \in \mbN^{+}$ -- number of rounds \\ 
$p_{T}:\{1,...,n\}\to [0,1]$ -- probability distribution of the random variable $T$ (taken to be uniform here) \\ 
$\score_{\sharp} \in [\eta^{\text{Q}}_{\text{min}},\eta^{\text{Q}}_{\text{max}}]$ -- expected value of the Bell functional \eqref{eq:appBop}\\
$\kappa > 0$ -- completeness parameter.

\begin{enumerate}
    \item\label{p1_step: 1} Generate a random variable $T$ according to the distribution $p_{T}$. If $T = t$, then store the state $\rho_{t}$ for the remainder of the protocol. 
    \item\label{p1_step: 3}  Generate the random bit string $\mathbf{X} = (X_1, X_2, \ldots, X_{t-1}, X_{t+1}, \ldots, X_{n})$ uniformly. Input each bit $X_{i}$ to the device $M_{i}^{A}$, producing the output bit $A_i$. Similarly, generate the random bit string $\mathbf{Y} = (Y_1, Y_2, \ldots, Y_{t-1}, Y_{t+1}, \ldots, Y_{n})$ uniformly, and input $Y_{i}$ to the device $M_{i}^{B}$, producing the output $B_i$. 
    \item For $i \in \{1,...,n\} \setminus t$, set    $W_{i} = \tilde{\gamma}_{X_{i}, Y_{i}}$ if $A_i \oplus B_i = X_i \cdot Y_i$, and   $W_{i} = -\tilde{\gamma}_{X_{i}, Y_{i}}$ otherwise, where $\tilde{\gamma}_{x,y} = (-1)^{xy}\gamma_{x,y}$ .
    \item Compute the empirical value:
    \begin{equation}
    \score_{\text{exp}} := \frac{4}{n} \sum_{i=1 \, : \, i \neq t}^{n} W_{i} 
    \end{equation}
    and abort the protocol if $\score_{\text{exp}} \leq \score_{\sharp} - \kappa$.
    \item If the protocol does not abort, apply the optimal channel $\Lambda_{t} \in \mcC$ to the state $\rho_{t}$, which transforms $\rho_{t}$ to a state $\Lambda_{t}(\rho_{t})$. Output $\Lambda_{t}(\rho_{t})$.
\end{enumerate}
\end{protocol}
\end{mdframed}
\end{center} 

\vspace{0.5cm}
To see how the empirical value relates the Bell functional \eqref{eq:appBop}, consider the variable $W$ for a single round (omitting the index $i$). Then
\begin{equation}
    \begin{aligned}
        \mathbb{E}[W] &= \sum_{x,y \in \{0,1\}} \Big(\mathbb{P}[W = \tilde{\gamma}_{x,y}] \tilde{\gamma}_{x,y} - \mathbb{P}[W = -\tilde{\gamma}_{x,y}] \tilde{\gamma}_{x,y}\Big).
    \end{aligned}
\end{equation}
Note that
\begin{equation}
    \begin{aligned}
        \mathbb{P}[W = \tilde{\gamma}_{x,y}] &= p(x,y) \sum_{a,b  \, : \, a\oplus b = xy} p(a,b|x,y) \\
        &= \frac{1}{8}\sum_{a,b \in \{0,1\}} p(a,b|x,y) \big( 1 + (-1)^{a+b+xy}\big) \\
        &= \frac{1}{8}\Bigg( 1 + (-1)^{xy}\sum_{a,b \in \{0,1\}}(-1)^{a+b}p(a,b|x,y)\Bigg) \\
        &= \frac{1}{8}\big( 1 + (-1)^{xy}\langle A_{x}B_{y}\rangle\big),
    \end{aligned}
\end{equation}
where we used the fact that $p(x,y) = 1/4$. We also have 
\begin{equation}
    \mathbb{P}[W = -\tilde{\gamma}_{x,y}] = p(x,y)\Bigg( 1-\sum_{a,b  \, : \, a\oplus b = xy} p(a,b|x,y) \Bigg) =  \frac{1}{4} - \frac{1}{8}\big( 1 + (-1)^{xy}\langle A_{x}B_{y}\rangle\big).
\end{equation}
As a result,
\begin{equation}
    \mathbb{E}[W] = \frac{1}{4} \sum_{x,y \in \{0,1\}} \tilde{\gamma}_{x,y}(-1)^{xy} \langle A_{x}B_{y}\rangle = \frac{1}{4} \sum_{x,y \in \{0,1\}} \gamma_{x,y} \langle A_{x}B_{y}\rangle = \frac{1}{4}\omega.
\end{equation}
A graphical description of Protocol \ref{prot: DISV} can be found in \Cref{fig: DISV}. 

\begin{figure}[h!]
    \centering
    \includegraphics[width=0.8  \textwidth]{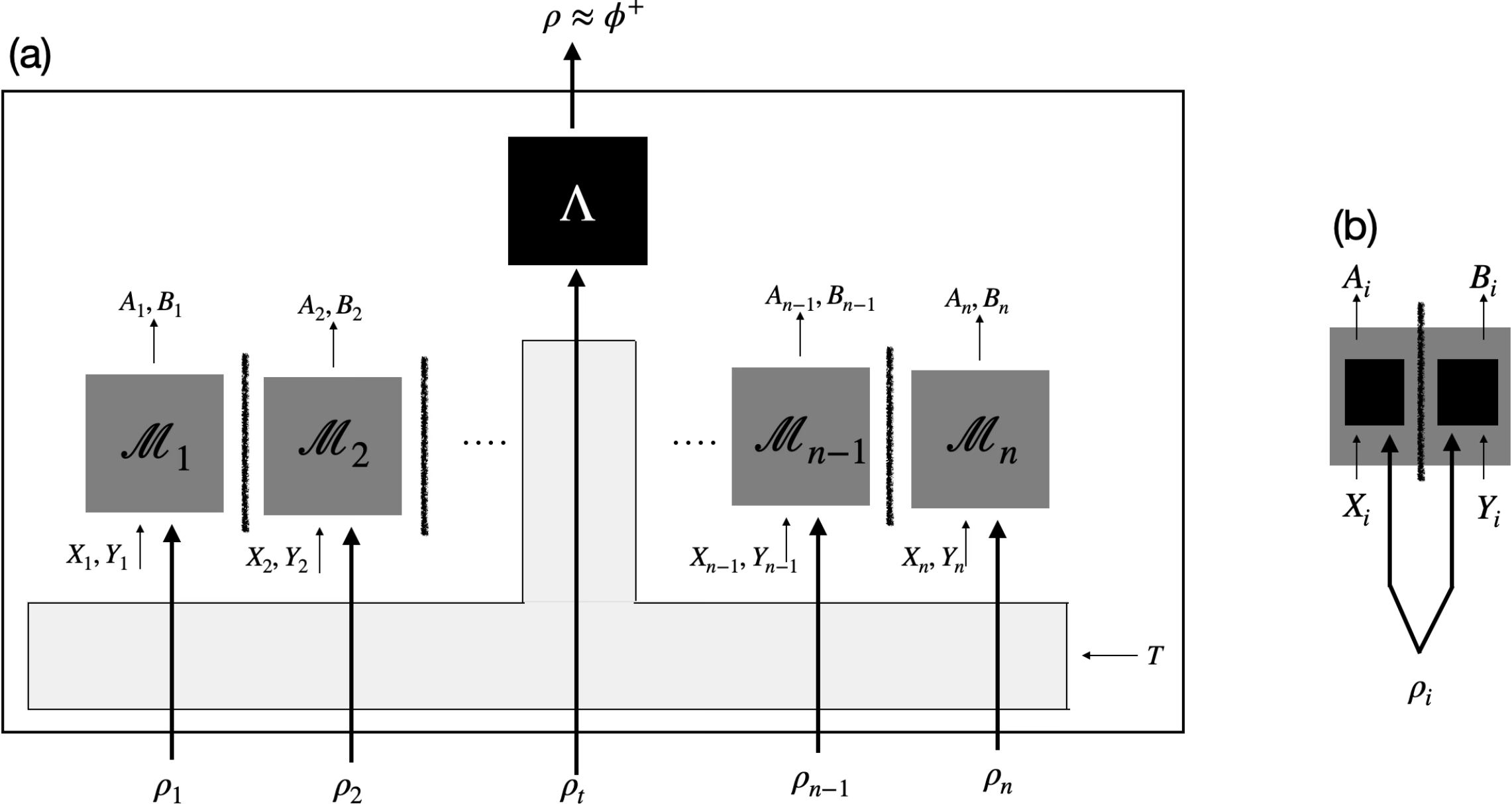}
    \caption{(a) Description of Protocol \ref{prot: DISV} in terms of its individual components. The bold arrows lines indicate quantum systems and the thin arrows represent classical variables. The bold line indicates that no communication is allowed between the components which it separates. (b) Description of the individual measurement channel $\M{M}_{i}$ in terms of devices that perform the Bell test.   }
    \label{fig: DISV}
\end{figure}

\Cref{prot: DISV} certifies the maximally entangled state $\ket{\phi^+}$. However, in practical scenarios, one may wish to certify a quantum state that is $\varepsilon$-close to $\phi^+$, where the closeness is measured using the trace norm. That is, if the protocol does not abort, then the state held in memory, $\Lambda_{t}(\rho_{t})$, satisfies $|| \Lambda_{t}(\rho_{t}) - \phi^+ ||_1 \leq \varepsilon$. The next protocol we present extends \Cref{prot: DISV} to account for deviations from the idealized scenario of perfect state preparation. Whilst the steps are the same as \Cref{prot: DISV}, the ideal protocol differs, resulting in a different security proof (see \Cref{app: security_proofs_par}).

\vspace{0.5cm}
 
\begin{center}
\begin{mdframed}[linecolor=black, roundcorner=5pt, skipabove=10pt, skipbelow=10pt, backgroundcolor=white, splittopskip=10pt, splitbottomskip=10pt]
\begin{protocol}[Certification of a state $\varepsilon$-close to the $\phi^+$ state]\label{prot: DISV_general}
\noindent\textbf{Parameters}:\\
$n \in \mbN^{+}$ -- number of rounds \\ 
$p_{T}:\{1,...,n\}\to [0,1]$ -- probability distribution of the random variable $T$ (taken to be uniform here) \\ 
$\score_{\sharp} \in [\eta^{\text{Q}}_{\text{min}},\eta^{\text{Q}}_{\text{max}}]$ -- expected value of the Bell functional \eqref{eq:appBop}\\
$\varepsilon \geq 0$ -- closeness parameter\\
$\kappa > 0$ -- completeness parameter.\\
Follow the same steps as in Protocol \ref{prot: DISV}.
\end{protocol}
\end{mdframed}
\end{center} 

\vspace{0.5cm}

Note that in the above protocols, the final step which involves applying the optimal channel $\Lambda_{t}$ to $\rho_{t}$ is somewhat fictitious. Indeed, knowing the channels $\Lambda_{i}$ implies solving the optimization $\sup_{\Lambda \in \mcC} F(\Lambda(\rho_{i}),\phi^+)$. This in turn requires knowledge of $\rho_{i}$, which is inaccessible by definition, since $\rho_{i}$ is produced by the untrusted source we wish to certify. Moreover, even if the channels $\Lambda_{i}$ were known, physically implementing them in the lab would go against the device-independent methodology, in which we only have access to observed statistics rather than trusted quantum operations. In \Cref{prot: DISV} and \Cref{prot: DISV_general} however, we are only concerned with the existence of such channels. In this way, when the protocol does not abort, we are guaranteed the existence of an extraction procedure from $\mcC$ which brings the stored state close to the target state.

To further address this point, we include an additional variant which does not include step 5. Specifically, the protocol outputs $\rho_{t}$ directly when it does not abort. We then show in the security proof that $\rho_{t}$ is equivalent to the target state in a well defined sense. 

\vspace{0.5cm}

\begin{center}
\begin{mdframed}[linecolor=black, roundcorner=5pt, skipabove=10pt, skipbelow=10pt, backgroundcolor=white, splittopskip=10pt, splitbottomskip=10pt]
\begin{protocol}[Certification of a state $\varepsilon$-close to the $\phi^+$ state]\label{prot: DISV_general_var}
\noindent\textbf{Parameters}:\\
$n \in \mbN^{+}$ -- number of rounds \\ 
$p_{T}:\{1,...,n\}\to [0,1]$ -- probability distribution of the random variable $T$ (taken to be uniform here) \\ 
$\score_{\sharp} \in [\eta^{\text{Q}}_{\text{min}},\eta^{\text{Q}}_{\text{max}}]$ -- expected value of the Bell functional \eqref{eq:appBop}\\
$\varepsilon \geq 0$ -- closeness parameter\\
$\kappa > 0$ -- completeness parameter.\\
Follow steps 1 to 4 in Protocol \ref{prot: DISV}.

\begin{enumerate}
    \item[5.] If the protocol does not abort, output $\rho_{t}$.
\end{enumerate}
\end{protocol}
\end{mdframed}
\end{center} 
 
\subsection{Sequential setup} \label{app:secSetup}
As discussed above, Protocols \ref{prot: DISV} to \ref{prot: DISV_general_var} assume that all measurements are independent, which may be unrealistic for real devices. To avoid this assumption, one must use $n-1$ isolated devices, which is wasteful and difficult to implement in practice. Alternatively, the aforementioned protocols are also equivalent to a protocol where a single memoryless measurement device is used. There is therefore strong motivation to lift this independence assumption in the security proof. In the following, we detail \Cref{prot: DISV_general_sequencial} which achieves this using the CHSH Bell score:
\begin{equation}
    \chg{p^{\text{win}}} = \frac{1}{4}\sum_{a,b,x,y \in \{0,1\}} w_{a,b,x,y} \, p(a,b|x,y),
\end{equation}
where $w_{a,b,x,y} = 1$ if $a \oplus b = x \cdot y$ and zero otherwise. 

\chg{\begin{remark}
    Up until this point, we have exclusively refered to the \textit{value} of a Bell expression, denoted by $\omega = \langle B \rangle \in [\eta^{\mathrm{Q}}_{\text{min}},\eta^{\mathrm{Q}}_{\text{min}}]$. In particular, this need not correspond to the winning probability of a nonlocal game (i.e., we do not require $\omega \in [0,1]$). When discussing sequential protocols, we will make use of the nonlocal game formulation of the CHSH inequality $B_{\mathrm{CHSH}}$. In this case, we will refer to the CHSH score, denoted $p^{\text{win}} \in [0,1]$, which denotes the winning probability of the CHSH game. Here, the random variables $W_{i}$ take values in $\{0,1\}$ indicating whether round $i$ was lost $(W_{i} = 0)$ or won $(W_{i} = 1)$.
\end{remark}}

As in the parallel setup, the source emits a sequence $\big\{\rho_{i} \in \mcS(\mcH_{Q^{A}_{i}} \otimes \mcH_{Q_{i}^{B}}) \big\}_{i=1}^{n} $, each associated to an optimal channel $\Lambda_{i} \in \mcC$. Instead of $n-1$ devices, we now consider a single measurement device $M$.

\vspace{0.5cm}

\begin{center}
\begin{mdframed}[linecolor=black, roundcorner=5pt, skipabove=10pt, skipbelow=10pt, backgroundcolor=white, splittopskip=10pt, splitbottomskip=10pt]
\begin{protocol}[Certification of a state $\varepsilon$-close to the $\phi^+$ state]\label{prot: DISV_general_sequencial}
\noindent\textbf{Parameters}:\\
$n \in \mbN^{+}$ -- number of rounds \\ 
$p_{T}:\{1,...,n\}\to [0,1]$ -- probability distribution of the random variable $T$ (taken to be uniform here) \\ 
\chg{$p^{\text{win}}_{\sharp} \in [0,1]$} -- expected CHSH score\\
$\varepsilon \geq 0$ -- closeness parameter\\
$\kappa > 0$ -- completeness parameter.\\
\begin{enumerate}
    \item\label{p3_step: 1} Generate a random variable $T$ according to the distribution $p_{T}$. If $T = t$, then store the state $\rho_{t}$. Set $i=1$.  
    \item\label{p3_step: 2} If $i = t-1$, increase $i$ by 2, otherwise, increase $i$ by 1. 
    \item\label{p3_step: 3}  Generate the random bit $X_i$ uniformly, and input to $M$ to obtain the output bit $A_i$. Likewise generate $Y_i$ uniformly and input $Y_{i}$ to $M$, giving the output $B_i$. 
    \item Set $W_{i}= 1$ if $A_i\oplus B_i=X_iY_i$ and $W_{i}= 0$ otherwise. 
    \item Return to Step~2 unless $i = n$ or $i = n- 1$ and $t = n$. 
    \item Calculate the number of rounds in which $W_i=0$ occurred, and abort the protocol if this is larger than $\lfloor(n-1)(1-\chg{p^{\text{win}}_{\sharp}} + \kappa)\rfloor$.
    \item If the protocol does not abort, then apply the optimal LOCC channel $\Lambda_{t}$ to the state $\rho_{t}$ that takes $\rho_{t}$ to a state $\Lambda_{t}(\rho_{t})$. Output $\Lambda_{t}(\rho_{t})$.
\end{enumerate}
\end{protocol}
\end{mdframed}
\end{center}

\vspace{0.5cm}

A graphical description of \Cref{prot: DISV_general_sequencial} can be found in \Cref{fig: DISV_sequencial}. Similarly to the parallel setup, we also include a variant which omits the final extraction step.

\begin{figure}[h!]
    \centering
    \includegraphics[width=0.8 \textwidth]{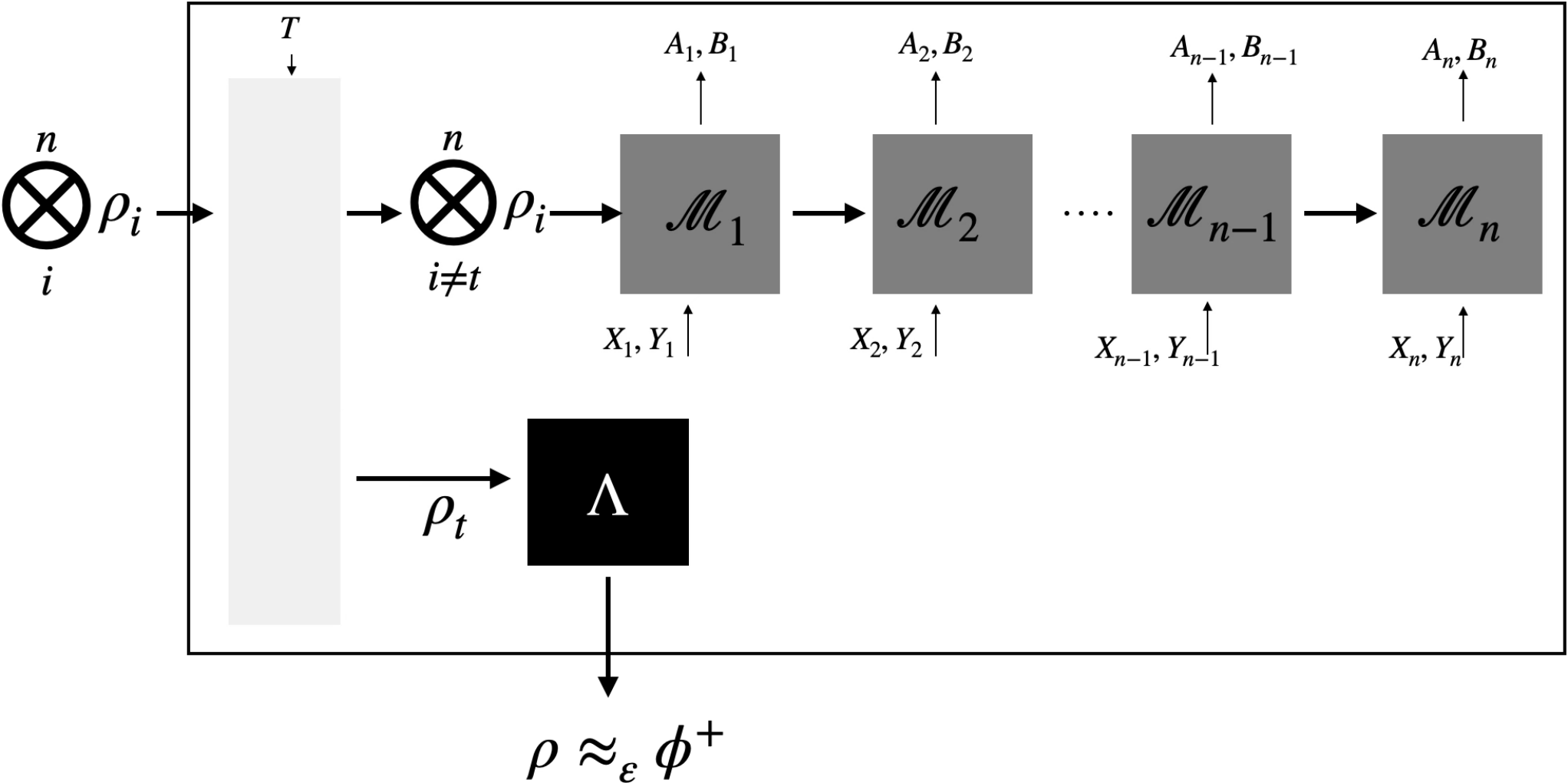}
    \caption{(a) Description of Protocol \ref{prot: DISV_general_sequencial} in terms of its individual components. The bold arrows indicate quantum systems, while the thin arrows represent classical variables. The measurement channels $\mathcal{M}_{i}$ are as described in \Cref{fig: DISV} (b). The measurement devices are uncharacterized and may possess memory; however, they measure independently generated states $\rho_{i}$ one at a time. The gray box, which accepts the input $T$ and $\bigotimes_{i}^n \rho_{i}$, stores the state $\rho_{t}$ for future use and sends the remaining states to the measurement device in a sequential fashion.}
    \label{fig: DISV_sequencial}
\end{figure}

\vspace{0.5cm}

\begin{center}
\begin{mdframed}[linecolor=black, roundcorner=5pt, skipabove=10pt, skipbelow=10pt, backgroundcolor=white, splittopskip=10pt, splitbottomskip=10pt]
\begin{protocol}[Certification of a state $\varepsilon$-close to the $\phi^+$ state]\label{prot: DISV_general_sequencial_var}
\noindent\textbf{Parameters}:\\
$n \in \mbN^{+}$ -- number of rounds \\ 
$p_{T}:\{1,...,n\}\to [0,1]$ -- probability distribution of the random variable $T$ (taken to be uniform here) \\ 
\chg{$p^{\text{win}}_{\sharp} \in [0,1]$} -- expected CHSH score\\
$\varepsilon \geq 0$ -- closeness parameter\\
$\kappa > 0$ -- completeness parameter.\\
Follow steps 1 to 6 in \Cref{prot: DISV_general_sequencial}.
\begin{enumerate}
    \item[7.] If the protocol does not abort, output $\rho_{t}$.
\end{enumerate}
\end{protocol}
\end{mdframed}
\end{center}

\section{Security proof of Protocols \ref{prot: DISV} to \ref{prot: DISV_general_var}}\label{app: security_proofs_par}
In this section, we prove the security of the parallel protocols presented in \Cref{sec:parPro}. As discussed in the main text, we do so according to a composable definition. This involves specifying an ideal protocol and demonstrating that the real implementation of the protocol cannot be distinguished from the ideal one by any hypothetical distinguisher with a probability higher than a pre-agreed threshold (see \cite{Portmann14} for more details).

\begin{remark}
    Protocols \ref{prot: DISV} to \ref{prot: DISV_general_sequencial_var} have been defined with respect to a class of free operations $\mcC$. As discussed in the main text, we prove security for the class of LOCC operations. However, the proofs can be straightforwardly adapted to any other class, such as local operations, which have been frequently studied in the literature~\cite{kaniewski2016analytic,govcanin2022sample}. Provided a lower bound on the LO extractability for the desired Bell inequality is known, this can be directly substituted for the LOCC extractability function used here, for any of the Protocols \ref{prot: DISV} to \ref{prot: DISV_general_sequencial_var}.  \label{rem:LOpro}
\end{remark}

\subsection{Security proof of Protocol \ref{prot: DISV}}

We begin by mathematically describing the real and ideal implementations of the protocol, followed by a proof of their indistinguishability.

\subsubsection{Real protocol}\label{sec:pro1_real}
To best describe the real and ideal protocols, we start by examining their classical-quantum (cq) states at key stages. For the real protocol:
\begin{enumerate}
    \item \label{st: 1} \textbf{Stage 1} (Pre-measurement stage): The user receives a set of (independently generated) states $\bigotimes_{i=1}^{n} \rho_{i}$ and measurement devices $M_{i}$ to which the states $\rho_{i}$ are sent. The variable $T=t$ is sampled to determine which state is kept.
    \item \label{st: 2}\textbf{Stage 2} (Post-measurement stage): For all $i \in \{1,...,n\}\setminus t$ each device $M_{i}$ implements the channel $\mcM_{i}$ detailed in \ifarxiv \cref{sec: measSetup} \else Section 4 of the main text \fi (with the input registers $I_{i}^{A}I_{i}^{B}$ omitted). The user collects a string of length $n-1$, $\mathbf{w} = (w_{1}, w_{2}, \cdots, w_{t-1},w_{t+1}, \cdots , w_{n}) \in \mathcal{W}^{\times (n-1)}$, where  $\mcW = \{\tilde{\gamma}_{x,y},-\tilde{\gamma}_{x,y}\}_{x,y \in \{0,1\}}$ consisting of the measurement outcomes of each round, i.e., $\mathbf{w}$ keeps a record of the ordered list of wins that were measured for each of the $n-1$ rounds. This list is stored in a classical register $\mathbf{W}$. The cq-state of the protocol at this stage then becomes:
    \[
        \rho = \sum_{t=1}^{n} p_{T}(t)\sum_{\mathbf{w} \in \mathcal{W}^{\times (n-1)}} p(\mathbf{w} | t) \, \ketbra{\mathbf{w}}{\mathbf{w}}_{\mathbf{W}} \otimes \ketbra{t}{t}_T \otimes \rho_{t},
    \]
    where $p(\mathbf{w} | t)$ is the conditional probability of generating the string $\mathbf{w}$ given that $T = t$ is observed during the protocol, and $p_{T}(t)$ is the probability that $T = t$.
    \item \label{st: 3}\textbf{Stage 3} (Parameter estimation stage): After collecting statistics, the protocol either aborts or proceeds to the final stage. Its state takes the form
    \[
      \rho = \left( \sum_{t=1}^{n}p_{T}(t) \sum_{\mathbf{w} \in \Omega} p(\mathbf{w} | t) \, \ketbra{\mathbf{w}}{\mathbf{w}}_{\mathbf{W}} \otimes \ketbra{t}{t}_T \otimes \rho_{t} \ot \proj{\Omega} \right) +  (1 - p_{\Omega}) \ketbra{\perp}{\perp},
    \]
    where $\Omega \subset \mcW^{\times (n-1)}$ is the set of observed strings $\mathbf{w}$ which do not cause the protocol to abort, $p_\Omega = \sum_{t=1}^{n}\sum_{\mathbf{w} \in \Omega} p(\mathbf{w}|t) p_{T}(t)$ is the probability of this event and $|\Omega\rangle$ and $|\perp\rangle$ are states indicating whether protocol passes or aborts, respectively.  
    \item \label{st: 4}\textbf{Stage 4} (Final output state): Conditioned on not aborting, the user applies the optimal LOCC channel $\Lambda_{t}$ to the stored state $\rho_{t}$ to obtain the state $\Lambda_{t}(\rho_{t})$, which may be used for future protocols. The final cq-state of the protocol then takes the form
    \[
        \rho_{\mathrm{real}} = \sum_{t=1}^{n} p_{T}(t)\sum_{\mathbf{w} \in \Omega} p(\mathbf{w} | t) \, \ketbra{\mathbf{w}}{\mathbf{w}}_{\mathbf{W}} \otimes \ketbra{t}{t}_T \otimes \Lambda_{t}(\rho_{t}) \ot \proj{\Omega} +  (1 - p_{\Omega})  \proj{\perp}.
    \]
\end{enumerate}
\subsubsection{Ideal protocol}
The ideal protocol differs from the real protocol only in stage \ref{st: 3} and \ref{st: 4}:  
\begin{itemize}
    \item \textbf{Stage 1 and 2:} The ideal protocol runs the real protocol during stage \ref{st: 1} and \ref{st: 2}.
    \item \textbf{Stage 3:} The ideal protocol aborts if the real protocol aborts. If the ideal protocol does not abort, then it replaces the stored state $\rho_{t}$ with the state $\phi^+ \otimes \sigma_{\text{aux}}$, where $\sigma_{\text{aux}} \in \mcS(\mcH_{\text{aux}})$ for some Hilbert space $\mcH_{\text{aux}}$ satisfying $\mbC^{2}\otimes \mbC^{2} \otimes \mcH_{\text{aux}} \cong \mcH_{Q_{t}^{A}}\otimes \mcH_{Q_{t}^{B}}$. The cq-state of the ideal protocol at this stage is: 
    \[
      \rho = \left( \sum_{t=1}^{n} p_{T}(t) \sum_{\mathbf{w} \in \Omega} p(\mathbf{w} | t) \, \ketbra{\mathbf{w}}{\mathbf{w}}_{\mathbf{W}} \otimes \ketbra{t}{t}_T \otimes \phi^+ \otimes \sigma_{\text{aux}} \ot \proj{\Omega} \right) +  (1 - p_{\Omega}) \proj{\perp},
    \]
    \item \textbf{Stage 4:} The ideal protocol throws away the contents of the auxiliary register and outputs $\phi^+$. This will give the final state of the ideal protocol 
    \[
        \rho_{\mathrm{ideal}} = \sum_{t=1}^{n} \sum_{\mathbf{w} \in \Omega} p(\mathbf{w} | t) p_{T}(t) \, \ketbra{\mathbf{w}}{\mathbf{w}}_{\mathbf{W}} \otimes \ketbra{t}{t}_T \otimes \phi^+ \ot \proj{\Omega} +  (1 - p_{\Omega}) \proj{\perp}.
    \]    
\end{itemize}

\subsubsection{Soundness}

Having defined the real and the ideal protocol, recall the definition of soundness discussed in the main text. 
\begin{definition}[Soundness] A DISC protocol is called $\epsilon_{s}$-sound if
\begin{eqnarray}
    \frac{1}{2}|| \rho_{\mathrm{real}} - \rho_{\mathrm{ideal}} ||_1 \leq \epsilon_{s}
\end{eqnarray}
where $\rho_{\mathrm{real}}$ and $\rho_{\mathrm{ideal}}$ are the cq-states obtained after a real and ideal implementation of the protocol and $|| \cdot ||_1$ denotes the trace norm.
\end{definition}

To prove the soundness of \Cref{prot: DISV}, we require Hoeffding's theorem for independent random variables. 
\begin{theorem}[Hoeffding's inequality]\label{thm: Hoeffding's theorem}
Let $X_1, X_2, ..., X_n$ be independent random variables such that $a_i \leq X_i \leq b_i$ for $1 \leq i \leq n$. Then, for any $r > 0$,
\begin{equation}
\begin{aligned}
    \mathbb{P}\left(\sum_{i=1}^n (X_i - \mathbb{E}[X_i])  \geq r\right) &\leq \exp\left(-\frac{2r^2}{\sum_{i=1}^n (b_i - a_i)^2}\right),\ \ \ \text{and}\\
    \mathbb{P}\left(\Big |\sum_{i=1}^n (X_i - \mathbb{E}[X_i]) \Big| \geq r\right) &\leq 2\exp\left(-\frac{2r^2}{\sum_{i=1}^n (b_i - a_i)^2}\right).
\end{aligned}
\end{equation}

\end{theorem}

\begin{lemma}\label{lem:p1sound}
   Protocol \ref{prot: DISV} is $\epsilon_{s}$-sound, where 
   \begin{equation}
   \begin{aligned}
       \epsilon_{s} &= \inf_{\delta > 0} \max \{ a(\delta) , b_{1}(\delta)\}, \\
       a(\delta) &= \exp \left( -   \frac{(n -1)}{\gamma^*} \delta^{2} \right), \\
       \gamma^* &= \max \{ |\gamma_{x,y}| \}_{x,y \in \{0,1\}}, \\
       b_{1}(\delta) &= \sqrt{1 - f \left( \frac{n-1}{n}( \score_{\sharp} -  \kappa - \delta) +  \frac{ \eta^{\mathrm{Q}}_{\mathrm{min}}}{n}\right)}.
    \end{aligned} \label{eq:p1sound}
   \end{equation}
   Here, $f(\omega)$ is any non-decreasing convex function that lower bounds extractability $\Xi_{B}(\omega)$.
\end{lemma}

\begin{proof}
Recall the initial state is denoted by $ \rho_{0} = \bigotimes_{i = 1}^{N} \rho_{i} $, and $\M{M}_{i}$ are the measurement channels (See \Cref{fig: DISV}). We further define $\mu_{i} := \tr(B_{i} \rho_{i})$ as the expected value of the Bell functional from \Cref{eq:appBop}, 
\begin{equation}
    B_{i} = \sum_{x,y \in \{0,1\}} \gamma_{x,y} (A_{x}^{i} \otimes B_{y}^{i}),
\end{equation}
where $A_{x}^{i}$ and $B_{y}^{i}$ are the observables on $Q_{i}^{A}$ and $Q_{i}^{B}$ induced by $\mathcal{M}_{i}$, respectively. 

We begin by recalling the probability that the protocol does not abort, denoted by $p_{\Omega}$, where $\Omega$ is the set of strings $\mathbf{w}$ for which the observed value $\omega$ is at least $\score_{\sharp}$, 
\begin{eqnarray}
    p_{\Omega} = \sum_{t=1 }^{n}p_{T}(t)p(\Omega|t), 
\end{eqnarray}
where $p(\Omega|t) = \sum_{\mathbf{w} \in \Omega} p(\mathbf{w}|t)$ is the probability that the protocol does not abort given that state sent in the $t^{\text{th}}$ round is stored. For convenience, we set $p_{T}(t) = \frac{1}{n}$. Then the trace norm is given by
\begin{equation}
    \begin{aligned}
        \chg{\frac{1}{2}}|| \rho_{\mathrm{real}} - \rho_{\mathrm{ideal}} ||_1  &=  \frac{1}{\chg{2}n}\sum_{t=1}^{n} \sum_{\mathbf{w} \in \Omega} p(\mathbf{w}|t) || \Lambda_{t}(\rho_{t}) - \phi^+ ||_1 \\ 
        &=  \frac{1}{\chg{2}n}\sum_{t=1}^{n} p(\Omega|t) || \Lambda_{t}(\rho_{t}) - \phi^+ ||_1 \\
        &\leq \frac{1}{n}\sum_{t=1}^{n} p(\Omega|t) \sqrt{1 - F(\Lambda_{t}(\rho_{t}) , \phi^+)}.
    \end{aligned}
\end{equation}
For the inequality, \chg{we used the Fuchs van de Graaf inequality}~\cite{Fuchs99} $\chg{\tfrac{1}{2}}||\rho- \sigma ||_{1} \leq \sqrt{1 - F(\rho , \sigma)}$ for two states $\rho,\sigma \in \mcS(\mcH)$. Next we use the relation
\begin{equation}
    F(\Lambda_{t}(\rho_{t}),\phi^+) = \sup_{\Lambda \in \mcC} F(\Lambda(\rho_{t}),\phi^+) \geq \inf_{\rho \in \mcB_{\mu_{t}}} \sup_{\Lambda \in \mcC} F(\Lambda(\rho),\phi^+) = \Xi_{B}(\mu_{t}),
\end{equation}
which follows from the definition of the optimal channels $\Lambda_{t}$, and the fact that $\rho_{t} \in \mcB_{\mu_{t}}$, where
\begin{equation}
    \mcB_{\mu_{i}} = \big\{ \rho \in \mcS(\mcH_{Q^{A}_{i}} \otimes \mcH_{Q^{B}_{i}}) \ : \ \exists \{A_{x}\}_{x}, \, \{B_{y}\}_{y} \ \text{s.t.} \  \tr[B\rho] \geq \mu_{i} \big\}, \label{eq:Bset}
\end{equation}
where $\{A_{x}\}$ and $\{B_{y}\}$ are understood to be sets of two-outcome observables on $\mcH_{Q_{i}^{A}}$ and $\mcH_{Q_{i}^{B}}$, respectively, and $B$ is the Bell operator \eqref{eq:appBop} constructed from $A_{x}$ and $B_{y}$. This allows us to write
\begin{equation}
    \frac{1}{n}\sum_{t=1}^{n} p(\Omega|t) \sqrt{1 - F(\Lambda_{t}(\rho_{t}) , \phi^+)} \leq \frac{1}{n}\sum_{t=1}^{n} p(\Omega|t) \sqrt{1 - \Xi_{B}(\mu_{t})}. \label{eq:extrRel}
\end{equation}

Now, by noting $\Xi_{\mu_{t}}\geq 0$ we can bound the trace norm in terms of the abort probability,
\begin{eqnarray} 
\chg{\frac{1}{2}}|| \rho_{\mathrm{real}} - \rho_{\mathrm{ideal}} ||_1 &\leq& \frac{1}{n}\sum_{t=1}^{n} p(\Omega| t) = p_{\Omega}. 
\end{eqnarray} 
Alternatively, by noting $p(\Omega|t) \leq 1$, we can bound the trace norm in terms of the average extractability,
\begin{eqnarray} 
\chg{\frac{1}{2}} || \rho_{\mathrm{real}} - \rho_{\mathrm{ideal}} ||_1 &\leq& \sum_{t=1}^{n} \frac{1}{n} \sqrt{1 - \Xi_{B}(\mu_{t})} \leq \sqrt{1 - \frac{1}{n}\sum_{t=1}^{n} \Xi_{B}(\mu_{t})}, 
\end{eqnarray} 
where for the second inequality we used the concavity of the square root.

Based on the above, we consider two cases, and introduce a free parameter $\delta > 0$.

\vspace{0.2cm}

\noindent \textbf{Case 1}: $\sum_{i = 1}^{n} \frac{\mu_{i}}{n} - \frac{ \eta^{\text{Q}}_{\text{min}}}{n}\leq  \frac{n- 1}{n} (\omega_{\sharp} - \kappa - \delta)$. That is, the average value of $\mu_i$ is less than $(\omega_{\sharp} - \kappa) \frac{n-1}{n} +  \frac{ \eta^{\text{Q}}_{\text{min}}}{n}$ (recall $\eta^{\text{Q}}_{\text{min}}$ is the minimum quantum value of the Bell expression \eqref{eq:appBop}). If this this is the case, then we have that 
\begin{equation}
    \sum_{i \neq t}^{n} \frac{\mu_{i}}{n-1}  \leq \sum_{i \neq t}^{n} \frac{\mu_{i}}{n-1} + \frac{\mu_{t}-\eta^{\text{Q}}_{\text{min}}}{n-1} = \sum_{i=1}^{n} \frac{\mu_{i}}{n -1} - \frac{ \eta^{\text{Q}}_{\text{min}}}{n-1} \leq \omega_{\sharp} - \kappa - \delta, \label{eq:probUb}
\end{equation}
where we used the fact that $\mu_{t} \geq \eta^{\text{Q}}_{\text{min}}$ for the first inequality. The probability that the protocol does not abort given $T=t$ is given by $p(\Omega | t) = \mathbb{P}\left( \sum_{i \neq t}^{n} \frac{W_{i}}{n-1} \geq \omega_{\sharp} - \kappa \right)$. We can now apply \Cref{thm: Hoeffding's theorem}, by choosing $X_{i} = W_{i}$ , $\mathbb{E}[X_{i}] = \mu_{i}$, $r = (n- 1) \delta$, $b_i =  \max \{ |\gamma_{xy}| \}$ and $a_{i} =  - \max \{ |\gamma_{xy}| \}$ to obtain the following bound:
\begin{equation}
    \begin{aligned}
        \mathbb{P}\Bigg( \sum_{i \neq t}^{n} \frac{W_{i}}{n-1} \geq \omega_{\sharp} - \kappa \Bigg) &= \mathbb{P}\Bigg( \sum_{i \neq t}^{n} \frac{W_{i}}{n-1} \geq \omega_{\sharp} - \kappa - \delta + \frac{r}{n-1} \Bigg) \\
        &\leq \mathbb{P}\Bigg( \sum_{i \neq t}^{n} \frac{W_{i}}{n-1} \geq \sum_{i \neq t}^{n}\frac{\mu_{i}}{n-1} + \frac{r}{n-1} \Bigg) \\
        &= \mathbb{P}\Bigg( \sum_{i \neq t}^{n} (W_{i} - \mu_{i}) \geq r \Bigg)\\
        &\leq \exp \left( -   \frac{(n -1)}{\gamma^*} \delta^{2} \right) =: a(\delta),
    \end{aligned}
\end{equation}
where for the first inequality we applied \Cref{eq:probUb}, and for the second we applied \Cref{thm: Hoeffding's theorem}. Since the calculations are identical for all values of $t$, we obtain a bound on $p_{\Omega}$, 
\begin{eqnarray}
    p_{\Omega} \leq a(\delta).
\end{eqnarray}
Thus, in this case, we have that $\chg{\frac{1}{2}}|| \rho_{\mathrm{real}} - \rho_{\mathrm{ideal}} ||_{1} \leq p_{\Omega} \leq a(\delta)$. 

\vspace{0.2cm}

\noindent \textbf{Case 2}: $\sum_{i= 1}^{n} \frac{\mu_{i}}{n} - \frac{ \eta^{\text{Q}}_{\text{min}}}{n} > \frac{n-1}{n} (\score_{\sharp} - \kappa - \delta)$.  Let $f(\omega)$ be any non-decreasing convex function that lower bounds extractability $\Xi_{B}(\omega)$. Then
\begin{equation}
    \frac{1}{n}\sum_{i=1}^{n} \Xi(\mu_{i}) \geq \frac{1}{n}\sum_{i=1}^{n} f(\mu_{i}) \geq f\Bigg(\frac{1}{n}\sum_{i=1}^{n} \mu_{i}\Bigg) \geq f\Big(\frac{n-1}{n}( \omega_{\sharp} - \kappa - \delta) +  \frac{ \eta^{\text{Q}}_{\text{min}}}{n}\Big).    
\end{equation}
We thus have that 
\begin{equation}
        \chg{\frac{1}{2}} || \rho_{\mathrm{real}} - \rho_{\mathrm{ideal}} ||_1  
     \leq  \sqrt{1 - \frac{1}{n}\sum_{t=1}^{n} \Xi_{B}(\mu_{t})} \leq   \sqrt{1 -   f\Big(\frac{n-1}{n}( \omega_{\sharp} - \kappa - \delta)+  \frac{ \eta^{\text{Q}}_{\text{min}}}{n}\Big)} =: b_{1}(\delta), 
\end{equation}
completing the proof. 
\end{proof}

\subsubsection{Completeness} 
\begin{definition}[Completeness]
    A DISC protocol is called $\epsilon_{c}$-complete if there exists an honest implementation such that $p_{\Omega} \geq 1 - \epsilon_{c}$.
\end{definition}
\begin{lemma}
    Protocol \ref{prot: DISV} is $\epsilon_{c}$-complete, where
    \begin{equation}
        \epsilon_{c} = 2\exp\Bigg(\frac{n-1}{\gamma^*}\kappa^{2}\Bigg). \label{eq:p1complete}
    \end{equation} \label{lem:p1complete}
\end{lemma}
\begin{proof}
    Consider an honest implementation for which the variables $W_{1},...,W_{n}$ are i.i.d. random variables with $\mathbb{E}[W_{i}] = \omega_{\sharp}$. Then
    \begin{equation}
    \begin{aligned}
        p(\Omega|t) &= 1 - \mathbb{P}\Bigg( \sum_{i \neq t}^{n} \frac{W_{i}}{n-1} < \omega_{\sharp} - \kappa \Bigg) \\
        &= 1 - \mathbb{P}\Bigg( -\sum_{i \neq t}^{n} (W_{i} - \mathbb{E}[W_{i}]) > (n-1)\kappa \Bigg) \\
        & \geq 1 - \mathbb{P}\Bigg( \Big|\sum_{i \neq t}^{n} (W_{i} - \mathbb{E}[W_{i}]) \Big| \geq (n-1)\kappa \Bigg)\\
        & \geq 1 - 2\exp\Bigg(\frac{n-1}{\gamma^*}\kappa^{2}\Bigg),
    \end{aligned}
    \end{equation}
    where we applied \Cref{thm: Hoeffding's theorem} to obtain the final inequality. The claim follows from the fact that $p_{\Omega} = \sum_{t=1}^{n}\frac{p(\Omega|t)}{n}$.
\end{proof}
Combining soundness and completeness, we arrive at our composable security definition for a DISC protocol,
\begin{definition}[Security]
    A DISC is $(\epsilon_{s},\epsilon_{c})$-secure if it is $\epsilon_{s}$-sound and $\epsilon_{c}$-complete. 
\end{definition}
\noindent It immediately follows from \Cref{lem:p1sound,lem:p1complete} that \Cref{prot: DISV} is $(\epsilon_{s},\epsilon_{c})$-secure, for $\epsilon_{s}$ and $\epsilon_{c}$ given by \Cref{eq:p1sound} and \Cref{eq:p1complete}, respectively. 

\subsection{Security proof of Protocol \ref{prot: DISV_general}} \label{sec:DISV_general}
In \Cref{prot: DISV_general}, we relax the certification goal of the maximally entangled state to a state $\varepsilon$-close to the maximally entangled. The security proof is appropriately modified in the following.  

\subsubsection{Real protocol}
The real protocol is identical to that of \Cref{prot: DISV}, outlined in \Cref{sec:pro1_real}.

\subsubsection{Ideal protocol}
The ideal protocol is modified as follows. We will need the following definition of the Heaviside step function,
\begin{eqnarray}
        \Theta(x) := \begin{cases} 1 & \quad \text{if}  \quad x > 0 \\ 
        0 & \quad \text{otherwise}. 
    \end{cases}
\end{eqnarray}

\begin{itemize}
    \item \textbf{Stage 1 and 2:} The ideal protocol runs the real protocol during stages \ref{st: 1} and \ref{st: 2}.
    \item \textbf{Stage 3:} The ideal protocol aborts if the real protocol aborts. If the ideal protocol does not abort, then it  replaces the stored state $\rho_{t}$ with the state $[(1-\lambda_{t})\Lambda_{t}(\rho_{t}) + \lambda_{t}\phi^+] \otimes \sigma_{\text{aux}}$ for any real number $\lambda_{t} \in (0,1)$ satisfying 
    \begin{equation}
        \lambda_{t} \leq \left( 1 - \frac{\varepsilon}{\sqrt{1 - \Xi_{B}(\mu_{t})}} \right) \Theta\left( \sqrt{1 - \Xi_{B}(\mu_{t})} - \varepsilon\right). \label{eq:lamDef}
    \end{equation}
    In the above, $\sigma_{\text{aux}} \in \mcS(\mcH_{\text{aux}})$ is an auxiliary state on a Hilbert space $\mcH_{\text{aux}}$ satisfying $\mbC^{2}\otimes \mbC^{2} \otimes \mcH_{\text{aux}} \cong \mcH_{Q_{t}^{A}}\otimes \mcH_{Q_{t}^{B}}$. The cq-state of the ideal protocol at this stage is: 
    \[
      \rho = \left( \sum_{t=1}^{n} p_{T}(t) \sum_{\mathbf{w} \in \Omega} p(\mathbf{w} | t) \, \ketbra{\mathbf{w}}{\mathbf{w}}_{\mathbf{W}} \otimes \ketbra{t}{t}_T \otimes [(1-\lambda_{t})\Lambda_{t}(\rho_{t}) + \lambda_{t}\phi^+] \otimes \sigma_{\text{aux}} \ot \proj{\Omega} \right) +  (1 - p_{\Omega}) \proj{\perp}.
    \]
    \item \textbf{Stage 4:} The ideal protocol throws away the contents of the auxiliary register and outputs $(1-\lambda_{t})\Lambda_{t}(\rho_{t}) + \lambda_{t}\phi^+$. This will give the final state of the ideal protocol 
    \[
        \rho_{\mathrm{ideal}} = \sum_{t=1}^{n} \sum_{\mathbf{w} \in \Omega} p(\mathbf{w} | t) p_{T}(t) \ketbra{\mathbf{w}}{\mathbf{w}}_{\mathbf{w}} \otimes \ketbra{t}{t}_T \otimes [(1-\lambda_{t})\Lambda_{t}(\rho_{t}) + \lambda_{t}\phi^+] \ot \proj{\Omega} +  (1 - p_{\Omega}) \proj{\perp}.
    \]    
\end{itemize}

\subsubsection{Soundness}
\begin{lemma} \label{lem:p2sound}
    Protocol \ref{prot: DISV_general} is $\epsilon_{s}$-sound, where 
    \begin{equation}
        \begin{aligned}
            \epsilon_{s} &= \inf_{\delta > 0} \max \{ a(\delta),b_{2}(\delta)\},\\
            b_{2}(\delta) &= G_{\varepsilon}\left( \frac{n-1}{n} (\score_{\sharp} - \kappa - \delta) + \frac{\eta^{\mathrm{Q}}_{\mathrm{min}}}{n}\right),
        \end{aligned} \label{eq:eps2}
    \end{equation}
    $a(\delta)$ is defined in \Cref{lem:p1sound} and $G_{\varepsilon}(\score)$ is any non-increasing concave function that upper bounds the function $\Theta(\sqrt{1 - \Xi(\omega)}  - \varepsilon) (\sqrt{1 - \Xi(\omega)} - \varepsilon)$.
\end{lemma}
\begin{proof}
    The proof proceeds similarly to that of \Cref{lem:p1sound}, with some key differences. We begin by writing
    \begin{equation}
        \begin{aligned}
            || \rho_{\mathrm{real}} - \rho_{\mathrm{ideal}} ||_1  &=  \frac{1}{n}\sum_{t=1}^{n} \sum_{\mathbf{w} \in \Omega} p(\mathbf{w}|t)  || \Lambda(\rho_{t}) - (1 - \lambda_{t})\Lambda_{t}(\rho_{t}) - \lambda_{t} \phi^+  ||_1 \\
            &=\frac{1}{n}\sum_{t=1}^{n}  p(\Omega|t)  \lambda_{t}|| \Lambda_{t}(\rho_{t}) -  \phi^+  ||_1\\
            &\leq\frac{\chg{2}}{n}\sum_{t=1}^{n} p(\Omega|t)  \lambda_{t}\sqrt{1 - F(\Lambda_{t}(\rho_{t}),\phi^+)}\\
            &\leq \frac{\chg{2}}{n}\sum_{t=1}^{n} p(\Omega|t)  \left( 1 - \frac{\varepsilon}{\sqrt{1 - \Xi_{B}(\mu_{t})}} \right) \Theta\left( \sqrt{1 - \Xi_{B}(\mu_{t})} - \varepsilon\right)\sqrt{1 - \Xi_{B}(\mu_{t})} \\
            &\leq \frac{\chg{2}}{n}\sum_{t=1}^{n} p(\Omega|t)  G_{\varepsilon}(\mu_{t}).
        \end{aligned}
    \end{equation}
    For the first inequality we used the relationship between the trace distance and fidelity, for the second we used \Cref{eq:extrRel} and \Cref{eq:lamDef}, and for the third we introduced the function $G_{\varepsilon}(\omega)$ as described in the theorem statement. 

    The proof now proceeds in two cases, and we introduce a free parameter $\delta > 0$.

    \vspace{0.2cm}

    \noindent \textbf{Case 1:} $\sum_{i=1}^{n}\frac{\mu_{i}}{n} - \frac{\eta^{\mathrm{Q}}_{\mathrm{min}}}{n} \leq \frac{n-1}{n}(\omega_{\sharp} - \kappa - \delta)$. The proof proceeds identically to Case 1 in the proof of \Cref{lem:p1sound}. 

    \vspace{0.2cm} 

    \noindent \textbf{Case 2:} $\sum_{i=1}^{n}\frac{\mu_{i}}{n} -\frac{\eta^{\mathrm{Q}}_{\mathrm{min}}}{n}> \frac{n-1}{n}(\omega_{\sharp} - \kappa - \delta)$. Following Case 2 in the proof of \Cref{lem:p1sound}, we use the bound $p(\Omega|t)\leq 1$ and the concavity of $G_{\varepsilon}(\omega)$ to obtain
    \begin{equation}
       \chg{\frac{1}{2}} \| \rho_{\text{real}} - \rho_{\text{ideal}} \|_{1} \leq \frac{1}{n}\sum_{t=1}G_{\varepsilon}(\mu_{t}) \leq G_{\varepsilon}\Big( \frac{n-1}{n}(\omega_{\sharp} - \kappa - \delta)+ \frac{\eta^{\mathrm{Q}}_{\mathrm{min}}}{n}\Big) =: b_{2}(\delta).
    \end{equation}
    This completes the proof.
\end{proof}

\subsubsection{Completeness} 
Note that \Cref{lem:p1complete} also applies to \Cref{prot: DISV_general}, resulting in $(\epsilon_{s},\epsilon_{c})$-security, for $\epsilon_{s}$ and $\epsilon_{c}$ given by \Cref{eq:eps2} and \Cref{eq:p1complete}, respectively.

\subsection{Security proof of \Cref{prot: DISV_general_var}}

In this subsection, we prove the security of \Cref{prot: DISV_general_var}, which does not require the user to apply the optimal channel $\Lambda_{t}$ to the stored state. Instead, the user outputs the state $\rho_{t}$ directly, and the certification of $\rho_{t}$ is described by its closeness to a companion state. 
\begin{definition}[Companion state]
    Let $\rho \in \mcS(\mcH_{Q^{A}} \otimes \mcH_{Q^{B}})$, $\mu = \tr[B\rho]$ for a given Bell operator $B$, $\ket{\phi} \in \mcH_{\hat{Q}^{A}} \otimes \mcH_{\hat{Q}^{B}}$ be a target state and $\mcC$ be a class of free operations. A state $\sigma \in \mcS(\mcH_{\hat{Q}^{A}} \otimes \mcH_{\hat{Q}^{B}} \otimes \mcH_{Q^{A}} \otimes \mcH_{Q^{B}})$ is a companion state of $\rho$ if it is of the form
    \begin{equation}
        \sigma = U^{\dagger}(\phi \otimes \sigma_{\text{aux}})U, \label{eq:comp}
    \end{equation}
    where $\sigma_{\text{aux}} \in \mcS(\mcH_{Q^{A}} \otimes \mcH_{Q^{B}} )$, $U$ is a unitary operator on $\mcH_{\hat{Q}^{A}} \otimes \mcH_{\hat{Q}^{B}}\otimes \mcH_{Q^{A}} \otimes \mcH_{Q^{B}}$ satisfying $\tr_{Q^{A}Q^{B}}[U(\ketbra{00}{00} \otimes \tau)U^{\dagger}] = \Lambda(\tau)$ for all $\tau \in \mcS(\mcH_{Q^{A}} \otimes \mcH_{Q^{B}})$ and a channel $\Lambda \in \mcC$, and 
    \begin{equation}
        F(\ketbra{00}{00} \otimes \rho , \sigma) \geq \Xi_{B}(\mu). \label{eq:comp2}
    \end{equation} \label{def:comp}
\end{definition}
The existence of a companion state $\sigma$ for a given state $\rho$ is significant, since implies the following chain of inequalities hold,
\begin{equation}
    \begin{aligned}
        \Xi_{B}(\mu) \leq F(\ketbra{00}{00} \otimes \rho , \sigma) = F(U(\ketbra{00}{00} \otimes \rho)U^{\dagger},\phi \otimes \sigma_{\text{aux}}) \leq F(\Lambda(\rho),\phi),
    \end{aligned}
\end{equation}
where we used the fact that the fidelity is invariant under unitaries, and does not contract under the partial trace. In words, there exists a channel $\Lambda$ which extracts the target state $\phi$ with fidelity at least $\Xi_{B}(\mu)$. The set of allowable unitaries $U$ is given by the Naimark dilation of every channel $\Lambda \in \mcC$. For example, if $\mcC$ corresponds to the set of local channels $\Lambda_{A} \otimes \Lambda_{B}$, $U = U_{A} \otimes U_{B}$ is a local unitary. The following lemma guarantees the existence of a companion state for every state $\rho$. 

\begin{lemma}
    Let $\rho \in \mcS(\mcH_{Q^{A}} \otimes \mcH_{Q^{B}})$, $\mu = \tr[B\rho]$ for a Bell operator $B$, $\ket{\phi} \in \mcH_{\hat{Q}^{A}} \otimes \mcH_{\hat{Q}^{B}}$ be a target state and $\mcC$ be a class of free operations. Then there always exists a companion state to $\rho$ according to \Cref{def:comp}.
\end{lemma}
\begin{proof}
    We prove the above though an explicit construction. Let 
    \begin{equation}
        \Lambda^* = \text{arg}\,\text{max}\Big\{ F(\Lambda(\rho),\phi) \ : \ \Lambda : \mcS(\mcH_{Q^{A}} \otimes \mcH_{Q^{B}}) \to \mcS(\mcH_{\hat{Q}^{A}} \otimes \mcH_{\hat{Q}^{B}}), \ \Lambda \in \mcC \Big\}.
    \end{equation}
    We can always describe $\Lambda^*$ by the action of an isometry $V: \mcH_{Q^{A}} \otimes \mcH_{Q^{B}} \to \mcH_{\hat{Q}^{A}} \otimes \mcH_{\hat{Q}^{B}} \otimes \mcH_{Q^{A}} \otimes \mcH_{Q^{B}}$ followed by a partial trace over $Q^{A}Q^{B}$, $\Lambda^{*}(\tau) = \tr_{Q^{A}Q^{B}}[V \tau V^{\dagger}]$. We can further assume that $V = U(\ket{00}\otimes \id_{Q^{A}}\otimes \id_{Q^{B}})$ where $U$ is a unitary on $\mcH_{\hat{Q}^{A}} \otimes \mcH_{\hat{Q}^{B}} \otimes \mcH_{Q^{A}} \otimes \mcH_{Q^{B}}$. Let $\ket{\Psi} \in \mcH_{Q^{A}} \otimes \mcH_{Q^{B}} \otimes \mcH_{E}$ be any purification of $\rho$. Note that the state $\ket{\Psi'} = (V \otimes \id_{E})\ket{\Psi}$ is a purification of $\Lambda^*(\rho)$. To see this, observe
\begin{equation}
    \tr_{\hat{Q}^{A}\hat{Q}^{B}E}[\Psi'] = \tr_{\hat{Q}^{A}\hat{Q}^{B}}\big[ \tr_{E}[(V \otimes \id_{E})\Psi(V^{\dagger} \otimes \id_{E}] \big] = \tr_{\hat{Q}^{A}\hat{Q}^{B}}[ V \rho V^{\dagger} ] = \Lambda^*(\rho).
\end{equation}
Recall, Uhlmann's theorem (see, e.g.,~\cite[Theorem 9.2.1]{Wilde2013}) relates the fidelity of two states to the maximum overlap between their purifications,
\begin{equation}
    F(\rho,\sigma) = \max_{\ket{\Psi_{\sigma}}} | \braket{\Psi_{\rho}}{\Psi_{\sigma}} |^{2},
\end{equation}
where $\Psi_{\rho}$ is any purification of $\rho$ and the maximization is taken over all purifications $\Psi_{\sigma}$ of $\sigma$. Applying this to $F(\Lambda^*(\rho),\phi)$, we find
\begin{equation}
    F(\Lambda^*(\rho),\phi) = \max_{\ket{\psi'}} | \bra{\Psi'} \big( \ket{\phi} \otimes \ket{\psi'}\big) |^{2} = \max_{\ket{\psi'}} F(\Psi',\phi \otimes \psi'), \label{eq:fid1}
\end{equation}
where the maximization is taken over all states $\ket{\psi'} \in \mcH_{Q^A} \otimes \mcH_{Q^B} \otimes \mcH_{E}$, and we used the fact that the purification of any pure state must be separable. Note that we can write
\begin{equation}
    \Psi' = (V \otimes \id_{E})\Psi(V^{\dagger} \otimes \id_{E}) = (U \otimes \id_{E})(\ketbra{00}{00} \otimes  \Psi)(U^{\dagger} \otimes \id_{E}).
\end{equation}
This implies 
\begin{equation}
\begin{aligned}
    F(\Psi',\phi \otimes \psi') &= F\Big((U \otimes \id_{E})(\ketbra{00}{00} \otimes  \Psi)(U^{\dagger} \otimes \id_{E}) , \phi \otimes \psi'\Big) \\
    &= F\Big(\ketbra{00}{00} \otimes  \Psi , (U^{\dagger} \otimes \id_{E})(\phi \otimes \psi')(U \otimes \id_{E})\Big) \\
    &\leq F\Big(\ketbra{00}{00} \otimes \rho , \tr_{E}\big[(U^{\dagger} \otimes \id_{E})(\phi \otimes \psi')(U \otimes \id_{E})\big]\Big),
\end{aligned} \label{eq:fid2}
\end{equation}
where we used the fact that the fidelity is invariant under unitary operations, followed by its monotonicity under the partial trace. Let 
\begin{equation}
    \sigma = \tr_{E}\big[(U^{\dagger} \otimes \id_{E})(\phi \otimes \psi^*)(U \otimes \id_{E})\big] = U^{\dagger}(\phi \otimes \tr_{E}[\psi^*])U,
\end{equation}
where $\psi^*$ achieves the optimal value of the maximization $\max_{\ket{\psi'}} F(\Psi',\phi \otimes \psi')$. Then we see $\sigma$ is of the form in \Cref{eq:comp}. Furthermore,  
\begin{equation}
     F(\ketbra{00}{00} \otimes \rho,\sigma) \geq F(\Psi',\phi \otimes \psi') = F(\Lambda^*(\rho),\phi) = \sup_{\Lambda \in \mcC} F(\Lambda(\rho),\phi) \geq \inf_{\rho' \in \mcB_{\mu}} \sup_{\Lambda \in \mcC} F(\Lambda(\rho'),\phi) = \Xi_{B}(\mu), \label{eq:compExt}
\end{equation}
where the first inequality follows from \Cref{eq:fid2}, the first equality follows from \Cref{eq:fid1}, the second equality follows from the definition of $\Lambda^*$ and the second inequality follows from the fact that $\rho \in \mcB_{\mu}$, where $\mcB_{\mu}$ is defined analogously to \Cref{eq:Bset}. Thus, $\sigma$ satisfies \Cref{def:comp}, completing the proof. 
\end{proof}
Having defined a companion state, we can now prove the security of \Cref{prot: DISV_general_sequencial_var}.

\subsubsection{Real protocol}
The real protocol is identical to the real protocol described in \Cref{sec:pro1_real}, except stage \ref{st: 4} is omitted, and it outputs the state $\ketbra{00}{00} \otimes \rho_{t}$. We therefore see the final cq-state takes the form
\begin{equation}
    \rho_{\text{real}} = \left( \sum_{t=1}^{n}p_{T}(t) \sum_{\mathbf{w} \in \Omega} p(\mathbf{w} | t) \, \ketbra{\mathbf{w}}{\mathbf{w}}_{\mathbf{W}} \otimes \ketbra{t}{t}_T \otimes \ketbra{00}{00} \otimes \rho_{t} \ot \proj{\Omega} \right) +  (1 - p_{\Omega}) \ketbra{\perp}{\perp}.
\end{equation}

\subsubsection{Ideal protocol}
The ideal protocol is modified as follows.
\begin{itemize}
    \item \textbf{Stage 1 and 2:} The ideal protocol runs the real protocol during stages \ref{st: 1} and \ref{st: 2}.
    \item \textbf{Stage 3:} The ideal protocol aborts if the real protocol aborts. If the ideal protocol does not abort, then it  replaces the stored state $\rho_{t}$ with the state $(1-\lambda_{t})\ketbra{00}{00}\otimes \rho_{t} + \lambda_{t}\sigma_{t}$, where $\lambda_{t} \in (0,1)$ satisfies \Cref{eq:lamDef} and $\sigma_{t}$ is a companion state to $\rho_{t}$. The final cq-state of the ideal protocol is given by 
    \[
      \rho_{\text{ideal}} = \left( \sum_{t=1}^{n} p_{T}(t) \sum_{\mathbf{w} \in \Omega} p(\mathbf{w} | t) \, \ketbra{\mathbf{w}}{\mathbf{w}}_{\mathbf{W}} \otimes \ketbra{t}{t}_T \otimes [(1-\lambda_{t})\ketbra{00}{00}\otimes\rho_{t} + \lambda_{t}\sigma_{t}] \ot \proj{\Omega} \right) +  (1 - p_{\Omega}) \proj{\perp}.
    \]
\end{itemize}

\subsubsection{Soundness}
\begin{lemma}
    Protocol \ref{prot: DISV_general_var} is $\epsilon_{s}$-sound, where $\epsilon_{s}$ is given by \Cref{eq:eps2}. \label{lem:p3sound}
\end{lemma}
\begin{proof}
    The proof follows that of \Cref{lem:p2sound}, expect with $F(\Lambda_{t}(\rho_{t}),\phi^+)$ replaced with $F(\ketbra{00}{00} \otimes \rho_{t},\sigma_{t})$, i.e.,
    \begin{equation}
        \| \rho_{\text{real}} - \rho_{\text{ideal}} \|_{1} \leq \frac{\chg{2}}{n}\sum_{t=1}^{n} p(\Omega|t)  \lambda_{t}\sqrt{1 - F(\ketbra{00}{00} \otimes \rho_{t},\sigma_{t})}.
    \end{equation}
    Using the property \ref{eq:comp2} of $\sigma_{t}$, we see $F(\ketbra{00}{00} \otimes \rho_{t},\sigma_{t}) \geq \Xi(\mu_{t})$, hence
    \begin{equation}
        \| \rho_{\text{real}} - \rho_{\text{ideal}} \|_{1} \leq \frac{\chg{2}}{n}\sum_{t=1}^{n} p(\Omega|t)  \lambda_{t}\sqrt{1 - \Xi_{B}(\mu_{t})}.
    \end{equation}
    The proof then proceeds identically to that of \Cref{lem:p2sound}.
\end{proof}

\subsubsection{Completeness} 
Lemma \ref{lem:p1complete} applies to \Cref{prot: DISV_general_var}, resulting in $(\epsilon_{s},\epsilon_{c})$-security, for $\epsilon_{s}$ and $\epsilon_{c}$ given by \Cref{eq:eps2} and \Cref{eq:p1complete}, respectively.

\begin{remark}
We note that the security statements presented in Lemmas \ref{lem:p1sound}, \ref{lem:p2sound} and \ref{lem:p3sound} are not tight. To derive an upper bound for the trace norm, we have set \( p(\Omega|i) \) to 1, which may not be optimal. In principle, it should be possible to bound this in terms of the average Bell value \( \frac{1}{n-1} \sum_{j \neq i} \mu_{j} \), resulting in a more complex expression for the soundness parameter. We leave this refinement for future work.
\end{remark}
\begin{remark}
As the extractability function  can be assumed to be  convex  (if it isn't, one can always take the convex lower bound) , the function $G_{\varepsilon}(\score)$ is automatically concave for $\varepsilon = 0$. For $\varepsilon > 0$, the optimal way to define the function $G_{\varepsilon}(\score)$ is by
\begin{eqnarray}
    G_{\varepsilon}(\score) = - \mathrm{conenv}\left( -\Theta(\sqrt{1 - \Xi_{B}(\score)} - \varepsilon) (\sqrt{1 - \Xi_{B}(\score)} - \varepsilon) \right), 
\end{eqnarray}
where $\mathrm{conenv}$ is the convex envelope (convex lower bound). Computing the convex envelopes of functions on $\mathbb{R}$ is relatively straightforward (see, for example,~\cite[Section 8.10]{Bhavsar_thesis}), but for functions on $\mathbb{R}^{n}$, it can be difficult in general. For $n = 2, 3$, there exist fast algorithms to compute this function (see \cite{Contento2015}). We plot the function $G_{\varepsilon}(\score)$ for different values of $\varepsilon$ in \Cref{fig: G_epsilon_supp} also given in the main text and provided here for completeness. \label{rem:g_fun}
\end{remark}
\begin{figure}[h!]
    \includegraphics[width=0.42\textwidth]{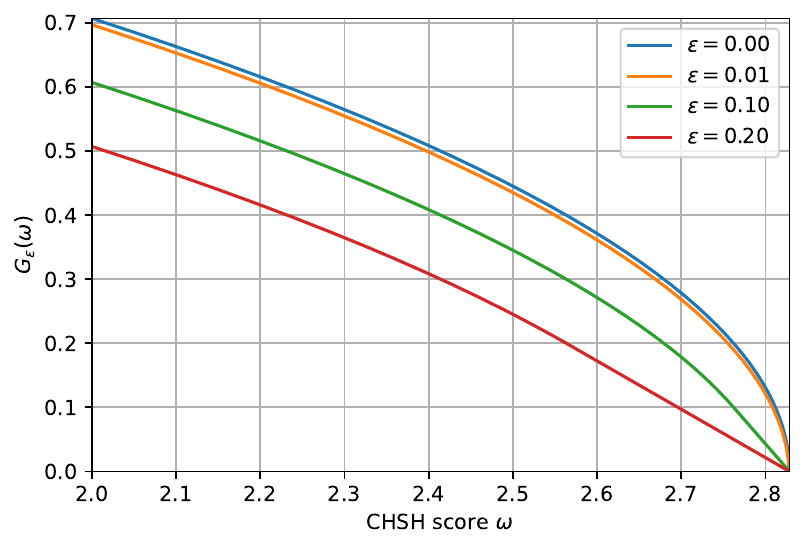}
        \caption{Graph of $G_{\varepsilon}(\omega)$ for different values of $\varepsilon$, using LOCC extractability.}
        \label{fig: G_epsilon_supp}
\end{figure}

\section{Security proof of Protocols \ref{prot: DISV_general_sequencial} and \ref{prot: DISV_general_sequencial_var}}\label{app: security_proofs_sec}

\subsection{Modeling the sequential process} \label{app:seqProc}

Before proving security of the sequential protocols, we analyze the structure of the channels $\mcN_{i} = \mcN_{i}^{A} \otimes \mcN_{i}^{B}$ defined in \ifarxiv \cref{sec: measSetup}. \else Section 4 of the main text. \fi Specifically, $\mcN_{i}^{A}:O_{i-1}^{A}Q_{i}^{A} \to A_{i}X_{i}O_{i}^{A}$ and $\mcN_{i}^{B}:O_{i-1}^{B}Q_{i}^{B} \to B_{i}Y_{i}O_{i}^{B}$, where the systems $O_{i}^{A}$ and $O_{i}^{B}$ model the internal memory of each device. It is important to emphasize that the channels in each round act on the state generated for round $i$ only, as well as the state held in the device's memory. Mathematically, this independence implies that the input state to $M_{i}^{A}$ is given by $\tr_{Q_{i}^{B}} (\rho_{i}) \otimes \sigma_{O_{i-1}}$, where $\sigma_{O_{i-1}}$ represents the quantum state stored in the device's memory. We then have, writing $\tau_{Q_{i}^{A}} = \tr_{Q_{i}^{B}} (\rho_{i})$,  
\begin{equation}
 \mcN^{A}_{i} (\tau_{Q_{i}^{A}} \ot \sigma_{O_{i-1}^{A}}) = \sum_{a,x \in \{0,1\}} p(x) \, \proj{x}_{X_i} \ot \proj{a}_{A_i} \ot  \mcM_{i}^{a|x}( \tau_{Q_{i}^{A}} \ot \sigma_{O_{i-1}^{A}}),
\end{equation}
where $\mcM_{i}^{a|x}:Q_{i}^{A}O_{i-1}^{A} \to O_{i}^{A}$ is a completely positive trace non-increasing map, which satisfies $\sum_{a}\tr[\mcM_{i}^{a|x}(\tau)] = 1$ for all states $\tau \in \mcS(\mcH_{Q_{i}^{A}} \otimes \mcH_{O_{i-1}^{A}})$. We denote the marginal probabilities
\begin{equation}
\begin{aligned}
    p^{A_{i}}(a|x) &= \tr[\mcM_{i}^{a|x}(\tau \ot \sigma)] \\
    &= \tr[(\tau \ot \sigma) M_{a|x}^{i}] \\
    &= \tr[(\tau \ot \id_{O_{i-1}^{A}})(\id_{Q_{i}^{A}} \otimes \sqrt{\sigma}) M_{a|x}^{i}(\id_{Q_{i}^{A}} \otimes \sqrt{\sigma})] \\
    &= \tr\Big[ \tau \, \tr_{O_{i-1}^{A}}\Big[ (\id_{Q_{i}^{A}} \otimes \sqrt{\sigma})M_{a|x}^{i}(\id_{Q_{i}^{A}} \otimes \sqrt{\sigma})\Big]\Big] \\
    &= \tr\Big[ \tau\tilde{M}_{a|x}^{i}\Big].
\end{aligned}
\end{equation}
In the above, we defined $M_{a|x}^{i} = \sum_{\mu}K^{\dagger}_{\mu}K_{\mu}$, where $\{K_{\mu}\}_{\mu}$ is a set of Kraus operators for the channel $\mcM^{a|x}_{i}$, used the identity $\tr_{B}[(Y_{A} \otimes \id_{B})X_{AB}] = Y_{A}\tr_{B}[X_{AB}]$ for the fourth equality, and defined
\begin{equation}
    \tilde{M}_{a|x}^{i} := \tr_{O_{i-1}^{A}}\Big( (\id_{Q_{i}^{A}} \otimes \sqrt{\sigma})M_{a|x}^{i}(\id_{Q_{i}^{A}} \otimes \sqrt{\sigma})\Big).
\end{equation}
Note the set of operators $\{M_{a|x}^{i}\}_{a}$ are a POVM, and as a result the set of operators $\{\tilde{M}_{a|x}^{i}\}_{a}$ are also a POVM. By a similar procedure, we can define the POVMs $\{\tilde{N}_{b|y}^{i}\}$ from the channel $\mcN_{i}^{B}$, and describe the joint behavior of round $i$ by
\begin{equation}
    p^{i}(a,b|x,y) = \tr\big[ \rho_{i}(\tilde{M}_{a|x}^{i} \otimes \tilde{N}_{b|y}^{i})\big]. \label{eq:dist1}
\end{equation}
Thus, from the point of view of the statistics, we can view $\mcN_{i}^{A} \otimes \mcN_{i}^{B}$ as performing an uncharacterized measurement acting on the generated state $\rho_{i}$, rather than the state and internal device memory. 

Based on the above, we define
\begin{equation}
    \mu_{i} = \tr[\tilde{B}_{i}\rho_{i}],
\end{equation}
where 
\begin{equation}
    \tilde{B}_{i} = \frac{1}{4}\sum_{a,b,x,y \in \{0,1\}} w_{abxy} (\tilde{M}_{a|x}^{i} \otimes \tilde{N}_{b|y}^{i}),
\end{equation}
and $w_{abxy} = 1$ if $a \oplus b = x\cdot y$ and $0$ otherwise. We also define the optimized CHSH values associated to each state $\rho_{i}$,
\begin{equation}
    \mu^{\uparrow}_{i} = \max_{\{\tilde{M}^{i}_{a|x}\}_{a},\{\tilde{N}^{i}_{b|y}\}_{b}} \tr[\tilde{B}_{i}\rho_{i}]. \label{eq:muUp}
\end{equation}
Note that $\mu^{\uparrow}_{i}$ is only dependent on the state $\rho_{i}$, and not the measurement device. The values $\mu^{\uparrow}_{i}$ thus only depend on round $i$.

\subsection{Security proof of \Cref{prot: DISV_general_sequencial}}

\subsubsection{Real and ideal protocols}
The final state of the real and ideal protocols have an identical structure to that of \Cref{prot: DISV_general}, though the statistics of each round are no longer distributed independently. For convenience, we recall the real and ideal protocol final outputs below, 
\begin{equation}
    \begin{aligned}
        \rho_{\mathrm{real}} &= \sum_{t=1}^{n} p_{T}(t)\sum_{\mathbf{w} \in \Omega} p(\mathbf{w} | t) \ketbra{\mathbf{w}}{\mathbf{w}}_{\mathbf{W}} \otimes \ketbra{t}{t}_T \otimes \Lambda_{t}(\rho_{t}) \ot \proj{\Omega} +  (1 - p_{\Omega})  \proj{\perp},\\
        \rho_{\mathrm{ideal}} &= \sum_{t=1}^{n} \sum_{\mathbf{w} \in \Omega} p(\mathbf{w} | t) p_{T}(t) \ketbra{\mathbf{w}}{\mathbf{w}}_{\mathbf{w}} \otimes \ketbra{t}{t}_T \otimes [(1-\lambda_{t})\Lambda_{t}(\rho_{t}) + \lambda_{t}\phi^+] \ot \proj{\Omega} +  (1 - p_{\Omega})  \proj{\perp}.
    \end{aligned}
\end{equation}

\subsubsection{Soundness}

\begin{lemma}
    Protocol \ref{prot: DISV_general_sequencial} is $\epsilon_{s}$-sound, 
    \begin{equation}
        \begin{aligned}
            \epsilon_{s} &= \inf_{\delta > 0} \max \{ a_{2}(\delta),b_{3}(\delta)\},\\
            a_{2}(\delta) &= \exp\Big(-\frac{\lfloor(n-1)\delta\rfloor^{2}}{n-1}\Big),\\
            b_{3}(\delta) &= G_{\varepsilon}\left( \lfloor(n-1) (p^{\text{win}}_{\sharp} - \kappa - \delta) \rfloor / n \right).
        \end{aligned} \label{eq:eps4}
    \end{equation}\label{lem:p4sound}
\end{lemma}
\begin{proof}
    We follow the proof of \Cref{lem:p2sound} to obtain
    \begin{equation}
        \| \rho_{\text{real}} - \rho_{\text{ideal}} \|_{1} \leq \frac{\chg{2}}{n}\sum_{t=1}^{n}p(\Omega|t)\sqrt{1-F(\Lambda_{t}(\rho_{t}),\phi^+)}.
    \end{equation}
    Recall the random variables $W_{i}$, governed by the distribution $p^{i}$ from \Cref{eq:dist1}, which indicate whether or not the CHSH game was won on that round, satisfy $\mathbb{P}(W_{i} = 1) = \mu_{i}$. We can define a new set of random variables, $\{\hat{W}_{i}\}_{i}$, distributed according to the optimized expectation values $\mu_{i}^{\uparrow}$, $\mathbb{P}(\hat{W}_{i} = 1) = \mu_{i}^{\uparrow}$. Note that $\{\hat{W}_{i}\}$ are a set of independently distributed random variables. We now consider two cases, and introduce a free parameter $\delta > 0$.

    \vspace{0.2cm}

    \noindent \textbf{Case 1:} $\sum_{i=1}^{n}\mu_{i}^{\uparrow} \leq \lfloor (n-1)(p^{\text{win}}_{\sharp} - \kappa - \delta)\rfloor$. That is, the average value of $\mu_{i}^{\uparrow}$ is less than $\lfloor(n-1)(p^{\text{win}}_{\sharp} - \kappa)\rfloor /n$. Note that, since the variables $\hat{W}_{i}$ are independent, by following the proof \Cref{lem:p2sound} exactly (using the fact that $\mu_{i}^{\uparrow} \geq 0$ to omit the contribution of $\frac{\eta^{\mathrm{Q}}_{\mathrm{min}}}{n}$), we find using \Cref{thm: Hoeffding's theorem}
    \begin{equation}
        \hat{p}(\Omega|t) := \mathbb{P}\Bigg(\sum_{i \neq t}^{n} \hat{W}_{i} \geq \lfloor (n-1)(p^{\text{win}}_{\sharp} - \kappa)\rfloor \Bigg) \leq \exp\Big(-\frac{\lfloor(n-1)\delta\rfloor^{2}}{n-1}\Big) =: a_{3}(\delta).
    \end{equation}
    That is, the probability of the independent protocol (which generates the variables $\hat{W}_{i}$) not aborting is small. However, we have not shown that the probability of the actual protocol not aborting is also small. To establish this, we apply \Cref{cor:seq} to show that the former upper bounds the latter. Specifically, we obtain
    \begin{equation}
        p(\Omega|t) = \mathbb{P}\Bigg(\sum_{i \neq t}^{n} W_{i} \geq \lfloor(n-1)(p^{\text{win}}_{\sharp} - \kappa) \rfloor \Bigg) \leq \hat{p}(\Omega|t) \leq a_{3}(\delta),
    \end{equation}
    which implies (by bounding $\sqrt{1 - F(\Lambda_{t}(\rho_{t}),\phi^+)} \leq 1$)
    \begin{equation}
        \frac{1}{2}\| \rho_{\text{real}} - \rho_{\text{ideal}} \|_{1} \leq a_{2}(\delta).
    \end{equation}

    \vspace{0.2cm}

    \noindent \textbf{Case 2:} $\sum_{i=1}^{n}\mu_{i}^{\uparrow} \leq \lfloor (n-1)(\score_{\sharp} - \kappa - \delta)\rfloor$. We apply the bound $p(\Omega|t) \leq 1$ to obtain
    \begin{equation}
        \| \rho_{\text{real}} - \rho_{\text{ideal}} \|_{1} \leq \frac{\chg{2}}{n}\sum_{t=1}^{n}\sqrt{1-F(\Lambda_{t}(\rho_{t}),\phi^+)}.
    \end{equation}
    We then note
    \begin{equation}
        F(\Lambda_{t}(\rho_{t}),\phi^+) = \sup_{\Lambda \in \mcC}F(\Lambda(\rho_{t}),\phi^+) \geq \inf_{\rho \in \mcB(\mu^{\uparrow}_{t})}\sup_{\Lambda \in \mcC}F(\Lambda(\rho),\phi^+) = \Xi_{B}(\mu^{\uparrow}_{t}).    
    \end{equation}
    The inequality follows from the fact that, by the definition of $\mu^{\uparrow}_{t}$, there exists measurements which achieve $\tr[\tilde{B}_{i}\rho_{t}] = \mu^{\uparrow}_{t}$, i.e., $\rho_{t} \in \mcB(\mu^{\uparrow}_{t})$. This implies 
    \begin{equation}
       \chg{\frac{1}{2}} \| \rho_{\text{real}} - \rho_{\text{ideal}} \|_{1} \leq \frac{1}{n}\sum_{t=1}^{n}\sqrt{1-\Xi_{B}(\mu^{\uparrow}_{t})}.
    \end{equation}
    The remainder of the proof follows identically that of \Cref{lem:p2sound}, Case 2. 
\end{proof}

\subsubsection{Completeness}
We could also apply \Cref{lem:p1complete} to bound the completeness error of \Cref{prot: DISV_general_sequencial}. However, since we are restricting to the CHSH case, where the variables $W_{i}$ are binary, we can use a sharper concentration inequality.

\begin{theorem}[\cite{Zubkov13}]
    Let $n \in \mathbb{N}$, $p \in (0,1)$ and $Z$ be a random variable distributed according to $Z \sim \mathsf{Binomial}(n,p)$. Then, for every $k = 0,...,n-1$ we have
    \begin{eqnarray}
        C(n,p,k) \leq \mathbb{P}( X \leq k ) \leq C(n,p,k+1),
    \end{eqnarray}
    where
    \begin{equation}
        \begin{aligned}
            C(n,p,k) &= \Phi\Big( \mathrm{sign}\big( k/n - p \big) \sqrt{2n G\big( k/n , p\big)} \Big),\\
            \Phi(x) &= \frac{1}{\sqrt{2\pi}}\int_{-\infty}^{x} \mathrm{d}u \, e^{-u^{2}/2}, \\
            G(x,p) &= x \ln\Big(\frac{x}{p}\Big) + (1-x) \ln\Big( \frac{1-x}{1-p}\Big).
        \end{aligned}
    \end{equation}
    \label{thm:binomial}
\end{theorem}

\begin{lemma}
    Protocol \ref{prot: DISV_general_sequencial} is $\epsilon_{c}$-complete, where
    \begin{equation}
        \epsilon_{c} = 1 - C\big(N-1,p^{\text{win}}_{\sharp}, \lceil (N-1)(p^{\text{win}}_{\sharp} - \kappa) \rceil\big). \label{eq:p4complete}
    \end{equation}\label{lem:p4complete}
\end{lemma}
\begin{proof}
    Consider an honest implementation for which the variables $W_{1},...,W_{n}$ are i.i.d. random variables with $\mathbb{E}[W_{i}] = \omega_{\sharp}$. Let $\bar{W}_{i} = 1 - W_{i}$. Then
    \begin{equation}
    \begin{aligned}
        1- p(\Omega|t) &= \mathbb{P} \Bigg( \sum_{i \neq t}^{n} \bar{W}_{i} > (n-1)(1 - [p^{\text{win}}_{\sharp} - \kappa]) \Bigg) \\
        &= 1 - \mathbb{P} \Bigg( \sum_{i \neq t} \bar{W}_{j} \leq (n-1)(1 - [p^{\text{win}}_{\sharp} - \kappa]) \Bigg) \\
        & \leq 1 - \mathbb{P} \Bigg( \sum_{i \neq t} \bar{W}_{j} \leq \lfloor (n-1)(1-[p^{\text{win}}_{\sharp} - \kappa]) \rfloor \Bigg).
    \end{aligned}
    \end{equation}
    Let $Z = \sum_{i \neq t}^{n} \bar{W}_{i}$. Then $Z$ is a random variable distributed according to $\mathsf{Binomial}(n-1,1-p^{\text{win}}_{\sharp})$. We can therefore apply Theorem \ref{thm:binomial} to obtain
    \begin{eqnarray}
        p(\Omega|t) \geq C\big(n-1,1-p^{\text{win}}_{\sharp}, \lfloor (n-1)(1-[p^{\text{win}}_{\sharp} - \kappa]) \rfloor \big).
    \end{eqnarray}
    We therefore have
    \begin{eqnarray}
       p_{\Omega} = \frac{1}{n}\sum_{t=1}^{n}p(\Omega|t) \geq  C\big(n-1,1-p^{\text{win}}_{\sharp}, \lfloor (n-1)(1-[p^{\text{win}}_{\sharp} - \kappa]) \rfloor \big),
    \end{eqnarray}
    proving the claim. 
\end{proof}
As a result, we find \Cref{prot: DISV_general_sequencial} is $(\epsilon_{s},\epsilon_{c})$-secure where $\epsilon_{s}$ and $\epsilon_{c}$ are given by \Cref{eq:eps2} and \Cref{eq:p4complete}, respectively. 

\subsection{Security proof of \Cref{prot: DISV_general_sequencial_var}}
In this final subsection, we prove the security of \Cref{prot: DISV_general_sequencial_var}, which is a variant of \Cref{prot: DISV_general_sequencial} in which the user is not required to apply the final extraction channel $\Lambda_{t}$. 

\subsubsection{Real and ideal protocols}
These have an identical structure to that of \Cref{prot: DISV_general_var}. Specifically, the final outputs are given by 

\begin{equation}
\begin{aligned}
    \rho_{\text{real}} &= \left( \sum_{t=1}^{n}p_{T}(t) \sum_{\mathbf{w} \in \Omega} p(\mathbf{w} | t) \, \ketbra{\mathbf{w}}{\mathbf{w}}_{\mathbf{W}} \otimes \ketbra{t}{t}_T \otimes \ketbra{00}{00} \otimes \rho_{t} \ot \proj{\Omega} \right) +  (1 - p_{\Omega}) \ketbra{\perp}{\perp},\\
    \rho_{\text{ideal}} &= \left( \sum_{t=1}^{n} p_{T}(t) \sum_{\mathbf{w} \in \Omega} p(\mathbf{w} | t) \, \ketbra{\mathbf{w}}{\mathbf{w}}_{\mathbf{W}} \otimes \ketbra{t}{t}_T \otimes [(1-\lambda_{t})\ketbra{00}{00}\otimes \rho_{t} + \lambda_{t}\sigma_{t}] \ot \proj{\Omega} \right) +  (1 - p_{\Omega}) \proj{\perp},
\end{aligned}
\end{equation}
where $\sigma_{t}$ is the companion state of $\rho_{t}$.

\subsubsection{Soundness}
\begin{lemma}
    Protocol \ref{prot: DISV_general_sequencial_var} is $\epsilon_{s}$-sound, where $\epsilon_{s}$ is given by \Cref{eq:eps4}. \label{lem:p5sound}
\end{lemma}
\begin{proof}
    The proof follows the structure to that of \Cref{lem:p3sound}, with the same modifications introduced to prove \Cref{lem:p4sound}. In detail we have
    \begin{equation}
        \| \rho_{\text{real}} - \rho_{\text{ideal}} \|_{1} \leq \frac{\chg{2}}{n}\sum_{t=1}^{n} p(\Omega|t)  \lambda_{t}\sqrt{1 - F(\ketbra{00}{00} \otimes \rho_{t},\sigma_{t})}.
    \end{equation}

    \vspace{0.2cm}

    \noindent \textbf{Case 1:} $\sum_{i=1}^{n} \frac{\mu_{i}^{\uparrow}}{n} \leq \frac{n-1}{n}(p^{\text{win}}_{\sharp} - \kappa - \delta)$. Here, the proof proceeds identically to that of \Cref{lem:p4sound}, Case 1. 

    \vspace{0.2cm}

    \noindent \textbf{Case 2:} $\sum_{i=1}^{n} \frac{\mu_{i}^{\uparrow}}{n} > \frac{n-1}{n}(p^{\text{win}}_{\sharp} - \kappa - \delta)$. We proceed by bounding $p(\Omega|t) \leq 1$ and lower bounding the fidelity via the extractability $\Xi_{B}(\mu_{i}^{\uparrow})$. Since $\sigma_{t}$ is a companion state of $\rho_{t}$, and $\rho_{t} \in \mcB(\mu^{\uparrow}_{t})$, it follows from the same reasoning used in \Cref{eq:compExt} that 
    \begin{equation}
        F(\ketbra{00}{00} \otimes \rho_{t},\sigma_{t}) \geq \sup_{\Lambda \in \mcC} F(\Lambda(\rho_{t}),\phi^+) \geq \inf_{\rho \in \mcB(\mu_{t}^{\uparrow})}\sup_{\Lambda \in \mcC} F(\Lambda(\rho),\phi^+) = \Xi_{B}( \mu_{t}^{\uparrow}).
    \end{equation}
    The remainder of the proof follows identically to that of \Cref{lem:p2sound}, Case 2.
\end{proof}

\subsubsection{Completeness}
\Cref{prot: DISV_general_sequencial_var} has a completeness error given by \Cref{lem:p4complete}, resulting in $(\epsilon_{s},\epsilon_{c})$-security where $\epsilon_{s}$ and $\epsilon_{c}$ are given by \Cref{eq:eps2} and \Cref{eq:p4complete}, respectively. 

\section{Bounding the abort probability in the sequential setting} \label{app:seqLem}
The aim of this section is to show that the abort probability of Protocol \ref{prot: DISV_general_sequencial} can be bounded in terms of the maximal achievable CHSH scores $\mu_{i}^{\uparrow}$, defined in \Cref{eq:muUp}. To do so, we describe the sequential protocol as stochastic process $(\Gamma , \mathcal{A} , P)$, where 
\begin{itemize}
    \item The sample space $\Gamma$ consists of all possible sequences of outcomes of the experiment. Each round measurement round $i$ results in either a loss (\(0\)) or a win (\(1\)) of the CHSH game, recorded in the classical register $W_{i}$. To ease notation, we consider $n$ such rounds, though in the actual protocol there are $n-1$. Hence
  \[
  \Gamma = \{ \mathbf{w} = (w_1, w_2, \dots,w_{n})  \ : \  w_i \in \{0, 1\} \}.
  \]
  \item The sigma algebra \( \mathcal{A} \) is defined by the cylinder sets generated by the trajectories $\mathbf{w}$.
  \item  In the CHSH game, the outcomes of future rounds can only depend upon the outcomes of previous rounds. The probability measure \(P\) is thus a product measure, i.e., the probability of an outcome \( \mathbf{w} = (w_1, w_2, \dots , w_{n}) \) is given by
\begin{eqnarray}
 P(\mathbf{w}) =  p_{1}(w_1) p_{2}(w_2 |  w_1) p_{3}(w_3 | w_2 , w_{1} ) \cdots  p_{n}(w_{n} | w_{n-1} , \dots , w_{1}). \label{eqn: product-distributions}  
\end{eqnarray}
We denote the set of such measures by $\mathfrak{P}$.
\end{itemize}

Let us also define for a given vector $\bm{\mu} = [\mu_{1},\mu_{2},...,\mu_{n}] \in [0,1]^{n}$ the set 
\begin{eqnarray}\label{lemm: optimality of independent distributions}
\mathfrak{P}_{\bm{\mu}} =  \Big\{ P \in \mathfrak{P} \ : \   p_{1}(1) \leq \mu_1 , \, p_{2}(1 | w_1) \leq \mu_2 , \, ... , \, p_{n}(1 | w_{n-1},\dots ,w_{1}) \leq \mu_{n} \ \forall w_{1},...,w_{n-1} \in \{0,1\}\Big \}.
\end{eqnarray}
We further define the product probability distribution $P^*_{\bm{\mu}} \in \mathfrak{P}_{\bm{\mu}}$, given by
\begin{equation}
    P^*_{\bm{\mu}}(\mathbf{w}) = p_1^*(w_1) p_2^*(w_2) \cdots p_{n}^*(w_{n}), 
\end{equation}
where
\begin{equation}
    p_{i}^*(1) = \mu_{i}.
\end{equation}
Now, consider the following claim. 
\begin{lemma}\label{lem:seq}
Let $c \in \mathbb{N}$ be any non-zero natural number. Let $\Omega$ be the event defined by 
\begin{eqnarray}
    \Omega = \Big\{ \mathbf{w} \in \Gamma \ : \ \sum_{i=1}^{n} w_{i} \geq c \Big\}.
\end{eqnarray}
Then for any fixed $\bm{\mu} \in [0,1]^{n}$, 
\begin{equation}
   \max_{P \in \mathfrak{P}_{\bm{\mu}} } P(\Omega) = P^*_{\bm{\mu}}(\Omega). \label{eq:maxP}
\end{equation}
\end{lemma}
\begin{proof}
Define for any positive integers $k$ and $m$ satisfying $k \leq m \leq n$, 
\begin{equation}
    \Omega^m_k := \Big\{ \mathbf{w} \in \Gamma \  : \ \sum_{i=1}^{m} w_i \geq k \Big\}.
\end{equation}
This is the set of trajectories $\mathbf{w}$ (of length $n$) for which the total number of wins (that is, the total number of instances when $w_{i} = 1$) up to round $m$ is at least \( k \). Now, suppose for a fixed $\bm{\mu}$, $P \in \mathfrak{P}_{\bm{\mu}}$. We then have the following recursion,

\begin{equation}
P(\Omega^{m}_{k}) = P(\Omega^{m-1}_{k}) + \sum_{\substack{\mathbf{w} \in \Gamma \ \text{s.t.}\\ \sum_{i=1}^{m-1} w_i = k-1,\\
w_{m}=1}} P(\mathbf{w}). \label{eqn: recursive_1}
\end{equation}
This follows from the fact that strings $\mathbf{w}$ achieving a sum of at least \( k \) at round \( m \) fall into one of two distinct cases.

\vspace{0.2cm}

\textbf{Case 1:} Strings $\mathbf{w}$ which have already achieved a sum of at least \( k \) by round \( m-1 \), i.e., $\sum_{i=1}^{m-1} w_{i} \geq k$.

\vspace{0.2cm}

\textbf{Case 2:} Strings $\mathbf{w}$ which achieve a sum of exactly \( k-1 \) by round \( m-1 \), i.e., $\sum_{i=1}^{m-1}w_{i} = k-1$, and then win round \( m \), i.e., $w_{m} = 1$.

\vspace{0.2cm}

\noindent The two terms on the right hand side of \Cref{eqn: recursive_1} account of these cases, respectively. Note that $P(\Omega_{k}^{m-1})$ is independent of the distribution assigned to final random variable $W_{m}$, i.e., independent of $p_{m}(w_{m}|w_{m-1},...,w_{1})$. We can also expand the second term to obtain
\begin{equation}
\begin{aligned}
    \sum_{\substack{\mathbf{w} \in \Gamma \ \text{s.t.}\\ \sum_{i=1}^{m-1} w_i = k-1,\\
w_{m}=1}} P(\mathbf{w}) &= \sum_{\substack{\mathbf{w} \in \Gamma \ \text{s.t.}\\ \sum_{i=1}^{m-1} w_i = k-1,\\
w_{m}=1}} p_{1}(w_{1})p_{2}(w_{2}|w_{1})\cdots p_{w_{m-1}}(w_{m-1}|w_{m-2},...,w_{1}) p_{w_{m}}(1|w_{m-1},...,w_{1}) \\
&\leq \sum_{\substack{\mathbf{w} \in \Gamma \ \text{s.t.}\\ \sum_{i=1}^{m-1} w_i = k-1,\\
w_{m}=1}} p_{1}(w_{1})p_{2}(w_{2}|w_{1})\cdots p_{w_{m-1}}(w_{m-1}|w_{m-2},...,w_{1}) \mu_{m} \\
&= \sum_{\substack{\mathbf{w} \in \Gamma \ \text{s.t.}\\ \sum_{i=1}^{m-1} w_i = k-1,\\
w_{m}=1}} p_{1}(w_{1})p_{2}(w_{2}|w_{1})\cdots p_{w_{m-1}}(w_{m-1}|w_{m-2},...,w_{1}) p_{m}^*(1),
\end{aligned} \label{eq:upSum}
\end{equation}
where we used the fact that $P \in \mathfrak{P}_{\bm{\mu}}$ for the inequality and the definition of $P^*_{\bm{\mu}}$ for the second equality. We therefore see that probability distribution achieving the optimal value of the maximization $\max_{P \in \mathfrak{P}_{\bm{\mu}}} P(\Omega_{k}^{m})$
must satisfy
\[
p_m(1|w_{m-1},...,w_{1}) = p_m^*(1).
\]
In particular, $P(\Omega) = P(\Omega_{c}^{n})$, which implies the distribution maximizing the left hand side of \Cref{eq:maxP} satisfies $p_{n}(w_{n}|w_{n-1},...,w_{1}) = p^*_{n}(w_{n})$. Restricting to distributions satisfying this, consider
\begin{equation}
\begin{aligned}
    P(\Omega_{c}^{n}) &= \sum_{\substack{\mathbf{w} \in \Gamma \ \text{s.t.}\\ \sum_{i=1}^{n}w_{i} \geq c}} p_{1}(w_{1})p_{2}(w_{2}|w_{1})\cdots p_{n}(w_{n}|w_{n-1},...,w_{1}) \\
    &= \sum_{z \in \{0,1\}}\sum_{\substack{\mathbf{w} \in \Gamma \ \text{s.t.}\\ \sum_{i=1}^{n-1}w_{i} \geq c - z\\w_{n} = z}} p_{1}(w_{1})p_{2}(w_{2}|w_{1})\cdots p_{n}(z|w_{n-1},...,w_{1})\\
    &= \sum_{z \in \{0,1\}}p^*_{n}(z)\sum_{\substack{\mathbf{w} \in \Gamma \ \text{s.t.}\\ \sum_{i=1}^{n-1}w_{i} \geq c - z\\w_{n}=z}} p_{1}(w_{1})p_{2}(w_{2}|w_{1})\cdots p_{n-1}(w_{n-1}|w_{n-2},...,w_{1})\\
    &= \sum_{z \in \{0,1\}}p^*_{n}(z)\sum_{\substack{\mathbf{w} \in \Gamma \ \text{s.t.}\\ \sum_{i=1}^{n-1}w_{i} \geq c - z}} p_{1}(w_{1})p_{2}(w_{2}|w_{1})\cdots p_{n-1}(w_{n-1}|w_{n-2},...,w_{1})\\
    &=\sum_{z \in \{0,1\}}p^*_{n}(z)P(\Omega_{c-z}^{n-1}).
\end{aligned}
\end{equation}
Notice that, using \Cref{eq:upSum} by setting $m=n-1$ and $k=c-z$, the maximum of $P(\Omega_{c-z}^{n-1})$ occurs when $p_{n-1}(w_{n-1}|w_{n-2},...,w_{1}) = p^*_{n-1}(w_{1})$ for both values of $z$. By the same line of reasoning above, we therefore find the optimal distribution $P$ must satisfy this constraint, implying
\begin{equation}
    P(\Omega_{c}^{n}) = \sum_{z_{1},z_{2}\in \{0,1\}}p^*_{n}(z_{1})p^*_{n-1}(z_{2}) P(\Omega_{c-z_{1}-z_{2}}^{n-2}). 
\end{equation}
We can keep iterating the above procedure, until we obtain
\begin{equation}
    P(\Omega_{c}^{n}) = \sum_{\mathbf{z} \in \Gamma}P^*_{\bm{\mu}}(\mathbf{z}) P(\Omega^0_{c - \sum_{i=1}^{n}z_{i}}),
\end{equation}
where
\begin{equation}
    \Omega^0_{c - \sum_{i=1}^{n}z_{i}} = \Big\{ \mathbf{w} \in \Gamma \ : \ 0 \geq c - \sum_{i=1}^{n}z_{i} \Big\} = \begin{cases}
        \Gamma \ \text{if} \ \sum_{i=1}^{n}z_{i} \geq c, \\
        \varnothing \ \text{otherwise.}
    \end{cases}
\end{equation}
Thus $P(\Omega^0_{c - \sum_{i=1}^{n}z_{i}}) = 1$ if $\sum_{i=1}^{n} z_{i} \geq c$ and zero otherwise, implying
\begin{equation}
    \sum_{\mathbf{z} \in \Gamma}P^*(\mathbf{z})_{\bm{\mu}} P(\Omega^0_{c - \sum_{i=1}^{n}z_{i}}) = \sum_{\substack{\mathbf{z} \in \Gamma \ \text{s.t.}\\ \sum_{i=1}^{n}z_{i} \geq c}}P^*_{\bm{\mu}}(\mathbf{z}) = P^*_{\bm{\mu}}(\Omega_{c}^{n}).
\end{equation}
We have therefore established
\begin{equation}
    \max_{P \in \mathfrak{P}_{\bm{\mu}} } P(\Omega) \leq P^*_{\bm{\mu}}(\Omega).
\end{equation}
The fact that $P_{\bm{\mu}}^* \in \mathfrak{P}_{\bm{\mu}}$ completes the proof. 
\end{proof}

We now state an immediate corollary for the particular case encountered in this work.
\begin{corollary}
    For $i=1,...,n$, let $\rho_{i}$ and $\mcN_{i}$ be a sequence of states and channels which induce the binary CHSH variables $W_{i}$, as described in \Cref{app:seqProc}. Let $\mu_{i}^{\uparrow}$ be defined in \Cref{eq:muUp} and $\hat{W}_{i}$ be independent binary random variables defined by $\mathbb{P}(\hat{W}_{i} = 1) = \mu_{i}^{\uparrow}$. Then for any $t \in \{1,...,n\}$,
    \begin{equation}
        \mathbb{P}\Bigg( \sum_{i \neq t}^{n} W_{i} \geq \lfloor(n-1)(\omega_{\sharp} - \kappa)\rfloor \Bigg) \leq \mathbb{P}\Bigg(\sum_{i \neq t}^{n} \hat{W}_{i} \geq \lfloor(n-1)(\omega_{\sharp} - \kappa)\rfloor \Bigg).
    \end{equation} \label{cor:seq}
\end{corollary}
\begin{proof}
    Let us relabel the string of $n-1$ binary variables $\mathbf{W} = \{W_{i}\}_{i \neq t}^{n} \equiv \{W_{1},...,W_{m}\}$ where $m = n-1$. They follow a distribution of the form
    \begin{equation}
        P(\mathbf{W}) = p_{1}(w_{1})p_{2}(w_{2}|w_{1})\cdots p_{m}(w_{m}|w_{m-1},...,w_{1}), 
    \end{equation}
    and by the definition of $\mu_{i}^{\uparrow}$, $P$ is a member of $ \mathfrak{P}_{\bm{\mu}^{\uparrow}}$, where $\bm{\mu}^{\uparrow} = [\mu_{1}^{\uparrow},...,\mu_{m}^{\uparrow}]$. Let
    \begin{equation}
        \hat{\Omega} = \Big\{ \mathbf{W} \in \{0,1\}^{m} \ : \ \sum_{i=1}^{m} W_{i} \geq \lfloor m(\omega_{\sharp} - \kappa) \rfloor \Big\}.
    \end{equation}
    Then
    \begin{equation}
        \begin{aligned}
            \mathbb{P}\Bigg(\sum_{i \neq t}^{n} \frac{W_{i}}{n-1} \geq \omega_{\sharp} - \kappa \Bigg) & \leq P(\hat{\Omega}) \leq \max_{P' \in \mathfrak{P}_{\bm{\mu}^{\uparrow}}} P'(\hat{\Omega}) = P^*_{\bm{\mu}^{\uparrow}}(\hat{\Omega}) = \mathbb{P}\Bigg(\sum_{i \neq t}^{n} \hat{W}_{i} \geq \lfloor(n-1)(\omega_{\sharp} - \kappa)\rfloor \Bigg)
        \end{aligned}
    \end{equation}
    as desired, where the final equality follows from the fact that the random variables $\hat{W}_{i}$ are distributed according to $P^*_{\bm{\mu}^{\uparrow}}$.
    
\end{proof}

\section{Bounding the LOCC extractability} \label{app: LOCC_extractibiltity}
The objective of this section is to derive reliable lower bounds on the LOCC extractability, as discussed in Section 8 of the main text. We begin by mathematically defining the general problem, followed by a reduction to qubit strategies when working in Bell scenarios with two inputs and two outputs per party. We then present a numerical method to bound the extractability in this case.  

\begin{figure}
    \centering
    \includegraphics[width=\linewidth]{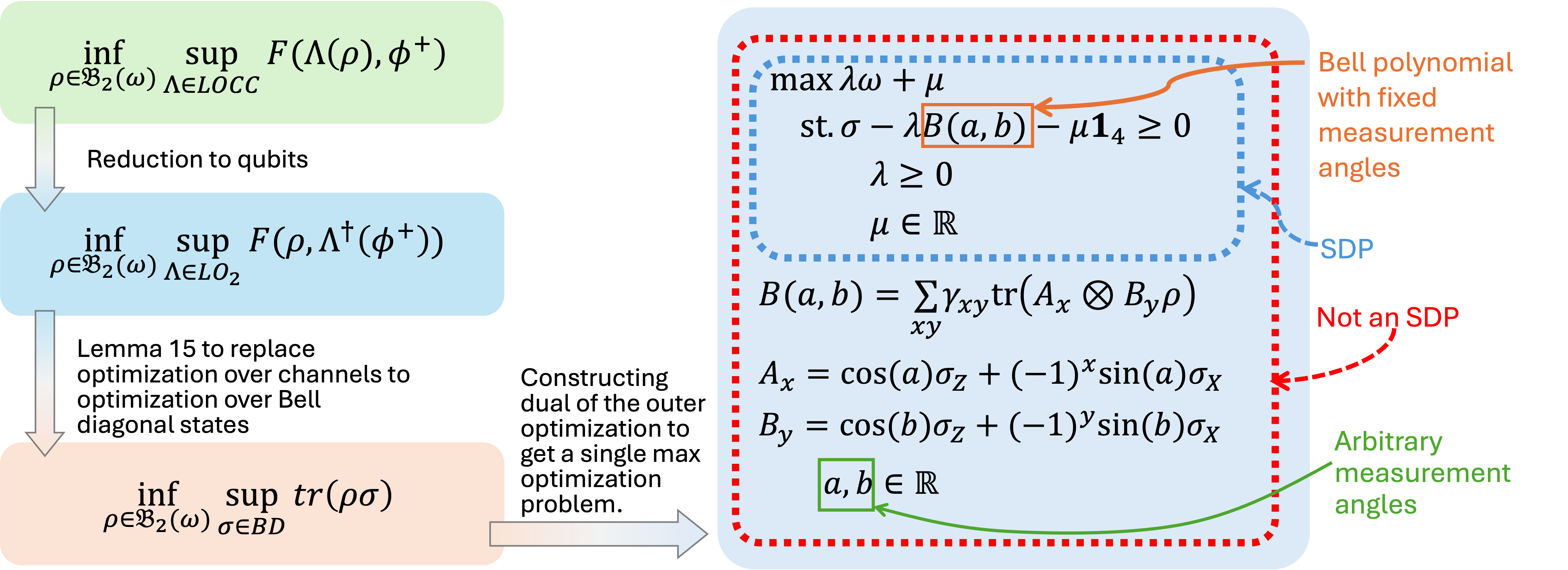}
   \caption{\chg{Sketch of the method for computing lower bounds on singlet extractability. After all reductions, the problem becomes a single maximization that is not an SDP in general, but reduces to an SDP when two parameters (corresponding to the Bell test measurement observables) are fixed. This allows the optimization to be solved by discretizing the parameter space, with a controllable penalty that can be reduced by refining the grid.}}
    \label{fig: proof}
\end{figure}
\subsection{Stating the problem}

We begin with some definitions. Recall the bipartite Bell scenario described in the main text, in which two parties perform local measurements on an entangled state $\rho \in \mcS(\mcH_{Q_{A}}\otimes \mcH_{Q_{B}})$. Their binary inputs are labeled by $X=x$ and $Y=y$, and outputs by $A=a$ and $B=b$, respectively. We label the corresponding POVMs $\{\{M_{a|x}\}_{a \in \{0,1\}}\}_{x \in \{0,1\}}$ on $\mcH_{Q_{A}}$ and $\{\{N_{b|y}\}_{b \in \{0,1\}}\}_{y \in \{0,1\}}$ on $\mcH_{Q_{B}}$, which define the observables $\{A_{x}\}_{x \in \{0,1\}}$ and $\{B_{y}\}_{y \in \{0,1\}}$. We define
\begin{equation}
    \mathcal{B}_{\omega} := \Big\{ \rho \in \mathcal{S}(\mathcal{H}_{Q_{A}} \otimes \mathcal{H}_{Q_{B}}) \ : \ \exists \{A_{x}\}_{x},\{B_{y}\}_{y} \ \ \text{s.t.} \ \ \tr[B\rho] \geq \omega \Big\}, 
\end{equation}
where 
\begin{equation}
    B = \sum_{x,y} \gamma_{xy} (A_{x} \otimes B_{y})
\end{equation}
is a Bell operator with some real coefficients $\gamma_{xy}$. Let $\mcC \in \{ \mathsf{U},\mathsf{LO},\mathsf{LOSR},\mathsf{LOCC}\}$ denote the set of unital, local, LOSR (local operations and shared randomness), LOCC and unital quantum channels from $\mathcal{S}(\mathcal{H}_{Q_{A}} \otimes \mathcal{H}_{Q_{B}})$ to $\mathcal{S}(\mathbb{C}^{d} \otimes \mathbb{C}^{d})$, and $\psi^* = \ketbra{\psi^*}{\psi^*} \in \mathcal{S}(\mathbb{C}^{d} \otimes \mathbb{C}^{d})$ be a target state. We also write $\mathcal{S}_{2} = \mathcal{S}(\mbC^{2} \otimes \mbC^{2})$ for the set of two-qubit states and denote by $\mcC_{2}$ the set of channels (in the class $\mcC$) from $\mathcal{S}_{2}$ to itself, i.e., from qubits to qubits. We will also label a qubit system held by Alice by $\hat{Q}_{A}$, and similarly for Bob.

The problem we wish to solve takes the form
\begin{equation}
    \Xi_{B}^{\mcC}(\omega) = \inf_{\substack{\rho \in \mathcal{B}_{\omega}}} \sup_{\Lambda \in \mcC} F(\Lambda(\rho),\psi^*),
\end{equation}
Note that since $\psi^*$ is pure, we can use the identity $F(\Lambda(\rho),\psi^*) = \tr[\Lambda(\rho)\psi^*]$.

\subsection{Reduction to qubits}
We now apply Jordan's lemma to reduce the problem to qubits.
\begin{lemma}[Jordan's lemma~\cite{Jordan}]
Let $A_{0},A_{1}$ be two binary observables on a Hilbert space $\mathcal{H}$. Then there exists a basis for which $\mathcal{H}$ can be decomposed block diagonally into subspaces with dimension $\leq 2$, where each subspace is preserved by $A_{0},A_{1}$. \label{lem:jordan}
\end{lemma}

\noindent This allows us to perform the following reduction (see, e.g.,~\cite{kaniewski2016analytic,Bhavsar2023} for details). We define the set of two-qubit states which can achieve a Bell value $\omega$ below:
\begin{equation}
    \mathcal{B}_{2,\omega} = \Big\{ \rho \in \mathcal{S}_{2} \ : \ \exists (a,b) \in [0,\pi/2] \times [0,\pi/2] \ \ \text{s.t.} \ \ \tr[B(a,b)\rho] \geq \omega \Big\}, \label{eq:B_set}
\end{equation}
where $B(a,b)$ is the Bell operator $B$ constructed from the qubit observables 
\begin{equation}
    A_{x} = \cos(a) \, \sigma_{Z} + (-1)^{x}\sin(a) \, \sigma_{X}, \ \ \ B_{y} = \cos(b) \, \sigma_{Z} + (-1)^{y}\sin(b) \, \sigma_{X}.
\end{equation}

\begin{lemma}
    Let $B$ be a Bell operator in the bipartite minimal Bell scenario. Then the following inequality holds:
    \begin{equation}
        \Xi^{\mathsf{LOCC}}_{B}(\omega) \geq \tilde{f}(\omega),
    \end{equation}
    where $\tilde{f}(\omega)$ is any convex lower bound on the qubit LO extractability, 
    \begin{equation}
        f_{2}(\omega) := \inf_{\substack{\rho \in \mathcal{B}_{2,\omega}}} \sup_{\Lambda \in \mathsf{LO}_{2} } \tr[\Lambda(\rho)\phi]. \label{eq:singQubitExt}
    \end{equation}
\end{lemma}
\noindent We refer to the function $f_{2}(\omega)$ as the singlet fidelity~\cite{bardyn2009device}.
\begin{proof}
Applying Lemma \ref{lem:jordan} to $\mathcal{H}_{Q_{A}}$ and $\mathcal{H}_{Q_{B}}$, we can write
\begin{equation}
    A_{x} = \sum_{\alpha} A_{x}^{(\alpha)} \otimes \ketbra{\alpha}{\alpha}_{F_{A}}, \ B_{y} = \sum_{\beta} B_{y}^{(\beta)} \otimes \ketbra{\beta}{\beta}_{F_{B}}.
\end{equation}
Above, we introduced flag registers $F_{A}$ and $F_{B}$, and qubit registers $\hat{Q}_{A}$ and $\hat{Q}_{B}$, such that $\mcH_{Q_{A}} = \mcH_{\hat{Q}_{A}} \otimes \mcH_{F_{A}}$ and $\mcH_{Q_{B}} = \mcH_{\hat{Q}_{B}} \otimes \mcH_{F_{B}}$. Without loss of generality, we can apply local unitaries to each block such that each qubit measurement is real and lies in the $Z-X$ plane of the Bloch sphere~\cite{kaniewski2016analytic}, that is
\begin{equation}
    A_{x}^{(\alpha)} = \cos(a_{\alpha}) \, \sigma_{Z} + (-1)^{x} \sin(a_{\alpha}) \, \sigma_{X}, \ \ \ \text{and} \ \ \ B_{y}^{(\beta)} = \cos(b_{\beta}) \, \sigma_{Z} + (-1)^{y} \sin(b_{\beta}) \, \sigma_{X}, \label{eq:realObs}
\end{equation}
for some $a_{\alpha},b_{\beta} \in (0,\pi/2]$. The Bell operator $B$ then decomposes,
\begin{equation}
    B = \sum_{\alpha,\beta} B^{(\alpha,\beta)} \otimes \ketbra{\alpha}{\alpha}_{F_{A}} \otimes \ketbra{\beta}{\beta}_{F_{B}}.    
\end{equation}
We denote a generic state $\rho \in \mcS(\mcH_{Q_{A}} \otimes \mcH_{Q_{B}})$ by
\begin{equation}
    \rho = \sum_{\alpha,\alpha',\beta,\beta'} \rho^{(\alpha,\alpha'),(\beta,\beta')} \otimes \ketbra{\alpha}{\alpha'}_{F_{A}} \otimes \ketbra{\beta}{\beta'}_{F_{B}}, \label{eq:genRho}
\end{equation}
and write $\rho^{(\alpha,\beta)} = p_{\alpha,\beta}\hat{\rho}^{(\alpha,\alpha),(\beta,\beta)}$, where $p_{\alpha,\beta} = \tr[\rho^{(\alpha,\alpha),(\beta,\beta)}]$ and $\hat{\rho}^{(\alpha,\alpha),(\beta,\beta)} = \rho^{(\alpha,\alpha),(\beta,\beta)}/p_{\alpha,\beta}$. We then consider the following LOCC channel:
\begin{equation}
    \Lambda = \mathcal{M} \circ \Pi, \label{eq:bigCh}
\end{equation}
where $\Pi : \mcS(\mcH_{Q_{A}} \otimes \mcH_{Q_{B}}) \to \mcS(\mcH_{\hat{Q}_{A}} \otimes \mcH_{\hat{Q}_{B}} \otimes \mcH_{C})$,
\begin{equation}
    \Pi(\sigma_{\hat{Q}_{A}\hat{Q}_{B}F_{A}F_{B}}) = \sum_{\alpha,\beta} (\id_{4} \otimes \bra{\alpha}_{F_{A}} \otimes \bra{\beta}_{F_{B}}) \sigma (\id_{4} \otimes \ket{\alpha}_{F_{A}} \otimes \ket{\beta}_{F_{B}}) \otimes \ketbra{\alpha,\beta}{\alpha,\beta}_{C} 
\end{equation}
and $\mcM : \mcS(\mcH_{\hat{Q}_{A}} \otimes \mcH_{\hat{Q}_{B}} \otimes \mcH_{C}) \to \mcS(\mcH_{\hat{Q}_{A}} \otimes \mcH_{\hat{Q}_{B}})$,
\begin{equation}
    \mathcal{M}(\tau_{\hat{Q}_{A}\hat{Q}_{B}C}) = \sum_{\alpha,\beta} \sum_{k} (E_{k}^{(\alpha,\beta)} \otimes \bra{\alpha,\beta})\tau(E_{k}^{(\alpha,\beta)\dagger} \otimes \ket{\alpha,\beta}).
\end{equation}
Above, $\{E_{k}^{(\alpha,\beta)}\}_{k}$ are the Kraus operators of a quantum channels $\Lambda^{(\alpha,\beta)} \in \mathsf{LO}_{2}$, i.e., $\Lambda^{(\alpha,\beta)}(\rho_{\hat{Q}_{A}\hat{Q}_{B}}) = \sum_{k} E^{(\alpha,\beta)}_{k}\rho E^{(\alpha,\beta)\dagger}_{k}$. Note that the register $C$ is shared by both devices, i.e., $\Pi$ is not a local channel. It is however an LOCC channel, since it can be performed by both devices measuring $F_{A}$ and $F_{B}$, and communicating the results.  Applied to a state of the form \eqref{eq:genRho}, we find
\begin{equation}
    \Lambda(\rho) = \sum_{\alpha,\beta}p_{\alpha,\beta}\Lambda^{(\alpha,\beta)}(\hat{\rho}^{(\alpha,\beta)}),
\end{equation}
which implies 
\begin{equation}
    \begin{aligned}
        \Xi^{\mathsf{LOCC}}_{B}(\omega) &= \inf_{\substack{\rho,\{A_{x}^{(\alpha)},B_{y}^{(\beta)}\}_{\alpha,\beta,x,y} \\ \text{s.t.} \ \sum_{\alpha,\beta}p_{\alpha,\beta} \tr[B^{(\alpha,\beta)}\hat{\rho}^{(\alpha,\beta)}] \geq \omega}} \sup_{\Lambda \in \mathsf{LOCC}} \tr[\Lambda(\rho) \psi^*] \\
        &\geq \inf_{\substack{\{\hat{\rho}^{(\alpha,\beta)},p_{\alpha,\beta},A_{x}^{(\alpha)},B_{y}^{(\beta)}\}_{\alpha,\beta} \\ \text{s.t.} \ \sum_{\alpha,\beta}p_{\alpha,\beta} \tr[B^{(\alpha,\beta)}\hat{\rho}^{(\alpha,\beta)}] \geq \omega}} \sum_{\alpha,\beta} p_{\alpha,\beta} \sup_{\Lambda^{(\alpha,\beta)} \in \mathsf{LO}_{2}} \tr[\Lambda^{(\alpha,\beta)}(\hat{\rho}^{(\alpha,\beta)}) \psi^*].
    \end{aligned}
\end{equation}
Let $\omega_{\alpha,\beta} = \tr[B^{(\alpha,\beta)}\hat{\rho}^{(\alpha,\beta)}]$, and $g(\rho) = \sup_{\Lambda \in \mathsf{LO}_{2}} \tr[\Lambda(\rho) \psi^*]$. We then have
\begin{equation}
    \begin{aligned}
        \inf_{\substack{\{\hat{\rho}^{(\alpha,\beta)},p_{\alpha,\beta},A_{x}^{(\alpha)},B_{y}^{(\beta)}\}_{\alpha,\beta} \\ \text{s.t.} \ \sum_{\alpha,\beta}p_{\alpha,\beta} \tr[B^{(\alpha,\beta)}\hat{\rho}^{(\alpha,\beta)}] \geq \omega}} \sum_{\alpha,\beta} p_{\alpha,\beta} \, g(\hat{\rho}^{(\alpha,\beta)}) = \inf_{\substack{\{\omega_{\alpha,\beta},p_{\alpha,\beta}\}_{\alpha,\beta} \\ \text{s.t.} \sum_{\alpha,\beta}p_{\alpha,\beta}\omega_{\alpha,\beta} \geq \omega}} \sum_{\alpha,\beta} p_{\alpha,\beta} \inf_{\rho \in \mathcal{B}_{2}(\omega_{\alpha,\beta})} g(\rho).
    \end{aligned}
\end{equation}
Let 
\begin{equation}
    f_{2}(\omega) = \inf_{\rho \in \mathcal{B}_{2,\omega}} g(\rho)
\end{equation}
be the qubit LO extractability, and $\tilde{f}_{2}(\omega)$ be any convex function satisfying $f_{2}(\omega) \geq \tilde{f}(\omega)$ for all $\omega$. Then
\begin{equation}
\begin{aligned}
    \inf_{\substack{\{\omega_{\alpha,\beta},p_{\alpha,\beta}\}_{\alpha,\beta} \\ \text{s.t.} \sum_{\alpha,\beta}p_{\alpha,\beta}\omega_{\alpha,\beta} \geq \omega}} \sum_{\alpha,\beta} p_{\alpha,\beta} \inf_{\rho \in \mathcal{B}_{2}(\omega_{\alpha,\beta})} g(\rho) &= \inf_{\substack{\{\omega_{\alpha,\beta},p_{\alpha,\beta}\}_{\alpha,\beta} \\ \text{s.t.} \sum_{\alpha,\beta}p_{\alpha,\beta}\omega_{\alpha,\beta} \geq \omega}} \sum_{\alpha,\beta} p_{\alpha,\beta} f_{2}(\omega_{\alpha,\beta}) \\
    &\geq \inf_{\substack{\{\omega_{\alpha,\beta},p_{\alpha,\beta}\}_{\alpha,\beta} \\ \text{s.t.} \sum_{\alpha,\beta}p_{\alpha,\beta}\omega_{\alpha,\beta} \geq \omega}} \tilde{f}\Bigg( \sum_{\alpha,\beta} p_{\alpha,\beta} \omega_{\alpha,\beta}\Bigg)  \\
    &\geq \tilde{f}(\omega),
\end{aligned}
\end{equation}
completing the proof.
\end{proof}

\subsection{Bounds on the singlet fidelity}\label{app: bounds on the singlet fidelity} 
In this section, we provide bounds on the singlet fidelity under local channels, $f_{2}(\omega)$, for the case $\psi^* = \phi^+$. We begin with the following lemmas, which allow us to reduce the problem. Consider writing
\begin{equation}
    \tr[\Lambda(\rho) \phi^{+}] = \langle \Lambda(\rho) , \phi^{+} \rangle = \langle \rho , \Lambda^{\dagger}(\phi^{+}) \rangle = \tr[\rho \Lambda^{\dagger}(\phi^{+})],
\end{equation}
where $\langle A, B\rangle = \tr[A^{\dagger}B]$ is the Hilbert Schmidt norm, and $\Lambda^{\dagger}$ is the adjoint channel of $\Lambda$. 

\begin{lemma}
    Let $\Lambda_{A}:\mathcal{S}(\mathcal{H}_{A}) \to \mathcal{S}(\mathcal{H}_{A})$ and $\Lambda_{B}:\mathcal{S}(\mathcal{H}_{B}) \to \mathcal{S}(\mathcal{H}_{B})$ be quantum channels and $\rho \in \mathcal{S}(\mathcal{H}_{A} \otimes \mathcal{H}_{B})$. Then 
    \begin{equation}
        \tr_{A}\big[ \Lambda_{A} \otimes \Lambda_{B} (\rho)] = \Lambda_{B}(\rho_{B}), \ \ \ \text{and} \ \ \ \tr_{B}\big[ \Lambda_{A} \otimes \Lambda_{B} (\rho)] = \Lambda_{A}(\rho_{A}).
    \end{equation} \label{lem:ptr}
\end{lemma}
\begin{proof}
    This fact is a consequence of the product channel strucuture. Let $\{K_{A}^{i}\}_{i}$ be a set of Kraus operators for $\Lambda_{A}$, and $\{K_{B}^{j}\}_{j}$ be a set of Kraus operators for $\Lambda_{B}$. Then we have
    \begin{equation}
    \begin{aligned}
        \tr_{A}\big[ \Lambda_{A} \otimes \Lambda_{B} (\rho)] &= \sum_{i,j}\tr_{A}\big[ (K_{A}^{i} \otimes K_{B}^{j})\rho(K_{A}^{i} \otimes K_{B}^{j})^{\dagger}] \\
        &= \sum_{j} \tr_{A}\Bigg[ \Bigg(\sum_{i}\big(K_{A}^{i}\big)^{\dagger}K_{A}^{i} \otimes K_{B}^{j}\Bigg)\rho (\id_{A} \otimes K_{B}^{j})^{\dagger}\Bigg] \\
        &= \sum_{j} K_{B}^{j} \rho_{B} \big(K_{B}^{j}\big)^{\dagger} = \Lambda_{B}(\rho_{B}),
    \end{aligned}
    \end{equation}
    where for the second equality we used that the partial trace is cyclic, and for the third we used the identities $\sum_{i}\big(K_{A}^{i}\big)^{\dagger}K_{A}^{i} = \id_{A}$ and $\tr_{A}[(\id_{A} \otimes Y_{B})X_{AB}(\id_{A} \otimes Y_{B})^{\dagger}] =Y_{B} \tr_{A}[X_{AB}]Y_{B}^{\dagger}$. The analogous statement holds when tracing out system $B$. 
\end{proof}

\noindent This allows us to show the following.

\begin{lemma}
    \chg{Let $\Lambda \in \mathsf{LO}_{2} \cap \mathsf{U}_{2}$ be a local, unital quantum channel.} Then the state $\sigma = \Lambda^{\dagger}(\phi^{+})$ is Bell diagonal in some basis, i.e., it satisfies
    \begin{equation}
        \tr_{\hat{Q}_{A}}[\sigma] = \tr_{\hat{Q}_{B}}[\sigma] = \id_{2}/2. \label{eq:bd}
    \end{equation} \label{lem:ch2BD}
\end{lemma}
\begin{proof}
    \chg{Let $\Lambda = \Lambda_{A} \otimes \Lambda_{B}$ where $\Lambda_{A}(\id_{2}) = \Lambda_{B}(\id_{2}) = \id_{2}$. This implies that both $\Lambda_{A}^{\dagger}$ and $\Lambda_{B}^{\dagger}$ are quantum channels.} We can therefore apply \Cref{lem:ptr},
    \begin{equation}
        \tr_{\hat{Q}_{A}}[\sigma] = \tr_{\hat{Q}_{A}}[\Lambda_{A}^{\dagger}\otimes \Lambda_{B}^{\dagger}(\phi^{+})] = \Lambda_{B}^{\dagger}(\id_{2}/2) = \id_{2}/2,
    \end{equation}
    where for the last line we used that $\Lambda^{\dagger}$ is unital. The analogous statement holds when tracing out system $\hat{Q}_{B}$. 
\end{proof}

\noindent We denote the set of Bell diagonal states, i.e., two-qubit  states that satisfy  \Cref{eq:bd}, $\mathcal{BD}$. We next prove the converse statement. 

\begin{lemma}
    For every state $\sigma \in \mathcal{BD}$, there exist unital channels $\Lambda_{A}:\mathcal{S}(\mbC^{2}) \to \mathcal{S}(\mbC^{2})$ and $\Lambda_{B}:\mathcal{S}(\mbC^{2}) \to \mathcal{S}(\mbC^{2})$ such that
    \begin{equation}
        \Lambda_{A}^{\dagger} \otimes \Lambda_{B}^{\dagger}(\phi^{+}) = \sigma.
    \end{equation} \label{lem:BD2ch}
\end{lemma}
\begin{proof}
    Suppose $\sigma$ is Bell diagonal, i.e., there exists a Bell basis $\{\tilde{\Phi}_{\alpha}\}_{\alpha=0}^{3}$ and a distribution $\{\lambda_{\alpha}\}_{\alpha=0}^{3}$ such that
    \begin{equation}
        \sigma = \sum_{\alpha} \lambda_{\alpha} \tilde{\Phi}_{\alpha}.
    \end{equation}
    Define the channel $\tilde{\Lambda}_{A}$ by the following Kraus operators 
    \begin{equation}
        \begin{aligned}
            E_{0} &= \sqrt{\lambda_{0}} \id, \\
            E_{1} &= \sqrt{\lambda_{1}} \sigma_{Z}\sigma_{X}, \\
            E_{2} &= \sqrt{\lambda_{2}} \sigma_{Z}, \\
            E_{3} &= \sqrt{\lambda_{3}} \sigma_{X}. 
        \end{aligned}
    \end{equation}
    Note that $\tilde{\Lambda}_{A}$ is unital, equal to its adjoint, and 
    \begin{equation}
       \sum_{i} (E_{i} \otimes \id_{2}) \phi^{+} (E_{i} \otimes \id_{2}) = \sum_{\alpha}\lambda_{\alpha} \Phi_{\alpha},
    \end{equation}
    where $\{\Phi_{\alpha}\}_{\alpha}$ is the standard Bell basis, $\Phi_{0} = \phi^{+}$, $\Phi_{1} = \psi^{-}$, $\Phi_{2} = \phi^{-}$ and $\Phi_{3} = \psi^{+}$, where $\ket{\phi^\pm} = (\ket{00} \pm \ket{11})/\sqrt{2}$ and $\ket{\psi^\pm} = (\ket{01} \pm \ket{10})/\sqrt{2}$. Let $U_{A} \otimes U_{B}$ be the local unitary which rotates the standard Bell basis $\{\tilde{\Phi}_{\alpha}\}_{\alpha}$ to the Bell basis $\{\tilde{\Phi}_{\alpha}\}_{\alpha=0}^{3}$, i.e.,
    \begin{equation}
        (U_{A} \otimes U_{B})  \Bigg(\sum_{\alpha}\lambda_{\alpha} \Phi_{\alpha} \Bigg) (U_{A} \otimes U_{B})^{\dagger} = \sum_{\alpha}\lambda_{\alpha} \tilde{\Phi}_{\alpha} = \sigma.
    \end{equation}
    Define the unital channels 
    \begin{equation}
        \Lambda_{A}(\rho) = \tilde{\Lambda}_{A}(U_{A}^{\dagger}\rho U_{A}), \ \ \ \text{and} \ \ \ \Lambda_{B}(\tau) = U_{B}^{\dagger} \tau U_{B}.
    \end{equation}
    We then have
    \begin{equation}
        \Lambda_{A}^{\dagger} \otimes \Lambda_{B}^{\dagger}(\phi^{+}) = \sum_{i} (U_{A} \otimes U_{B}) (E_{i} \otimes \id_{2})\phi^{+} (E_{i} \otimes \id_{2})(U_{A} \otimes U_{B})^{\dagger} = \sigma, 
    \end{equation}
    as desired. 
\end{proof}
By combining Lemmas \ref{lem:ch2BD} and \ref{lem:BD2ch}, we arrive at the following reduction.
\begin{corollary}
    Let $\rho \in \mathcal{S}_{2}$ be an arbitrary two-qubit state. Then following equality holds:
    \begin{equation}
       \sup_{\Lambda \in \mathsf{LO}_{2} \cap \mathsf{U}_{2}} \tr[\rho \Lambda^{\dagger}(\phi^+)] = \sup_{\sigma \in \mathcal{BD}} \tr[\rho\sigma]. \label{eq:opt3}
    \end{equation}
     \label{cor: bellDiag}
\end{corollary}
\begin{proof}
    Let
    \begin{equation}
        S = \Big\{ \Lambda^{\dagger}(\phi^+) \ : \ \Lambda \in \mathsf{LO}_{2} \cap \mathsf{U}_{2} \Big\}.
    \end{equation}
    Then by Lemma \ref{lem:ch2BD}, we know $S \subseteq \mathcal{BD}$, and by Lemma \ref{lem:BD2ch}, we know $\mathcal{BD} \subseteq S$. We therefore have $S = \mathcal{BD}$, proving the claim. 
\end{proof}

Having lower bounded the maximization over local channels by a maximization over states, we now show how, given for fixed pair of measurement angles $a,b \in \mathbb{R}$, the function $f_{2}(\omega)$ can be bounded by an SDP. 
\begin{lemma}
    Define the function
    \begin{equation}
        \begin{aligned}
            f_{a,b}(\omega) := \max_{\lambda,\mu,\sigma} \ & \lambda \, \omega + \mu \\
            \mathrm{s.t.} \ & \ \sigma - \lambda B(a,b) - \mu \id_{4} \geq 0 \\
            & \tr_{\hat{Q}_{A}}[\sigma] = \frac{\id_{2}}{2} \\
            & \tr_{\hat{Q}_{B}}[\sigma] = \frac{\id_{2}}{2} \\
            & \sigma \in \mcS_{2}, \ \lambda \geq 0, \ \mu \in \mathbb{R}.
        \end{aligned} \label{eq:dualSDP}
    \end{equation}
    Then 
    \begin{equation}
        f_{2}(\omega) \geq \min_{(a,b) \in \mcF_{\omega}} f_{a,b}(\omega),
    \end{equation}
    where $\mcF_{\omega} = \{ (a,b) \in [0,\pi/2] \times [0,\pi/2] \ : \ \exists \rho \in \mathcal{S}_{2} \ \mathrm{s.t.} \ \tr[B(a,b)\rho] \geq \omega \}$.\label{lem:SDP}
\end{lemma}
\begin{proof}
    \chg{We first employ the bound $f_{2}(\omega) = \inf_{\substack{\rho \in \mathcal{B}_{2,\omega}}} \sup_{\Lambda \in \mathsf{LO}_{2} } \tr[\Lambda(\rho)\phi] \geq \inf_{\substack{\rho \in \mathcal{B}_{2,\omega}}} \sup_{\Lambda \in \mathsf{LO}_{2} \cap \mathsf{U}_{2} } \tr[\Lambda(\rho)\phi]$.} Then, by applying \Cref{cor: bellDiag}, we have
    \begin{equation}
    \begin{aligned}
        f_{2}(\omega) &\geq \inf_{\rho \in \mathcal{B}_{2,\omega}} \sup_{\sigma \in \mathcal{BD}} \tr[\rho \sigma] \\ 
        &= \min_{\rho \in \mathcal{B}_{2,\omega}} \max_{\sigma \in \mathcal{BD}} \tr[\rho \sigma] \\
        &= \min_{(a,b) \in \mcF_{\omega}} \min_{\rho \in \mathcal{B}^{a,b}_{\omega}} \max_{\sigma \in \mathcal{BD}} \tr[\rho \sigma] \\
        &= \min_{(a,b) \in \mcF_{\omega}}  \max_{\sigma \in \mathcal{BD}} \min_{\rho \in \mathcal{B}^{a,b}_{\omega}} \tr[\rho \sigma], \label{eq:optSwap}
    \end{aligned}
    \end{equation}
    where $\mathcal{B}^{a,b}_{\omega}$ is the set of two-qubit states which can achieve a Bell value of $\omega$ with measurement angles $a,b$, i.e., 
    \begin{equation}
    \mathcal{B}^{a,b}_{\omega} = \Big\{ \rho \in \mathcal{S}_{2} \ : \ \ \tr[B(a,b)\rho] \geq \omega \Big\}. \label{eq:B_ab}
\end{equation}
    In \Cref{eq:optSwap}, we used the following facts:
    \begin{enumerate}
        \item The set $\mcB_{2,\omega}$ defined in \Cref{eq:B_set} is compact. For proof see Claim \ref{claim:1}.
        \item The set $\mcF_{\omega}$ defined in the lemma statement is compact. For proof see Claim \ref{claim:2}.
        \item The sets $\mathcal{B}^{a,b}_{\omega}$ and $\mcB \mcD$ are compact and convex. This follows from the fact that any subset of $\mcS_{2}$ defined by linear constraints inherits the compactness and convexity of $\mcS_{2}$. 
        \item For any function $g:\mcS_{2} \to \mathbb{R}$, the set $\mcB_{2,\omega}$ satisfies $\min_{\rho \in \mcB_{2,\omega}} g(\rho) = \min_{(a,b) \in \mcF_{\omega}} \min_{\rho \in \mathcal{B}^{a,b}_{\omega}} g(\rho)$. For proof see Claim \ref{claim:3}.
        \item The function $g(\rho,\sigma) = \tr[\rho \sigma]$ is linear in one of its arguments when the other is fixed.
    \end{enumerate}
    The first equality in \Cref{eq:optSwap} then follows from point 1, allowing us to replace $\inf_{\rho \in \mathcal{B}_{2,\omega}}$ with $\min_{\rho \in \mathcal{B}_{2,\omega}}$, and point 3, allowing us to replace $\sup_{\sigma \in \mathcal{BD}}$ with $\max_{\sigma \in \mathcal{BD}}$. The second equality follows from point 4. The third equality follows from points 3 and 5, allowing us to apply von Neumann's minimax theorem~\cite{Simons1995}. 
    
    Consider, for a fixed $\sigma \in \mathcal{BD}$, $\omega$ and $(a,b) \in [0,\pi/2] \times [0 ,\pi/2]$ such that $\mathcal{B}^{a,b}_{\omega}$ is non-empty, the optimization
    \begin{equation}
    \begin{aligned}
        \min \ &\tr[\rho \sigma] \\
        \text{s.t.} \ & \tr[B(a,b)\rho] \geq \omega \\
        & \tr[\rho] = 1 \\
        & \rho \geq 0.
    \end{aligned}
    \end{equation}
    This has the dual (see, e.g.,~\cite[Example 5.11]{Boyd_Vandenberghe_2004}) 
    \begin{equation}
    \begin{aligned}
        \max \ & \lambda \omega + \mu \\
        \text{s.t.} \ & \sigma - \lambda B(a,b) - \mu \id_{4} \geq 0 \\
        & \lambda \geq 0 \\
        & \mu \in \mathbb{R},
    \end{aligned} \label{eq:dual}
    \end{equation}
    whose optimal value lower bounds that of the primal by weak duality. Furthermore, we proceed to show that strong duality holds. Consider the point $(\lambda , \mu) = (1 , -\eta^{\mathrm{Q}}_{2}-\epsilon)$ for any $\epsilon > 0$, where $\eta^{\mathrm{Q}}_{2}$ is the maximum quantum value of the Bell functional $B$ for qubits, i.e., 
    \begin{equation}
        \eta^{\mathrm{Q}}_{2} = \sup_{\substack{\rho \in \mathcal{S}_{2}, \\ (a,b) \in [0,\pi/2] \times [0\pi/2]}} \tr[B(a,b)\rho].
    \end{equation}
    \chg{The constraint $\sigma - \lambda B(a,b) - \mu \id_{4} \geq 0$ is equivalent to 
    \begin{equation}
        \bra{\psi}\sigma \ket{\psi} - \lambda \bra{\psi}B(a,b)\ket{\psi} - \mu \geq 0, \ \forall \ket{\psi}.
    \end{equation}
    The point $(\lambda , \mu) = (1 , -\eta^{\mathrm{Q}}_{2} - \epsilon)$ satisfies $\lambda > 0$, and
    \begin{equation}
    \begin{aligned}
        \bra{\psi}\sigma \ket{\psi} - \lambda \bra{\psi}B(a,b)\ket{\psi} - \mu &= \bra{\psi}\sigma \ket{\psi} - \bra{\psi}B(a,b)\ket{\psi} + \eta^{\mathrm{Q}}_{2} + \epsilon \\
        &\geq \epsilon \\
        &> 0,
    \end{aligned}
    \end{equation}
    where we used the fact that $\eta^{\mathrm{Q}}_{2}  \geq \bra{\psi}B(a,b)\ket{\psi}$ for all measurements $(a,b)$ and states $\ket{\psi}$, and that $\bra{\psi}\sigma \ket{\psi} \geq 0$. We have shown that the dual is strictly feasible, and therefore strong duality holds. Inserting \Cref{eq:dual} into \Cref{eq:optSwap} establishes the claim.}
\end{proof}

\chg{We provide proofs of the referenced claims below.
\begin{claim}
    The set $\mcB_{2,\omega}$ defined in \Cref{eq:B_set} is compact. \label{claim:1}
\end{claim}

\begin{proof}
We first show $\mathcal{B}_{2,\omega}$ is closed. Let $K:=[0,\tfrac{\pi}{2}]\times[0,\tfrac{\pi}{2}]$ and note that $(a,b)\mapsto B(a,b)$ is continuous (see Claim \ref{claim: continuity of B(a,b)}). Define
\[
f(\rho):=\max_{(a,b)\in K} \operatorname{tr}[B(a,b)\rho].
\]
The maximum is achievable because $K$ is compact and $(a,b)\mapsto B(a,b)$ is continuous.  

Let $\{\rho_{n}\}_{n \in \mathbb{N}} \subset \mathcal{B}_{2,\omega}$ be a sequence of states and suppose $\rho_{n} \to \rho$ (in trace norm) as $n \to \infty$. For any $(a,b)\in K$ and any $\epsilon > 0$, there exists an $n \in \mathbb{N}$ such that
\[
\big|\tr[B(a,b)\rho]-\tr[B(a,b)\rho_n]\big| \le \|B(a,b)\|_\infty \, \|\rho-\rho_n\|_1 \le c \, \epsilon,
\]
where $c := \max_{(a,b)\in K} \|B(a,b)\|_\infty$ and for the first inequality we applied H\"older's inequality. This follows from the fact that $\| \rho - \rho_{n}\|_{1} \to 0$ as $n \to \infty$. Since $\rho_n \in \mathcal{B}_{2,\omega}$, there exist $(a_n,b_n)\in K$ such that $\tr[B(a_n,b_n)\rho_n]\ge \omega$. Applying the above bound with the substitution $(a,b) \mapsto (a_{n},b_{n})$, for any $\epsilon >0$ there exists an $n \in \mathbb{N}$ such that 
\[
\tr[B(a_n,b_n)\rho] \ge \omega - c \, \epsilon.
\]
As this holds for every $\epsilon>0$, it follows that $f(\rho) \ge \omega - \epsilon'$ for every $\epsilon' > 0$. Taking $\epsilon'$ to be arbitrarily small therefore implies $\rho \in \mathcal{B}_{2,\omega}$.  

Finally, $\mathcal{B}_{2,\omega} \subset \mathcal{S}_2$ and the state space $\mathcal{S}_2$ is bounded, so $\mathcal{B}_{2,\omega}$ is a closed and bounded subset of a finite-dimensional space. Therefore, $\mathcal{B}_{2,\omega}$ is compact.
\end{proof}

\begin{claim}
    The set $\mcF_{\omega}$ defined in \Cref{lem:SDP} is compact. \label{claim:2}
\end{claim}

\begin{proof}
We follow an argument analogous to Claim~\ref{claim:1}.  

Let $\{(a_n,b_n)\}_{n\in\mathbb{N}} \subset \mathcal{F}_\omega$ be a sequence and suppose the limit $(a_n,b_n) \to (a,b) \in [0,\pi/2]^2$ as $n \to \infty$ exists.  
By definition of $\mathcal{F}_\omega$, for each $n$ there exists $\rho_n \in \mathcal{S}_2$ such that 
\[
\tr[\rho_n B(a_n,b_n)] \ge \omega.
\]
Define
\[
\tilde{g}(a,b) := \max_{\rho \in \mathcal{S}_2} \tr[\rho B(a,b)].
\]
The maximum exists (and is achievable) because $\mathcal{S}_2$ is compact and $(a,b) \mapsto B(a,b)$ is continuous (see Claim \ref{claim: continuity of B(a,b)}).  

By Hölder's inequality,
\[
\big|\tr[\rho_n B(a_n,b_n)] - \tr[\rho_n B(a,b)]\big| \le \|B(a_n,b_n) - B(a,b)\|_1 \cdot \|\rho_n\|_\infty \le \|B(a_n,b_n)-B(a,b)\|_1.
\]
Continuity of $(a,b)\mapsto B(a,b)$ implies $\|B(a_n,b_n)-B(a,b)\|_1 \to 0$ as $n \to \infty$, so for every $\epsilon > 0$ there exists a large enough $n \in \mathbb{N}$ such that
\[
\tr[\rho_n B(a,b)] \ge \tr[\rho_n B(a_n,b_n)] - \epsilon \ge \omega - \epsilon.
\]
Since $\epsilon>0$ can be arbitrarily small, we conclude
\[
\tilde{g}(a,b) \ge \omega,
\] 
and hence $(a,b)\in \mathcal{F}_\omega$. Therefore, $\mathcal{F}_\omega$ is closed. The boundedness follows from the boundedness of the Bell operators $B(a,b)$. 
\end{proof}}

\chg{\begin{claim}\label{claim: continuity of B(a,b)}
The map $(a,b) \mapsto B(a,b)$ is continuous on $[0,\pi/2] \times [0 , \pi/2]$, i.e.,
\[
\lim_{(a',b') \to (a,b)} \|B(a',b') - B(a,b)\|_1 = 0.
\]
\end{claim}

\begin{proof}
Recall that $B(a,b)$ is constructed from tensor products of local operators $A_x(a)$ and $B_y(b)$, each acting on a qubit:
\[
B(a,b) = \sum_{x,y} \gamma_{x,y}\, A_x(a) \otimes B_y(b),
\]
where the coefficients $\gamma_{x,y}$ are fixed real constants, and $A_x(a), B_y(b)$ depend continuously on $a$ and $b$.  

For any $(a',b') \in [0,\pi/2]^2$, we have
\begin{align*}
\|B(a',b') - B(a,b)\|_1 
&= \Big\|\sum_{x,y} \gamma_{x,y} \big( A_x(a') - A_x(a) \big) \otimes B_y(b') 
      + \sum_{x,y} \gamma_{x,y} A_x(a) \otimes \big( B_y(b') - B_y(b) \big) \Big\|_1 \\
&\le \sum_{x,y} |\gamma_{x,y}| \, \| (A_x(a') - A_x(a)) \otimes B_y(b') \|_1
    + \sum_{x,y} |\gamma_{x,y}| \, \| A_x(a) \otimes (B_y(b') - B_y(b)) \|_1
\end{align*}
using the triangle inequality.

Using the fact that $\|A \otimes B\|_{1} = \|A\|_{1}\|B\|_{1}$ and $A_{x}(a)^{\dagger}A_{x}(a)=B_{y}(b)^{\dagger}B_{y}(b) = \id_{2}$ (since $A_{x}(a)$ and $B_{y}(b)$ are the observables of a rank 1 projective measurement on a qubit),
\[
\| (A_x(a') - A_x(a)) \otimes B_y(b') \|_1 \le 2\| A_x(a') - A_x(a) \|_1, \qquad
\| A_x(a) \otimes (B_y(b') - B_y(b)) \|_1 \le 2\| B_y(b') - B_y(b) \|_1.
\]
Hence,
\[
\|B(a',b') - B(a,b)\|_1 \le \sum_{x,y} 2|\gamma_{x,y}| \Big( \| A_x(a') - A_x(a) \|_1 + \| B_y(b') - B_y(b) \|_1 \Big).
\]

Finally, each $A_x(a)$ and $B_y(b)$ is continuous in trace norm (for instance,
$\|A_x(a') - A_x(a)\|_1 = \|(\cos(a')-\cos(a))\sigma_z + (-1)^x(\sin(a')-\sin(a))\sigma_x\|_1 \to 0$ as $a'\to a$, and similarly for $B_y(b)$), so the right-hand side tends to zero as $(a',b') \to (a,b)$. Therefore,
\[
\lim_{(a',b') \to (a,b)} \|B(a',b') - B(a,b)\|_1 = 0,
\]
i.e., $B(a,b)$ is continuous in trace norm.
\end{proof}}

\chg{\begin{claim}
    Let $g : \mcS_{2} \to \mathbb{R}$ be a function. The set $\mcB_{2,\omega}$ defined in \Cref{eq:B_set} satisfies $\min_{\rho \in \mcB_{2,\omega}} g(\rho) = \min_{(a,b) \in \mcF_{\omega}} \min_{\rho \in \mathcal{B}^{a,b}_{\omega}} g(\rho)$, where $\mathcal{B}^{a,b}_{\omega}$ and $\mcF_{\omega}$ are defined in \Cref{eq:B_ab} and \Cref{lem:SDP}, respectively. \label{claim:3}
\end{claim}
\begin{proof}
    Since $\mcB_{2,\omega}$ is compact by Claim \ref{claim:1}, let $\rho^* \in \mcB_{2,\omega}$ denote the optimal state that satisfies $\min_{\rho \in \mcB_{2,\omega}} g(\rho) = g(\rho^*)$. By definition of $\mcB_{2,\omega}$, there exists a pair of angles $(a^*,b^*)$ satisfying $\tr[\rho^* B(a^*,b^*)] \geq \omega$. Hence $\rho^* \in \mcB_{\omega}^{a^*,b^*}$ by the definition of $\mcB_{\omega}^{a^*,b^*}$ in \Cref{eq:B_ab}. We therefore have $g(\rho^*) \geq \min_{\rho \in \mcB_{\omega}^{a^*,b^*}} g(\rho) \geq \inf_{(a,b) \in \mcF_{\omega}} \min_{\rho \in \mcB_{\omega}^{a,b}} g(\rho)$, where the second inequality follows from the fact that $(a^*,b^*) \in \mcF_{\omega}$ as defined in \Cref{lem:SDP}. By Claim \ref{claim:2} $\inf_{(a,b) \in \mcF_{\omega}} \min_{\rho \in \mcB_{\omega}^{a^*,b^*}} g(\rho) = \min_{(a,b) \in \mcF_{\omega}} \min_{\rho \in \mcB_{\omega}^{a^*,b^*}} g(\rho)$, establishing the lower bound. 
    
    For the upper bound, note that the minimization $\min_{(a,b) \in \mcF_{\omega}} \min_{\rho \in \mcB_{\omega}^{a^*,b^*}} g(\rho)$ is attained by some state $\tilde{\rho}$ and pair of angles $(\tilde{a},\tilde{b})$ such that $\tr[\tilde{\rho}B(\tilde{a},\tilde{b})] \geq \omega$. Hence $\tilde{\rho} \in \mcB_{2,\omega}$, and $\min_{\rho \in \mcB_{2,\omega}} g(\rho) \leq g(\tilde{\rho}) = \min_{(a,b) \in \mcF_{\omega}} \min_{\rho \in \mcB_{\omega}^{a^*,b^*}} g(\rho)$, completing the proof.   
\end{proof}

\begin{remark}
    The dual SDP in \Cref{eq:dualSDP} is related to the ``self-testing via operator inequalities'' (STOI) approach first introduced in Ref.~\cite{kaniewski2016analytic} (see also Refs.~\cite{Coopmans19,Baccari20,Murta23}). The STOI approach seeks to lower bound the fidelity under local channels, $F(\Lambda_{A} \otimes \Lambda_{B}(\rho),\phi^+) = \tr[\rho (\Lambda_{A} \otimes \Lambda_{B})^{\dagger}(\phi^+)]$, by establishing an operator inequality of the form $ K \geq \lambda B(a,b) + \mu \id_{4}$, where $K = (\Lambda_{A} \otimes \Lambda_{B})^{\dagger}(\phi^+)$, $\lambda \geq 0$, $\mu \in \mathbb{R}$ and the inequality holds for all $a,b$. This has the same form as the matrix positivity constraint in \Cref{eq:dualSDP}. Indeed, for the case of the CHSH inequality, the choice of dephasing channels from~\cite{kaniewski2016analytic}, which we denote $\Lambda_{A} = \Lambda_{A}^*(a), \ \Lambda_{B} = \Lambda_{B}^*(b)$, along with the values $\lambda = \lambda^* = (4 + 5\sqrt{2})/ 16$, $\mu = \mu^* = -(1+2\sqrt{2})/4$, correspond to a feasible point $(\lambda,\mu,\sigma) = (\lambda^*,\mu^*,K(a,b))$ of the SDP in \Cref{eq:dualSDP}, where $K(a,b) = (\Lambda_{A}^*(a) \otimes \Lambda_{B}^*(b))^{\dagger}(\phi^+)$. \label{rem:STOI}
\end{remark}}

\subsection{Solving the outer minimization via gridding}\label{app: gridding}
In the previous subsections, we showed that bounding the LOCC fidelity can be reduced to finding a lower bound on $\inf_{(a,b) \in \mcF_{\omega}} f_{a,b}(\omega)$. In this subsection, we develop a generic approach based on evaluating the function $f_{a,b}(\omega)$ over a finite grid (see also \cite{Bhavsar2023, sharma2024enhancing} for an application of this technique to generic optimization problems). Let $\mcI = \{0,1,...,|\mcI|-1\}$ and
\begin{equation}
    \mcG = \{(a_{j},b_{j})\}_{j \in \mathcal{I}} \subset [0,\pi/2] \times [0,\pi/2]
\end{equation}
be a finite grid over $[0,\pi/2] \times [0,\pi/2]$ with a spacing $|a_{i} - a_{i+1}| = |b_{i} - b_{i+1}| = \delta > 0$, containing $|\mathcal{I}|$ points. Consider the region inside a square of length $2\delta$ centered on the point $(a_{j},b_{j}) \in \tilde{G}(\omega)$:
\begin{equation}
    S_{j} :=  \Big\{ \big(a_{j} + \delta(1-2\lambda_{0}), b_{j} + \delta(1-2\lambda_{1})\big) \ : \lambda_{0},\lambda_{1} \in [0,1] \Big\}. \label{eq:sj}
\end{equation}
We will now upper bound the maximum value of the function $f_{a,b}(\omega)$ for $(a,b) \in S_{j}$ when the Bell operator $B(a,b)$ admits a first order expansion.
\begin{lemma}
Suppose, for all $(a,b),(a^*,b^*) \in \mcF_{\omega}$, 
\begin{equation}
    \tr[B(a,b)\rho] \leq \tr[B(a^{*},b^{*})\rho]  + c_{0}(a - a^{*}) + c_{1}(b - b^{*}) \label{eq:Op_ub}
\end{equation}
for some constants $c_{0},c_{1} \geq 0$. Then for any $(a',b') \in \mcF_{\omega} \cap S_{j}$,
\begin{equation}
    \mathcal{B}_{a',b'}^\omega \subset \tilde{\mathcal{B}}_{j}^\omega := \Big \{ \rho \in \mathcal{S}_{2} \ : \ \tr[B(a_{j},b_{j})\rho] \geq \omega - \delta (c_{0} +  c_{1}) \Big\} .
\end{equation} \label{lem:grid}
\end{lemma}
\begin{proof}
    Take any $\rho \in \mathcal{B}_{a',b'}^\omega$ for $(a',b') \in \mcF_{\omega} \cap S_{j}$ (note $(a',b') \in \mcF_{\omega}$ implies $\mathcal{B}_{a',b'}^\omega$ is non-empty). Then we know $\tr[B(a',b')\rho] \geq \omega$, and there exists constants $\lambda_{0},\lambda_{1} \in [0,1]$ such that
    \begin{equation}
        a' = a_{j} + \delta(1-2\lambda_{0}) \ \ \text{and} \ \ b' = b_{j} + \delta(1-2\lambda_{1}).
    \end{equation}
    We then have by \Cref{eq:Op_ub}, choosing $a^{*} = a_{j}$ and $b^{*} = b_{j}$,
    \begin{equation}
        \omega \leq \tr[B(a',b')\rho] \leq \tr[B(a_{j},b_{j})\rho]  + c_{0}(a' - a_{j}) + c_{1}(b' - b_{j}) = \tr[B(a_{j},b_{j})\rho] + \delta( 1 - 2\lambda_{0}) c_{0} + \delta( 1 - 2\lambda_{1}) c_{1}.
    \end{equation}
    Rearranging, we see
    \begin{equation}
        \tr[B(a_{j},b_{j})\rho] \geq \omega - \delta( 1 - 2\lambda_{0}) c_{0} - \delta( 1 - 2\lambda_{1}) c_{1} \geq \omega - \delta (c_{0} +  c_{1}),
    \end{equation}
    where for the second inequality we took $\lambda_{0},\lambda_{1} \geq 0$. As a result, $\rho \in \mcB_{j}(\omega)$ as desired.
\end{proof}
\noindent As a consequence, for all $(a',b') \in S_{j}$ and any function $f : \mathcal{S}_{2} \to \mathbb{R}$,
\begin{equation}
    \min_{(a,b) \in \mathcal{F}} \min_{\rho \in \mathcal{B}^{a,b}_{\omega}} f(\rho) = \min_{j \in \mathcal{I}} \min_{(a',b') \in S_{j}} \min_{\rho \in \mathcal{B}_{a',b'}^\omega} f(\rho) \geq \min_{j \in \mathcal{I}} \min_{\rho \in \tilde{\mathcal{B}}_{j}^\omega} f(\rho).
\end{equation}

To apply \Cref{lem:grid} we require a linear approximation to the Bell operator $B(a,b)$. 

\begin{lemma}
    Let $A_{x} = \cos(a) \, \sigma_{Z} + (-1)^{x} \sin(a) \, \sigma_{X}$ and $B_{y} = \cos(b) \, \sigma_{Z} + (-1)^{y} \sin(b) \, \sigma_{X}$ be qubit observables parameterized by $(a,b) \in [0,\pi/2] \times [0,\pi/2]$. Let $B(a,b)$ be an arbitrary Bell operator
    \begin{equation}
        B(a,b) = \sum_{x \in \{0,1\}} c_{x}^{A} (A_{x} \otimes \id_{\hat{Q}_{B}}) + \sum_{y \in \{0,1\}} c_{y}^{B} (\id_{\hat{Q}_{A}} \otimes B_{y}) + \sum_{x,y \in \{0,1\}} c_{x,y}^{AB} (A_{x} \otimes B_{y}),
    \end{equation}
    with coefficients $c_{x}^{A},c_{y}^{B},c_{xy}^{AB} \in \mathbb{R}$. Then for all $\rho \in \mathcal{S}_{2}$ and $(a',b'),(a,b) \in [0,\pi/2] \times [0,\pi/2]$,
    \begin{equation}
        \tr[B(a',b')\rho] \leq \tr[B(a,b) \rho] + c_{0}(a'-a) + c_{1}(b'-b), \label{eq:bapprox}
    \end{equation}
    where
    \begin{equation}
    \begin{aligned}
        c_{0} &= \sum_{x \in \{0,1\}} |c_{x}^{A} | + \sum_{x,y\in \{0,1\}} |c_{x,y}^{AB} |, \ \text{and}\\
        c_{1} &= \sum_{y \in \{0,1\}} |c_{y}^{B} | + \sum_{x,y \in \{0,1\}} |c_{x,y}^{AB} |.
    \end{aligned}
    \end{equation} \label{lem:bellAppprox}
\end{lemma}
\begin{proof}
Let $O(\theta) = \cos(\theta)\, \sigma_{Z} + \sin(\theta) \, \sigma_{X}$ be an arbitrary qubit observable in the ZX plane of the Bloch sphere, parameterized by an angle $\theta \in [-\pi,\pi]$. Let $\delta \in [-\pi/2,\pi/2]$. By direct calculation,
\begin{equation}
    \| O(\theta) - O(\theta + \delta) \|_{\infty} = 2 |\sin(\delta/2)| \leq |\delta|, \label{eq:OpNormA}
\end{equation}
where $\| A \|_{\infty}$ is the norm of an operator $A$ on a Hilbert space $\mcH$, equal to its largest eigenvalue. Similarly, for $\delta_{0},\delta_{1} \in [-\pi/2,\pi/2]$,
\begin{equation}
\begin{aligned}
    \| O(\theta) \otimes O(\phi) - O(\theta + \delta_{0}) \otimes O(\phi + \delta_{1}) \|_{\infty} &= \sqrt{2 | \sin(\delta_{0})\sin(\delta_{1})| + 2(1-\cos(\delta_{0})\cos(\delta_{1}))} \\
    &= 2|\sin\Big( \frac{\delta_{0} \pm \delta_{1}}{2} \Big) |\leq |\delta_{0} + \delta_{1}|, \label{eq:OpNormB}
\end{aligned}
\end{equation}
where we used the fact that $|\sin(\delta_{0})\sin(\delta_{1})| = \pm\sin(\delta_{0})\sin(\delta_{1})$ depending on the values of $\delta_{0},\delta_{1}$, and applied relevant trigonometric identities. The Bell operator $B(a,b)$ is given by
\begin{eqnarray}
    B(a,b) = \sum_{x}c^{A}_{x} ( O[(-1)^{x}a] \otimes \id) + \sum_{y}c^{B}_{y} ( \id \otimes O[(-1)^{y}b])  + \sum_{x,y}c^{AB}_{x,y}    (O[(-1)^{x}a] \otimes O[(-1)^{y}b]). 
\end{eqnarray}
Let us define, for $\delta_{0},\delta_{1} > 0$,
\begin{eqnarray}
    \Delta := B(a,b) - B(a + \delta_{0},b + \delta_{1}).
\end{eqnarray}
Note for any $\rho \in \mathcal{S}_{2}$ with a spectral decomposition $\rho = \sum_{i}p_{i} \ketbra{\phi_{i}}{\phi_{i}}$,
\begin{eqnarray}
    | \tr[\Delta \rho] | \leq \sum_{i}p_{i} |\bra{\phi_{i}}\Delta \ket{\phi_{i}}| \leq \| \Delta \|_{\infty}. 
\end{eqnarray}
We now bound the operator norm of $\Delta$ by repeatedly applying the triangle inequality, to find
\begin{equation}
\begin{aligned}
    \| \Delta \|_{\infty} &\leq \sum_{x} |c_{x}^{A}| \cdot \Big \| (O[(-1)^{x}a] - O[(-1)^{x}(a+\delta_{0})])\otimes \id \Big  \|_{\infty} + \sum_{y} |c_{y}^{B}| \cdot  \Big \| \id \otimes (O[(-1)^{y}b] - O[(-1)^{y}(b+\delta_{1})]) \Big \|_{\infty} \\
    & \hspace{2cm} + \sum_{x,y} |c_{x,y}^{AB} | \cdot \Big \| O[(-1)^{x}a] \otimes O[(-1)^{y}b] - O[(-1)^{x}(a+\delta_{0})] \otimes O[(-1)^{y} (b+\delta_{1})] \Big \|_{\infty} \\
    & \leq \delta_{0} \sum_{x} |c_{x}^{A} | + \delta_{1} \sum_{y} |c_{y}^{B} | + (\delta_{0} + \delta_{1}) \sum_{x,y} |c_{x,y}^{AB}| \\
    &=: c_{0} \, \delta_{0} + c_{1} \, \delta_{1}.
\end{aligned}
\end{equation}
For the second equality, we applied \Cref{eq:OpNormA,eq:OpNormB}. 

Therefore, we have for all $\rho \in \mathcal{S}_{2}$ and all $(a,b) \in [0,\pi/2] \times [0,\pi/2]$,
\begin{eqnarray}
    |\tr[\Delta \rho]| =\Big | \tr[B(a,b)\rho] - \tr[B(a + \delta_{0},b+\delta_{1})\rho] \Big| \leq c_{0} \, \delta_{0} + c_{1} \, \delta_{1},
\end{eqnarray}
which implies
\begin{eqnarray}
    \tr[B(a+ \delta_{0},b+\delta_{1}) \leq \tr[B(a,b)\rho] + c_{0} \, \delta_{0} + c_{1} \, \delta_{1}.
\end{eqnarray}
Let $a' = a + \delta_{0}$ and $b' = b + \delta_{1}$. Then
\begin{eqnarray}
    \tr[B(a',b')\rho] \leq \tr[B(a,b)\rho] + c_{0} \, (a' - a) + c_{1} \, (b' - b),
\end{eqnarray}
proving the claim.
\end{proof}

The approximation given by \Cref{lem:bellAppprox} is exactly of the form required by Lemma \ref{lem:grid}. By a modification to the proof of Lemma \ref{lem:SDP}, we arrive at the following consequence.

\begin{corollary}
   Let $\mathcal{G} = \{(a_{j},b_{j}) \ : \ j \in \mathcal{I}\}$ be a finite grid over $[0,\pi/2] \times [0,\pi/2]$, with a spacing $\delta > 0$. Define $\tilde{\mathcal{I}}_{\omega} = \{ j \in \mcI \ : \ (a_{j},b_{j}) \in \mcF_{\omega}\}$. Let $S_{j}$ be defined in \Cref{eq:sj}, $B(a,b),c_{0},c_{1}$ be defined in Lemma \ref{lem:bellAppprox}, and $f_{a,b}(\omega)$ be defined in Lemma \ref{lem:SDP}. Suppose $\mathcal{G}$ satisfies the following property:
    \begin{equation}
        \mcF_{\omega} \subset \bigcap_{j \in \tilde{\mathcal{I}}_{\omega}} S_{j}. \label{eq:prop1}
    \end{equation}
    Then 
    \begin{equation}
        f_{2}(\omega) \geq \min_{j \in \tilde{\mathcal{I}}_{\omega}} f_{a_{j},b_{j}}\big(\omega - \delta(c_{0} + c_{1})\big).
    \end{equation}
    \label{cor:SDP}
\end{corollary}
\begin{proof}
    We start by writing 
    \begin{equation}
        \begin{aligned}
            f_{2}(\omega) &\geq \inf_{\rho \in \mathcal{B}_{2,\omega}} \sup_{\sigma \in \mathcal{BD}} \tr[\rho \sigma] \\ 
            &= \min_{(a,b) \in \mcF_{\omega}} \min_{\rho \in \mathcal{B}^{a,b}_{\omega}} \max_{\sigma \in \mathcal{BD}} \tr[\rho \sigma],
        \end{aligned}
    \end{equation}
    as argued in \Cref{eq:optSwap}. By assumption, the grid $\mathcal{G}$ satisfies 
    \begin{equation}
        \mcF_{\omega} \subset \bigcap_{j \in \tilde{\mathcal{I}}_{\omega}} S_{j}, 
    \end{equation}
    and $(a_{j},b_{j}) \in \mcF_{\omega}$ for all $j \in \tilde{\mathcal{I}}_{\omega}$. Then for every $(a,b) \in \mcF_{\omega}$ there exists an index $j(a,b) \in \tilde{\mathcal{I}}_{\omega}$ such that $(a,b) \in S_{j(a,b)}$. Since $\mathcal{B}_{a,b}^\omega \subset \tilde{\mathcal{B}}_{j(a,b)}^\omega$ by \Cref{lem:grid} with the Bell operator approximation from \Cref{lem:bellAppprox},
    \begin{equation}
    \begin{aligned}
        \min_{(a,b) \in \mcF_{\omega}} \min_{\rho \in \mathcal{B}^{a,b}_{\omega}} \max_{\sigma \in \mathcal{BD}} \tr[\rho \sigma] &\geq  \min_{(a,b) \in \mcF_{\omega}} \min_{\rho \in \tilde{\mathcal{B}}_{j(a,b)}^\omega} \max_{\sigma \in \mathcal{BD}} \tr[\rho \sigma] \\
        &=  \min_{j \in \tilde{\mathcal{I}}} \min_{\rho \in \tilde{\mathcal{B}}_{j}^\omega} \max_{\sigma \in \mathcal{BD}} \tr[\rho \sigma].
    \end{aligned}
    \end{equation}
    The remainder of the proof follows identically to that of \Cref{lem:SDP}, except with the constraint $\tr[B(a,b)\rho] \geq \omega$ substituted with $\tr[B(a,b)\rho] \geq \omega - \delta(c_{0} + c_{1})$. 

\end{proof}

\subsection{Example: CHSH}
To illustrate the numerical method derived in this section, we applied \Cref{cor:SDP} to bound the singlet fidelity certified by the CHSH inequality. This is known analytically to equal~\cite{bardyn2009device}
\begin{equation}
    f_{2}(\omega) \geq \frac{1 + \sqrt{(\omega/2)^{2}-1}}{2}, 
\end{equation}
where $\omega \in [2,2\sqrt{2}]$ is the violation of the CHSH inequality in correlator form. In \Cref{fig: combined_singFidCHSH} we compare this analytical bound to the lower bound generated by our numerics at different grid spacings $\delta$. One can see that the bounds can be made tighter as the spacing decreases, which results in a larger computation time. Specifically, the computation time scales as $O(\frac{1}{\delta^2})$. 

\begin{remark}
    While we have compared our technique to known analytical results for the CHSH inequality, unlike Ref.~\cite{bardyn2009device}, our approach allows one to bound the singlet fidelity for any Bell inequality, or combination of Bell functionals, in the minimal Bell scenario.
\end{remark}

\begin{figure}[h!]
    \centering
    \begin{subfigure}[b]{0.45\textwidth}
        \centering
        \includegraphics[width=\textwidth]{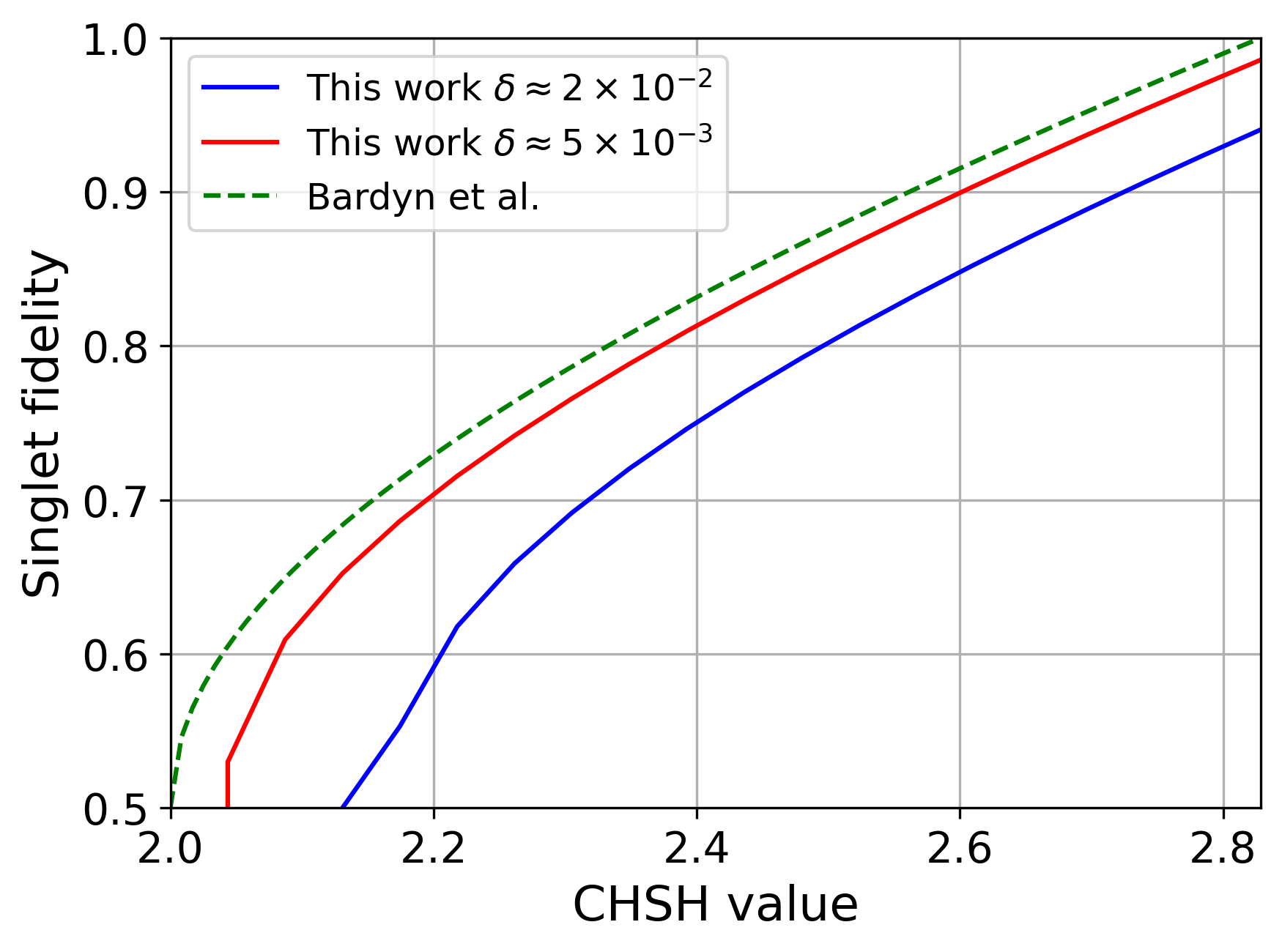}
        \label{fig: singFidCHSH_1}
    \end{subfigure}
    \hfill
    \begin{subfigure}[b]{0.45\textwidth}
        \centering
        \includegraphics[width=\textwidth]{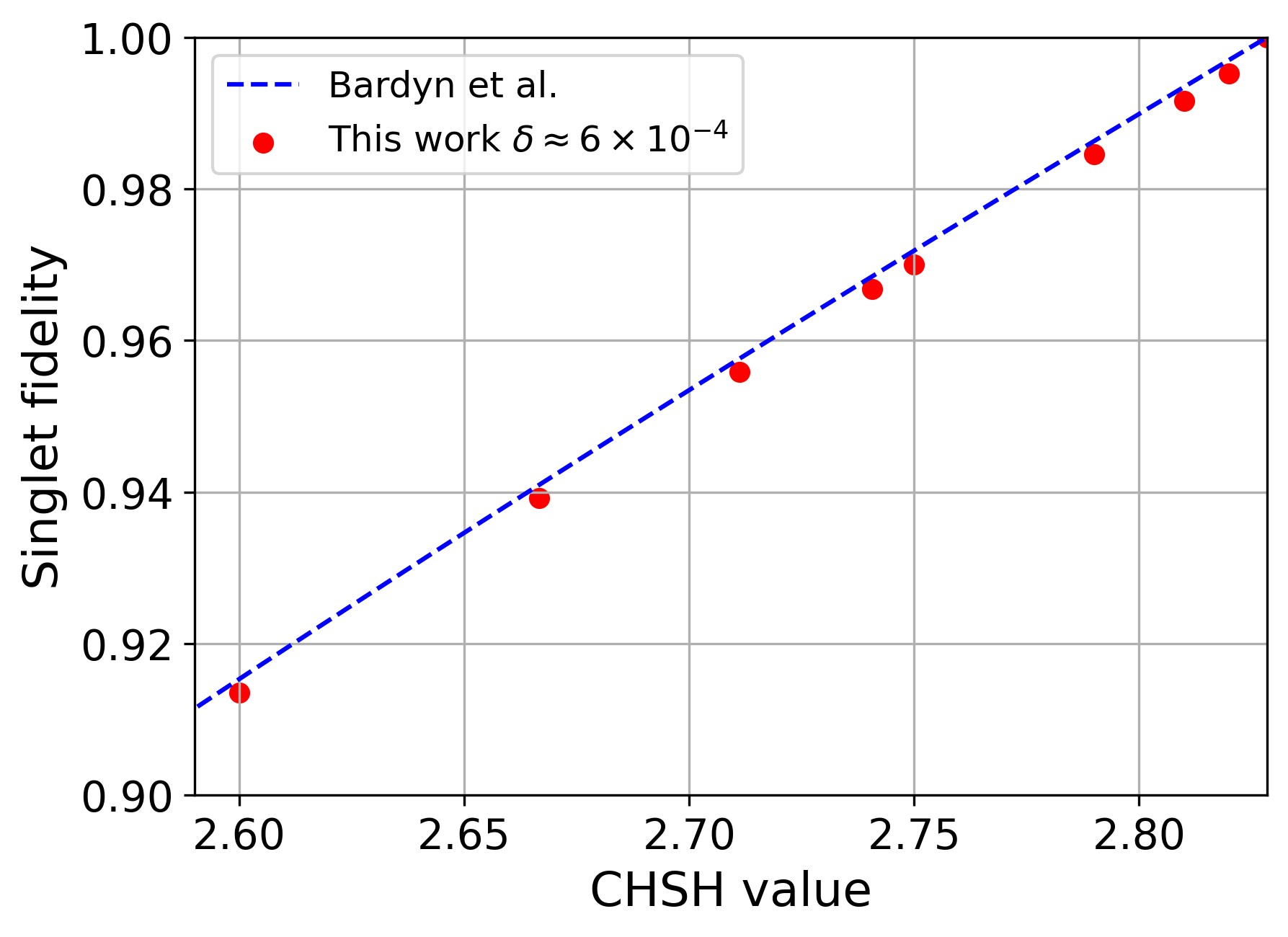}
        \label{fig: singFidCHSH_2}
    \end{subfigure}
    \caption{Lower bounds on singlet fidelity $f_{2}(\omega)$ for a given CHSH violation. Compared are the bounds generated by the numerical method of this work at different grid spacings $\delta$, and the known analytical result of Ref.~\cite{bardyn2009device}. }
    \label{fig: combined_singFidCHSH}
\end{figure}

\end{document}